\documentclass[preprint,12pt]{elsarticle}
\biboptions{sort&compress}
\usepackage[margin=1in]{geometry}

\usepackage{amssymb}
\usepackage{amsmath}
\usepackage{amsthm}
\usepackage{algorithm}
\usepackage{algpseudocode}
\usepackage{placeins}
\usepackage{xcolor}
\newtheorem{theorem}{Theorem}

\theoremstyle{definition}
\newtheorem{definition}{Definition}
\theoremstyle{remark}
\newtheorem{remark}{Remark}

\usepackage{lineno}

\usepackage{hyperref}
\hypersetup{
  pdfauthor={Juan Diego Toscano; Sascha Ranftl; Zhiwei Gao; Li-Lian Wang; George Em Karniadakis}
}

\journal{Journal of Computational Physics}

\begin{document}

\begin{frontmatter}



\title{Tractable Inference with Kolmogorov-Arnold Networks}


\author[inst1]{Juan Diego Toscano\fnref{eq}}
\author[inst1]{Sascha Ranftl\corref{cor1}\fnref{eq,current-address}}
\author[inst1]{Zhiwei Gao}
\author[inst3]{Li-Lian Wang}
\author[inst1]{George Em Karniadakis\corref{cor1}}

\affiliation[inst1]{organization={Division of Applied Mathematics, Brown University},
            city={Providence},
            postcode={02912},
            state={RI},
            country={USA}}
\affiliation[inst3]{organization={Division of Mathematical Sciences, Nanyang Technological University},
            postcode={637371},
            country={Singapore}}

\cortext[cor1]{Corresponding authors: sranftl@purdue.edu, george\_karniadakis@brown.edu}
\fntext[eq]{These authors contributed equally.}
\fntext[current-address]{Current address: Purdue University, 585 Purdue Mall, West Lafayette, IN 47907, USA}

\begin{abstract}
Machine-learning surrogates in computational science often require not only accurate predictions but also expectations, variances, conserved integrals, and normalization constants. Conventional neural networks, including Kolmogorov--Arnold networks (KANs), generally evaluate these quantities through Monte Carlo sampling or multidimensional quadrature, making repeated inference costly, particularly in high dimensions. We introduce KANs with Tractable Operators for Uncertainty and Constraints (TOUCKANs), whose finite-degree polynomial outer functions yield a decomposable sum--product representation. Under factorizing input measures, expectations and higher moments then reduce algebraically to deterministic recursions over one-dimensional leaf integrals. Independent and shared inner-block variants form a parameter-efficiency hierarchy, and a universal-approximation proof for its most restricted member extends to the complete class. The same deterministic operators allow conserved integrals and normalization to be imposed throughout training. Across uniform and Gaussian benchmarks through 50 dimensions, deterministic mean and variance errors remained below the median fixed-budget Monte Carlo errors in every case. TOUCKANs also conserved Cahn--Hilliard mass to machine precision and imposed normalization and orthogonality through deterministic moments, enabling joint computation of the lowest four Schr\"odinger eigenpairs in two, three, and five dimensions. TOUCKANs therefore connect KAN approximation with tractable probabilistic circuits, making statistical operators and integral constraints directly available within an expressive neural surrogate.
\end{abstract}

%


\end{frontmatter}



\section{Introduction}
\label{sec:introduction}
Kolmogorov--Arnold Networks (KANs)~\cite{liu2025kan} are neural architectures motivated by the Kolmogorov--Arnold representation theorem (KART)~\cite{kolmogorov1963representation}, which represents continuous multivariate functions through compositions of univariate functions and addition. Rather than assigning fixed nonlinearities to nodes, KANs place learnable univariate functions on the edges of the computational graph. Since their introduction, KANs have expanded into architectures based on alternative basis functions, deep compositions, and low-rank parameterizations~\cite{guilhoto2025deep,chao2026rbf,bozorgasl2024wav,Gao2025}, while recent studies have established broader approximation conditions and surveyed the growing family~\cite{ismailov2026necessary,faroughi2026kolmogorov,noorizadegan2026practitioner}. These developments have supported applications in physics-informed machine learning~\cite{liu2025kan,shukla2024comprehensive,jacob2025spikans}, operator learning~\cite{ABUEIDDA2025,lee2025kano}, and uncertainty quantification~\cite{ZhiweiGao2025,giroux2025uncertainty,mollaali2025conformalizedkans}.

At the same time, many tasks in computational science and engineering require not only accurate function approximation but also efficient evaluation of statistical or integral quantities of the learned surrogate. These quantities include expectations, moments, variances, sensitivity indices, conserved integrals, and normalization constants~\cite{smith2013,saltelli2008,sobol2001global,sudret2008Global}. Conventional neural networks, including KANs, generally evaluate them through Monte Carlo sampling or multidimensional quadrature~\cite{gal2016dropout,lakshminarayanan2017simple,wilson2020bayesian,abdar2021review,psaros2022Uncertainty}. The required evaluations become costly when the dimension is high or when the same functional must be recomputed throughout training or inference.

Low-dimensional factorization provides one route around this bottleneck. Separable and tensorized neural networks combine low-dimensional factors, while tractable probabilistic circuits arrange them through recursive sums and decomposable products~\cite{wang2022tensor,cho2024separable,choi2020Probabilistic}. Recent KAN-based separable models use related structure to reduce physics-informed training cost~\cite{jacob2025spikans}, whereas low-rank KAN variants target parameter-efficient adaptation~\cite{Gao2025}. These methods address valuable but distinct objectives and do not make statistical operators tractable through the KAN architecture itself.

To address this gap, we introduce KANs with Tractable Operators for Uncertainty and Constraints (TOUCKANs). The key idea is to restrict the KAN representation so that it acquires the decomposable multiplicative--additive structure of tractable probabilistic circuits, including Deep Polynomial Chaos Expansions~\cite{exenberger2025deep} and Sum--Product Networks~\cite{poon2012SumProduct,peharz2015theoretical}. Specifically, applying finite-degree polynomials to sums of univariate functions allows the multinomial theorem to rewrite the network as a decomposable sum--product representation. Under factorizing input measures, high-dimensional moments then reduce to one-dimensional leaf moments assembled through deterministic recursions. The same quantities can be incorporated into the architecture to enforce conserved integrals and normalization. Tractable inference therefore follows from the architecture of the learned surrogate rather than from a separate sampling procedure.

We realize this structure within the two-block K\r{u}rkov\'a--Kolmogorov--Arnold network (KKAN) formulation~\cite{toscano2025kurkovakans}. The resulting inner-block hierarchy ranges from independent per-dimension networks to shared formulations motivated by the Lorentz and Laczkovich representation theorems~\cite{lorentz1962metric,laczkovich2021superposition}. The shared formulations reduce the parameter footprint without changing the moment recursions, while preserving universal approximation throughout the hierarchy. Their smaller footprint also makes strong second-order optimizers such as SSBroyden computationally feasible~\cite{urban2025unveiling,jnini2026curvature}. Such optimizers have enabled physics-informed networks to attain near machine-precision accuracy~\cite{toscano2025athena,toscano2026graft}, but their computational cost generally restricts them to models with relatively few parameters.

The main contributions of this work are as follows:
\begin{enumerate}
    \item We show that finite-degree polynomial outer functions induce decomposable sum--product representations in KANs and derive deterministic recursions for expectations and variances from one-dimensional leaf moments. The corresponding theorem extends the construction to arbitrary depth and polynomial observables.
\smallskip
    \item We construct independent, full-head, diagonal, and rank-one inner parameterizations with progressively smaller parameter footprints. By proving density of the most restricted rank-one class, we establish universal approximation for the complete hierarchy.
\smallskip
    \item We use the deterministic integral evaluations to impose conserved linear functionals and normalization by construction. These constraints hold for every parameter value under the quadrature-discretized functional and remain differentiable during training.
\smallskip
    \item Across uniform and Gaussian benchmarks through $D=50$, deterministic TOUCKAN moment errors remain below the median fixed-budget Monte Carlo errors. The same framework conserves mass to machine precision for the Cahn--Hilliard equation and imposes normalization and orthogonality through deterministic moments, enabling joint computation of the lowest four Schr\"odinger eigenpairs in $D=2$, $3$, and $5$.
\end{enumerate}

Together, these results connect KAN approximation with tractable probabilistic circuits, making statistical operators and integral constraints directly available during inference and training.

\section{From Kolmogorov--Arnold superposition to tractable inference}
\label{sec:tractability-principle}

The Kolmogorov--Arnold representation theorem (KART) states that for any continuous function $f : [0,1]^D \rightarrow \mathbb{R}$, there exist continuous univariate functions $\phi_{q,d} : \mathbb{R} \rightarrow \mathbb{R}$ and $\Phi_q : \mathbb{R} \rightarrow \mathbb{R}$ such that~\cite{kolmogorov1963representation}
\begin{align} \label{eq:def-KAST}
f(\mathbf{x})
=
\sum_{q=1}^{2D+1}
\Phi_q
\left(
\sum_{d=1}^{D}
\phi_{q,d}(x_d)
\right).
\end{align}
Here, $\phi_{q,d}$ are the univariate inner functions, and $\Phi_q$ are the outer functions. Although KART provides an exact representation, the inner functions arising from its constructive proofs can be highly non-smooth and are therefore unsuitable as trainable components~\cite{Girosi1989}. To address this problem, K\r{u}rkov\'a reformulated the theorem as an approximation result in which the inner and outer functions are represented by finite sums of sigmoidal units, while preserving the two-block composition~\cite{kurkova1992kolmogorov}. KKANs generalized this construction by allowing the inner and outer univariate functions to be approximated using broader learnable function classes, including neural networks and basis expansions~\cite{toscano2025kurkovakans}. KART also inspired modern multilayer architectures such as KANs and their variants~\cite{liu2025kan,guilhoto2025deep,chao2026rbf,bozorgasl2024wav,jacob2025spikans,Gao2025,ABUEIDDA2025,lee2025kano,ZhiweiGao2025,giroux2025uncertainty,mollaali2025conformalizedkans,faroughi2026kolmogorov,noorizadegan2026practitioner}; \ref{app:background-kan} details their relation to the two-block KKAN construction used here.

These approximation results address representation accuracy, but scientific applications often also require expectations, moments, variances, and integral constraints. These quantities can be expressed through observables of the form
\begin{align}
\mathbb{E}\big[ h(f(\mathbf{x})) \big]
=
\int h(f(\mathbf{x}))\, \rho(\mathbf{x}) \, \mathrm{d}V_{\mathbf{x}},
\label{eq:marginalization}
\end{align}
where $\rho(\mathbf{x})$ is the input density and $h$ specifies the quantity of interest; for example, $h(z)=z$ gives the first moment and $h(z)=z^2$ gives the second moment. General neural representations typically evaluate such quantities through sampling or multidimensional quadrature. The cost becomes substantial when the dimension is high or the functional must be recomputed during training, inference, or uncertainty propagation.

One established route to reducing this cost is to construct the representation from low-dimensional factors. Separable and tensorized neural networks combine products of univariate or low-dimensional functions through summation~\cite{wang2022tensor,cho2024separable},
\begin{align} \label{eq:separable-form}
f(\mathbf{x})
=
\sum_{r} \prod_{d=1}^{D} g_{r,d}(x_d).
\end{align}
Tractable probabilistic circuits generalize this organization by arranging low-dimensional leaves through recursive sums and products~\cite{poon2012SumProduct,choi2020Probabilistic}. Their central structural condition is decomposability: the children of every product node depend on pairwise disjoint variable scopes. Under a factorizing input density, each decomposable product integrates as a product of lower-dimensional integrals, reducing the full integral to low-dimensional computations~\cite{peharz2015theoretical}. The corresponding circuit definitions and structural conditions are provided in \ref{app:background-circuits}.

The same reduction can be extended to K\r{u}rkov\'a and KKAN variants. Since their inner approximants remain univariate, decomposability can be enforced through the outer functions. Consider the quadratic choice $\Phi_q(t)=t^2$. Substitution into \eqref{eq:def-KAST} gives
\begin{align}
f(\mathbf{x})
&=
\sum_{q=1}^{2D+1}
\left(
\sum_{d=1}^{D}\phi_{q,d}(x_d)
\right)^2 \nonumber\\
&=
\sum_{q=1}^{2D+1}
\left[
\sum_{d=1}^{D}\phi_{q,d}^{2}(x_d)
+
2\sum_{d<d'}
\phi_{q,d}(x_d)\phi_{q,d'}(x_{d'})
\right].
\label{eq:expansion-intuition}
\end{align}
The first sum contains only univariate squares, while each cross term is a product of functions defined on distinct coordinates. The expanded representation is therefore decomposable.

Assume now that the input density factorizes as $\rho(\mathbf{x})=\prod_{d=1}^{D}\rho_d(x_d)$, and define the one-dimensional leaf moments
\begin{align}
\mu_{q,d}^{(k)}
=
\int
\phi_{q,d}^{k}(x_d)\rho_d(x_d)\,\mathrm{d}x_d.
\end{align}
Taking the expectation of \eqref{eq:expansion-intuition} and applying the factorization of $\rho$ yields
\begin{align}
\mathbb{E}_{\rho}[f(\mathbf{x})]
=
\sum_{q=1}^{2D+1}
\left[
\sum_{d=1}^{D}\mu_{q,d}^{(2)}
+
2\sum_{d<d'}
\mu_{q,d}^{(1)}\mu_{q,d'}^{(1)}
\right].
\label{eq:quadratic-moment-factorization}
\end{align}
Hence, the $D$-dimensional expectation depends only on univariate leaf moments of orders one and two. This reduction is algebraically exact; any remaining approximation arises only from the quadrature used to evaluate those moments.

For finite-degree polynomial outer functions, the multinomial theorem produces the same decomposable structure: every expanded monomial is a product of powers of univariate inner functions, and its expectation separates under a factorizing density. Recursive expansion extends the argument to deeper polynomial compositions and polynomial observables $h$, so higher moments reduce to the corresponding leaf moments whenever they exist; in particular, the first and second moments determine the variance. \ref{app:tractability-details} gives the complete higher-order and deep-layer derivations, the circuit implementation, and the complexity bounds for the explicitly expanded representation, while Theorem~\ref{thm:tkan-tractability} establishes the result for arbitrary depth and polynomial observables. 

\section{Method}
\label{sec:method}

\subsection{A tractable KAN architecture}
\label{sec:tractable-architecture}

The preceding section identifies two structural properties that are sufficient to obtain a decomposable sum--product representation: (i) the outer nonlinearities must be finite-degree polynomials, so that the multinomial theorem applies; and (ii) the argument of the innermost outer nonlinearity must be a sum of univariate functions, so that the factors in every product generated by the expansion have pairwise disjoint scopes. We incorporate these properties into a tractable modification of the two-block KKAN architecture, which we call TOUCKAN~\cite{toscano2025kurkovakans}.

The resulting TOUCKAN architecture comprises an inner block, a combination layer, and an outer block (Figure~\ref{fig:tkan-architecture}),
\begin{align} \label{eq:tkan-architecture}
f(\mathbf{x})
=
\sum_{q=1}^{Q}\Phi_q(\xi_q),
\qquad
\xi_q
=
\frac{1}{D}\sum_{d=1}^{D}\Psi_{d,q}(x_d).
\end{align}
Here, $Q$ denotes the number of branches, $\Psi_{d,q}$ denotes a trainable univariate inner branch, and $\Phi_q$ denotes a trainable outer function. We use $\Psi_{d,q}$ to distinguish these trainable branches from the exact KART inner functions $\phi_{q,d}$ in \eqref{eq:def-KAST}. The normalized sums $\xi_q$ define the combination variables. Because every $\Psi_{d,q}$ depends on a single coordinate, each $\xi_q$ remains a sum of univariate functions, satisfying condition (ii) by construction.

\begin{figure}[!ht]
    \centering
    \includegraphics[width=0.99\linewidth]{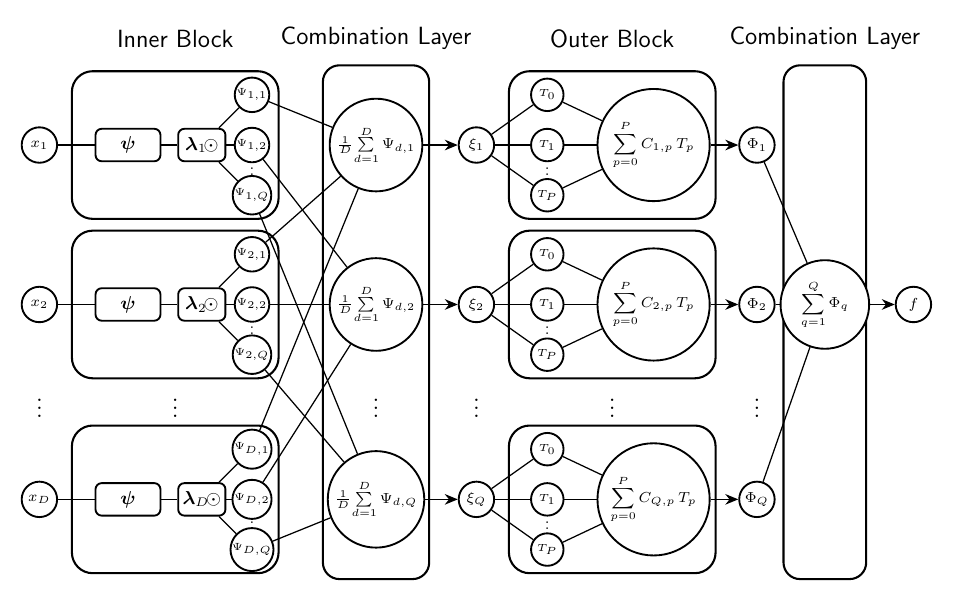}
    \caption{TOUCKAN architecture \eqref{eq:tkan-architecture}, shown in the parameter-efficient diagonal form \eqref{eq:diag-specialization}. Every input dimension $x_d$ is processed by the \emph{same} shared network $\boldsymbol{\psi} : \mathbb{R} \to \mathbb{R}^{Q}$ and specialized by the per-dimension diagonal scaling $\boldsymbol{\lambda}_d$, producing the bounded univariate branches $\Psi_{d,q} = \tanh\!\big(\lambda_{d,q}\,\psi_q(x_d)\big)$. The first combination layer forms the normalized sums $\xi_q = \frac{1}{D}\sum_{d=1}^{D}\Psi_{d,q} \in (-1,1)$, each outer block evaluates a Chebyshev expansion $\Phi_q(\xi_q) = \sum_{p=0}^{P} C_{q,p}\,T_p(\xi_q)$, and the final combination layer sums all branches to produce $f$.}
    \label{fig:tkan-architecture}
\end{figure}

The branches remain free modeling choices provided the required univariate leaf moments exist. In the most direct instantiation, each dimension carries its own independent inner network, and the $Q$ output channels of that network are the branches $\Psi_{d,1},\dots,\Psi_{d,Q}$. Denoting the independent network associated with dimension $d$ by $\boldsymbol{\psi}_d=(\psi_{d,1},\dots,\psi_{d,Q})$, the inner branches and outer functions are
\begin{align}
\Psi_{d,q}(x_d)
&=
\tanh\!\big(\psi_{d,q}(x_d)\big),
\\
\Phi_q(\xi_q)
&=
\sum_{p=0}^{P} C_{q,p}\,T_p(\xi_q).
\label{eq:tkan-outer}
\end{align}
The finite Chebyshev expansion makes each $\Phi_q$ a polynomial, satisfying condition (i). To retain the natural bounded domain of the Chebyshev basis, the combination variable $\xi_q$ must remain in $(-1,1)$. Conventional Chebyshev-based KAN formulations impose this bound by applying $\tanh$ after the sum over dimensions~\cite{toscano2025kurkovakans}. That placement introduces a nonpolynomial function of a multivariate sum into the outer map, compromising decomposability. TOUCKAN instead applies $\tanh$ within each univariate branch. Consequently, every $\Psi_{d,q}$ remains univariate and bounded, while the normalized average in \eqref{eq:tkan-architecture} confines $\xi_q$ to $(-1,1)$. The construction therefore satisfies condition (ii) while providing the appropriate input domain for the outer Chebyshev expansion. Although this direct construction satisfies both requirements, it assigns a complete inner network to every input dimension, causing the inner-block parameter count to grow proportionally with $D$.

To address this problem, we turn to Lorentz's variant of KART~\cite{lorentz1962metric}, which states that the inner functions can be chosen independently of the input dimension up to per-dimension constants,
\begin{align} \label{eq:lorentz-form}
f(\mathbf{x})
=
\sum_{q=1}^{Q}
g\Big(
\sum_{d=1}^{D} \lambda_d\,\psi_q(x_d)
\Big).
\end{align}
In this representation, the same inner functions $\psi_q$ are used for every input dimension, while the constants $\lambda_d$ distinguish their contributions. Laczkovich~\cite{laczkovich2021superposition} extended this result by allowing the constants to depend on both the dimension and the branch through $\lambda_{d,q}$. Following this construction, we replace the $D$ independent inner networks with a single trainable network $\boldsymbol{\psi}=(\psi_1,\dots,\psi_Q):\mathbb{R}\to\mathbb{R}^{Q}$ that is evaluated on every input dimension (see Figure~\ref{fig:tkan-architecture}). The outer block remains unchanged, with one branch-specific function $\Phi_q$ for each branch as defined in \eqref{eq:tkan-outer}.

In the resulting formulation, each dimension $d$ carries one trainable vector $\boldsymbol{\lambda}_d\in\mathbb{R}^{Q}$, and its branches are obtained through a diagonal specialization of the shared functions,
\begin{align} \label{eq:diag-specialization}
\Psi_{d,q}(x_d)
&=
\tanh\!\big(\lambda_{d,q}\,\psi_q(x_d)\big),
\\
\boldsymbol{\Psi}_{d}(x_d)
&=
\tanh\!\big(\mathrm{diag}(\boldsymbol{\lambda}_d)\,\boldsymbol{\psi}(x_d)\big).
\end{align}
Here, $\mathrm{diag}(\boldsymbol{\lambda}_d)$ denotes the diagonal matrix with entries $\lambda_{d,1},\dots,\lambda_{d,Q}$, and $\tanh$ acts elementwise. Every dimension therefore evaluates the same $Q$ functions and rescales their outputs through $\lambda_{d,q}$. These coefficients are the trainable counterparts of Laczkovich's constants. Since each branch remains a univariate function of $x_d$, condition (ii) is preserved.

Sharing the inner network also reduces the parameter count. Let $N_{\boldsymbol{\psi}}$ denote the number of trainable parameters in one complete inner network, and let $H$ denote the width of its trunk. The independent specialization assigns one complete network to every input dimension and requires $D\,N_{\boldsymbol{\psi}}$ parameters. The full-head specialization shares the trunk but retains one dimension-specific linear head $W_d\in\mathbb{R}^{Q\times H}$, requiring $N_{\boldsymbol{\psi}}+D\,Q\,H$ parameters. The diagonal specialization \eqref{eq:diag-specialization} replaces each full head with $Q$ scalars and requires $N_{\boldsymbol{\psi}}+D\,Q$ parameters. Finally, the rank-one specialization imposes $\lambda_{d,q}=\lambda_d$ and requires $N_{\boldsymbol{\psi}}+D$ parameters, recovering the sharing pattern of the Lorentz--Laczkovich form in \eqref{eq:lorentz-form}.

Together, these choices define a hierarchy from independent per-dimension networks to the compact rank-one formulation. The inner parameterization can therefore be selected according to the required flexibility and available parameter budget. All four parameterizations retain bounded univariate branches and the polynomial outer functions in \eqref{eq:tkan-outer}, so conditions (i) and (ii) hold throughout the hierarchy. The rank-one class is contained in every richer parameterization. Its universal approximation property therefore extends to the complete hierarchy (Theorem~\ref{thm:tkan-uat} and Remark~\ref{rem:touckan-hierarchy-uat}). Section~\ref{sec:param-efficiency} compares their finite-parameter accuracy and cost quantitatively.

\subsection{Exact expectation via binomial moment recursion}
\label{sec:moment-recursion}

We assume a factorizing input density, $\rho(\mathbf{x})=\prod_{d=1}^{D}\rho_d(x_d)$. The expectation of every TOUCKAN parameterization in Section~\ref{sec:tractable-architecture} then reduces exactly to one-dimensional leaf moments. The factorization makes the coordinates $x_1,\dots,x_D$ mutually independent, enabling the dimensionwise recursion below. The complete derivation, numerical-stability analysis, and pseudocode (Algorithm~\ref{alg:tkan-expectation}) are provided in \ref{app:algorithms}.

Each Chebyshev polynomial admits the monomial form $T_p(\xi) = \sum_{j=0}^{p} a_{p,j}\,\xi^{j}$ with fixed integer coefficients $a_{p,j}$. Substituting it into the outer expansion \eqref{eq:tkan-outer} and applying linearity of expectation to \eqref{eq:tkan-architecture} gives
\begin{align} \label{eq:tkan-mean-reduction}
\mathbb{E}\big[f(\mathbf{x})\big]
=
\sum_{q=1}^{Q}\sum_{p=0}^{P} C_{q,p}
\sum_{j=0}^{p} a_{p,j}\,
\mathbb{E}\big[\xi_q^{\,j}\big],
\end{align}
so the entire computation reduces to the moments $m_{q,p} = \mathbb{E}[\xi_q^{\,p}]$ of each combination variable for $p = 0,\dots,P$. Since expanding the power $\xi_q^{\,p}$ directly by the multinomial theorem would generate $\binom{D+p-1}{p}$ cross terms, the moments $m_{q,p}$ are assembled one dimension at a time. To simplify the notation, we absorb the normalization of the combination layer into the branches by defining $\widetilde{\Psi}_{d,q}=\Psi_{d,q}/D$, which gives $\xi_q=\sum_{d=1}^{D}\widetilde{\Psi}_{d,q}(x_d)$. The moments of the partial sums $S_{j,q}$ then obey the binomial recursion
\begin{align}
\mathbb{E}\big[S_{j,q}^{\,p}\big]
=
\sum_{k=0}^{p}\binom{p}{k}\,
\mathbb{E}\big[S_{j-1,q}^{\,p-k}\big]\,
\mu_{j,q}^{(k)},
\qquad
S_{j,q}
=
\sum_{d=1}^{j}\widetilde{\Psi}_{d,q}(x_d),
\qquad
\mu_{d,q}^{(k)}
=
\mathbb{E}\big[\widetilde{\Psi}_{d,q}^{\,k}(x_d)\big].
\label{eq:tkan-binomial-recursion}
\end{align}
Notice that we can apply this recursion since $S_{j-1,q}$ and $\widetilde{\Psi}_{j,q}(x_j)$ are independent under a factorizing density. It is initialized by $\mathbb{E}[S_{1,q}^{\,p}]=\mu_{1,q}^{(p)}$ and terminates at $S_{D,q}=\xi_q$, yielding the moments $m_{q,p}$. The implementation uses the equivalent centered recursion derived in \ref{app:alg-mean} to keep the intermediate sums bounded.

Each leaf moment $\mu_{d,q}^{(k)}$ is a univariate integral evaluated with an $M$-node Gaussian quadrature rule matched to the marginal $\rho_d$. We use Gauss--Legendre quadrature for uniform marginals and Gauss--Hermite quadrature for Gaussian marginals. These leaf integrals are the only quantities discretized in the moment calculation. The recursion \eqref{eq:tkan-binomial-recursion} and reduction \eqref{eq:tkan-mean-reduction} are algebraically exact in the resulting leaf moments. The computed expectation is therefore exact for the quadrature-induced discrete product measure. Its continuum error is the corresponding one-dimensional quadrature error, which decays spectrally in $M$ for the smooth branches used here.

The quadrature rule controls only the error in evaluating the trained surrogate. For an underlying target $F$ approximated by $f$, the error of the computed value $\widehat{\mathbb{E}}_{M}[f]$ decomposes into a surrogate approximation term and a quadrature term,
\begin{align} \label{eq:tkan-error-decomposition}
\big|\mathbb{E}[F] - \widehat{\mathbb{E}}_{M}[f]\big|
\le
\mathbb{E}\big[|F - f|\big]
+
\big|\mathbb{E}[f] - \widehat{\mathbb{E}}_{M}[f]\big|.
\end{align}
The first term measures the error of the fitted surrogate, whereas the second is the one-dimensional quadrature error. Because the quadrature term decays spectrally with $M$, it becomes negligible once the leaf integrals are sufficiently resolved, leaving the fitted-surrogate error as the dominant contribution. Further quadrature refinement cannot reduce this error. Improving the estimate instead requires a more accurate surrogate fit, which motivates the strong optimization procedures used in our experiments. \ref{app:sensitivity} measures both contributions against exact and independent high-resolution references.

Substituting the moments $m_{q,p}$ into \eqref{eq:tkan-mean-reduction} yields the expectation of the surrogate without Monte Carlo sampling or multidimensional quadrature. The calculation requires $DM$ inner-block evaluations and $\mathcal{O}(QDPM+QDP^{2})$ arithmetic operations. For fixed $P$ and $M$, the arithmetic cost is proportional to $DQ$. The resulting recursion realizes the general leaf-replacement identity \eqref{eq:composition-operator-commutator}, established in \ref{app:tractability-details}, for the TOUCKAN architecture.

\subsection{Exact variance via cross-moment recursion}
\label{sec:variance-recursion}

The same moment framework also gives the variance, $\mathrm{Var}[f]=\mathbb{E}[f^{2}]-\big(\mathbb{E}[f]\big)^{2}$. Computing the second moment $\mathbb{E}[f^{2}]$ extends the recursion of Section~\ref{sec:moment-recursion} from individual combination variables to pairs. \ref{app:algorithms} provides the complete derivation and pseudocode (Algorithm~\ref{alg:tkan-variance}).

Collecting powers of the Chebyshev expansion \eqref{eq:tkan-outer} writes each outer block in monomial form, $\Phi_q(\xi_q) = \sum_{j=0}^{P}\beta_{q,j}\,\xi_q^{\,j}$. Squaring $f = \sum_{q}\Phi_q(\xi_q)$ and taking the expectation termwise gives
\begin{align} \label{eq:tkan-second-moment-reduction}
\mathbb{E}\big[f^{2}(\mathbf{x})\big]
=
\sum_{q,q'=1}^{Q}
\sum_{j,j'=0}^{P}
\beta_{q,j}\,\beta_{q',j'}\;
m_{q,q'}^{(j,j')},
\qquad
\beta_{q,j}
=
\sum_{p=j}^{P} C_{q,p}\, a_{p,j},
\qquad
m_{q,q'}^{(j,j')}
=
\mathbb{E}\big[\xi_q^{\,j}\,\xi_{q'}^{\,j'}\big].
\end{align}
We call $m_{q,q'}^{(j,j')}$ the cross-moments between pairs of combination variables. Although the input coordinates are independent, $\xi_q$ and $\xi_{q'}$ depend on the same coordinates and are therefore generally dependent. Their cross-moments cannot in general be assembled from the individual moments computed in Section~\ref{sec:moment-recursion}, so the binomial recursion must be extended to pairs.

The joint moments of the partial-sum pairs obey the bivariate recursion
\begin{align}
\mathbb{E}\big[S_{j,q}^{\,p}\, S_{j,q'}^{\,p'}\big]
=
\sum_{a=0}^{p}\sum_{b=0}^{p'}
\binom{p}{a}\binom{p'}{b}\,
\mathbb{E}\big[S_{j-1,q}^{\,p-a}\, S_{j-1,q'}^{\,p'-b}\big]\,
\nu_{j,q,q'}^{(a,b)},
\qquad
\nu_{d,q,q'}^{(k,k')}
=
\mathbb{E}\big[\widetilde{\Psi}_{d,q}^{\,k}(x_d)\,
\widetilde{\Psi}_{d,q'}^{\,k'}(x_d)\big].
\label{eq:tkan-cross-recursion}
\end{align}
This is the bivariate analogue of \eqref{eq:tkan-binomial-recursion}, which is recovered by setting $p'=0$. The recursion is initialized by $\mathbb{E}[S_{1,q}^{\,p}S_{1,q'}^{\,p'}]=\nu_{1,q,q'}^{(p,p')}$ and delivers all cross-moments $m_{q,q'}^{(j,j')}$ at the final coordinate.

Despite the bivariate recursion, each joint leaf moment $\nu_{d,q,q'}^{(k,k')}$ remains a univariate integral because both factors depend only on the scalar coordinate $x_d$. These moments are evaluated with the same measure-matched quadrature used for the marginal leaf moments in Section~\ref{sec:moment-recursion}.

Substituting the cross-moments into \eqref{eq:tkan-second-moment-reduction} and subtracting the squared expectation computed in Section~\ref{sec:moment-recursion} yields the variance,
\begin{align} \label{eq:tkan-exact-variance}
\mathrm{Var}\big[f(\mathbf{x})\big]
=
\sum_{q,q'=1}^{Q}
\sum_{j,j'=0}^{P}
\beta_{q,j}\,\beta_{q',j'}\,
m_{q,q'}^{(j,j')}
-
\Big(\mathbb{E}\big[f(\mathbf{x})\big]\Big)^{2}.
\end{align}
As in Section~\ref{sec:moment-recursion}, the cross-moment recursion is algebraically exact in the joint leaf moments. The computed variance is therefore exact for the quadrature-induced discrete product measure, while its continuum integration error is inherited from the same one-dimensional quadrature rules. The calculation involves no Monte Carlo sampling or multidimensional integration and requires $\mathcal{O}(Q^{2}DP^{2}M+Q^{2}DP^{4})$ arithmetic operations. For fixed $P$ and $M$, this cost is proportional to $DQ^{2}$. The construction extends to the $r$-th moment of $f$ by propagating joint moments of $r$-tuples of combination variables. For fixed $r$, the resulting cost remains polynomial in $D$, $Q$, and $P$. The first- and second-moment recursions provide the algebraic machinery used below to enforce conserved integrals and normalization.

\subsection{Exact enforcement of conserved integrals}
\label{sec:exact-mass}

The expectation recursion also makes it possible to enforce linear integral constraints. Consider a function $u(t,\mathbf{x})$ on a box domain $\Omega=\prod_{d=1}^{D}[a_d,b_d]$ whose weighted spatial integral must remain constant in time,
\begin{align} \label{eq:conserved-functional}
\mathcal{Q}_w[u](t)
=
\int_{\Omega} w(\mathbf{x})\, u(t,\mathbf{x})\, d\mathbf{x}
=
\mathcal{Q}_0
\qquad \text{for all } t,
\end{align}
where the weight factorizes as $w(\mathbf{x})=\prod_{d=1}^{D}w_d(x_d)$. The choice $w\equiv1$ gives the total mass, while $w(\mathbf{x})=x_e$ gives the first spatial moment in coordinate $e$.
For a TOUCKAN surrogate, the same dimensionwise recursion evaluates $\mathcal{Q}_w[f](t)$. Let $f(t,\mathbf{x})$ follow \eqref{eq:tkan-architecture} with input vector $(t,x_1,\dots,x_D)$, so that the complete input dimension is $D+1$ and time has index $d=0$. Each combination variable then contains one time branch and one branch for every spatial coordinate,
\begin{align} \label{eq:tkan-spacetime-xi}
\xi_q
=
\frac{1}{D+1}\Big(\Psi_{0,q}(t) + \sum_{d=1}^{D}\Psi_{d,q}(x_d)\Big).
\end{align}
As in Section~\ref{sec:moment-recursion}, we absorb the normalization into every branch by defining $\widetilde{\Psi}_{d,q}=\Psi_{d,q}/(D+1)$ for $d=0,\dots,D$. At fixed $t$, the time branch is a scalar, whereas each spatial branch depends on only one coordinate. Expanding the outer Chebyshev polynomials therefore reduces $\mathcal{Q}_w[f](t)$ to weighted integrals of powers of $\xi_q$, which the binomial recursion assembles from one-dimensional leaf integrals. Relative to Section~\ref{sec:moment-recursion}, the marginal density $\rho_d$ is replaced by the weight $w_d$, while the time branch contributes the known powers $\widetilde{\Psi}_{0,q}(t)^k$ in place of moments. The resulting weighted leaf moments and combination-variable moments are
\begin{align} \label{eq:tkan-weighted-moments}
\mu_{d,q}^{(k)}(w)
=
\int \widetilde{\Psi}_{d,q}(x_d)^{k}\, w_d(x_d)\, \mathrm{d}x_d,
\qquad
m_{q,j}(t)
=
\int_{\Omega} \xi_q^{\,j}\, w(\mathbf{x})\, \mathrm{d}\mathbf{x}.
\end{align}
Here, $m_{q,j}(t)$ is the weighted counterpart of the moment $m_{q,j}=\mathbb{E}[\xi_q^{\,j}]$ from Section~\ref{sec:moment-recursion}. The recursion requires the weight to factorize but not to be positive, so signed factors such as $w_d=x_d$ are admissible. Substituting these weighted moments into \eqref{eq:tkan-mean-reduction} gives
\begin{align} \label{eq:tkan-conserved-closed-form}
\mathcal{Q}_w[f](t)
=
\sum_{q=1}^{Q}\sum_{p=0}^{P} C_{q,p}
\sum_{j=0}^{p} a_{p,j}\, m_{q,j}(t).
\end{align}
This expression is differentiable with respect to the network parameters and time. Its weighted leaf moments are evaluated by one-dimensional quadrature, and the complete calculation has the same $\mathcal{O}(QDPM+QDP^{2})$ arithmetic cost as the expectation recursion in Section~\ref{sec:moment-recursion}.

The deterministic evaluation of $\mathcal{Q}_w[f](t)$ allows the conservation law \eqref{eq:conserved-functional} to be imposed by a final output reprojection. At each time, the drift is $\mathcal{Q}_w[f](t)-\mathcal{Q}_0$. We remove it by subtracting a multiple of a fixed correction profile $g(\mathbf{x})$, whose spatial form determines how the correction is distributed. The coefficient is chosen so that the correction carries exactly the drift and requires only $\mathcal{Q}_w[g]\neq0$. For total mass, $g\equiv1$ gives $\mathcal{Q}_w[g]=|\Omega|$. The constrained surrogate is
\begin{align} \label{eq:tkan-mass-projection}
\widetilde{f}(t,\mathbf{x})
=
f(t,\mathbf{x})
-
\frac{\mathcal{Q}_w[f](t) - \mathcal{Q}_0}{\mathcal{Q}_w[g]}\;
g(\mathbf{x}),
\qquad\text{so that}\qquad
\mathcal{Q}_w[\widetilde{f}\,](t) = \mathcal{Q}_0
\quad \text{for all } t.
\end{align}
Equation~\eqref{eq:tkan-mass-projection} satisfies the prescribed functional for every value of the trainable parameters. Because the weighted leaf moments are evaluated by one-dimensional quadrature, the reprojection is exact for the quadrature-discretized functional. Conservation of the continuum integral \eqref{eq:conserved-functional} inherits the accuracy of those quadrature rules. The correction is differentiable with respect to the network parameters and can therefore be treated as part of the network during training.

\subsection{Exact density estimation: normalization as a hard constraint}
\label{sec:exact-density}

The conserved functional in Section~\ref{sec:exact-mass} is linear and can be enforced by an additive reprojection. A normalized density must instead remain nonnegative and integrate to one. Existing model classes handle this normalization in different ways. Energy-based models \cite{liu2021density,rippel2013high} generally require sampling when estimating an intractable partition function. Normalizing flows \cite{papamakarios2021normalizing} avoid direct integration by restricting the architecture to invertible maps with tractable Jacobians, while penalty formulations enforce unit integral only approximately. For a TOUCKAN surrogate, the normalizer is the second moment computed by the cross-moment recursion of Section~\ref{sec:variance-recursion}. Normalization can therefore be enforced multiplicatively by division.

Let $\rho(\mathbf{x}) = \prod_{d=1}^{D}\rho_d(x_d)$ be a factorizing base density and let $f$ be a TOUCKAN \eqref{eq:tkan-architecture}. Define the normalized model
\begin{align} \label{eq:tkan-born-density}
p(\mathbf{x})
=
\frac{f^{2}(\mathbf{x})\,\rho(\mathbf{x})}{\mathbb{E}_{\rho}\big[f^{2}\big]},
\qquad\text{equivalently}\qquad
\psi(\mathbf{x})
=
\frac{f(\mathbf{x})\,\sqrt{\rho(\mathbf{x})}}{\sqrt{\mathbb{E}_{\rho}\big[f^{2}\big]}},
\qquad
p = \psi^{2},
\end{align}
The denominator $\mathbb{E}_{\rho}[f^{2}]$ is supplied by \eqref{eq:tkan-second-moment-reduction} and the cross-moment recursion \eqref{eq:tkan-cross-recursion}. The model is well defined whenever $\mathbb{E}_{\rho}[f^{2}]>0$, a condition that fails only when $f$ vanishes $\rho$-almost everywhere. Under this condition, squaring guarantees $p\geq0$, while division by the second moment enforces $\int p=1$ for every parameter value. Numerically, this identity is exact for the quadrature-discretized normalizer, and continuum normalization inherits the accuracy of the one-dimensional quadrature rules. Because the denominator is recomputed from the current parameters, normalization holds throughout training and does not depend on the optimizer.

The cross-moment recursion evaluates $\mathbb{E}_{\rho}[f^{2}]$ as a deterministic, differentiable function of the trainable parameters with $\mathcal{O}(Q^{2}DP^{2}M+Q^{2}DP^{4})$ arithmetic cost. For fixed $P$ and $M$, this cost is proportional to $DQ^{2}$. Gradients therefore pass directly through the normalizer during training.

For any nonzero scalar $c$, the transformation $f\mapsto cf$ leaves $p$ unchanged. Global scale is therefore removed from the represented density rather than controlled by an auxiliary constraint. In variational problems such as quantum eigenvalue problems, where $p=\psi^{2}$, normalization requires no separate loss term. The deterministic denominator also avoids the sampling floor introduced by a Monte Carlo estimate of the normalizer. This hard-normalized representation underlies the Schr\"odinger eigensolver evaluated in Section~\ref{sec:schrodinger}.

\section{Results}
\label{sec:experiments}

All experiments use the two-block TOUCKAN architecture introduced in Section~\ref{sec:tractable-architecture}, with the inner-block specialization stated for each case. The implementation uses JAX/Flax; benchmark definitions and implementation details are provided in \ref{app:problems} and~\ref{app:implementation}, respectively.

\subsection{Parameter efficiency}
\label{sec:param-efficiency}

\begin{table}[H]
\centering
\caption{Parameter efficiency of the four TOUCKAN inner-block specializations on the L\"uthen benchmark at $D=20$. The table reports trainable parameters, relative $L_2$ test error, and relative errors of the analytically propagated mean and variance against exact references. All variants share the same outer block and moment recursion. Error columns report the mean $\pm$ standard deviation over five random seeds, and bold marks the best mean in each column.}
\medskip
\resizebox{\linewidth}{!}{%
\begin{tabular}{lcccc}
\hline
Specialization and inner block $\boldsymbol{\Psi}_{d}(x_d)$ & Params & rel.\ $L_2$ & $\mathbb{E}[f]$ rel.\ error & $\mathrm{Var}[f]$ rel.\ error \\
\hline
Independent $\quad \tanh\!\big(\boldsymbol{\psi}_d(x_d)\big)$ \cite{toscano2025kurkovakans} & $41{,}796$ & $\mathbf{(8.8 \pm 5.5)\times10^{-4}}$ & $\mathbf{(7.2 \pm 6.5)\times10^{-4}}$ & $\mathbf{(1.9 \pm 2.4)\times10^{-4}}$ \\
Full-head $\quad \tanh\!\big(W_d\, h(x_d) + b_d\big)$ & $19{,}072$ & $(9.5 \pm 2.9)\times10^{-4}$ & $(8.8 \pm 3.0)\times10^{-4}$ & $(2.3 \pm 1.9)\times10^{-4}$ \\
Diagonal $\quad \tanh\!\big(\mathrm{diag}(\boldsymbol{\lambda}_d)\, \boldsymbol{\psi}(x_d)\big)$ \cite{laczkovich2021superposition} & $3{,}533$ & $(3.6 \pm 1.2)\times10^{-3}$ & $(2.9 \pm 1.4)\times10^{-3}$ & $(2.4 \pm 1.7)\times10^{-4}$ \\
Rank-one $\quad \tanh\!\big(\lambda_{d}\, \boldsymbol{\psi}(x_d)\big)$ \cite{lorentz1962metric} & $2{,}733$ & $(7.7 \pm 0.6)\times10^{-2}$ & $(1.7 \pm 1.1)\times10^{-2}$ & $(4.0 \pm 1.9)\times10^{-3}$ \\
\hline
\end{tabular}%
}
\label{tab:param-efficiency}
\end{table}

We evaluate the parameter--accuracy trade-off of the inner-block hierarchy on the L\"uthen benchmark~\cite{luthen2021Sparse} at $D=20$, whose mean and variance are available exactly. The comparison includes the independent, full-head, diagonal, and rank-one specializations defined in Section~\ref{sec:tractable-architecture}. Only the inner parameterization changes. All variants use the same outer block and SOAP training protocol~\cite{vyas2024soap} across five random seeds. The benchmark, data preparation, and complete training configuration are provided in \ref{app:luthen} and~\ref{app:impl-luthen}.

For every trained surrogate, the analytically propagated moments agree with a $2\times10^{6}$-sample Monte Carlo estimate of the same network to within one Monte Carlo standard error. This independent check confirms that the errors against the exact benchmark moments are dominated by surrogate fit rather than moment propagation. The mean and variance together evaluate in $0.02$ s, and the ordering of the four specializations is preserved in every seed.

The diagonal specialization uses $11.8$ times fewer parameters than the independent specialization while matching its variance accuracy within seed variability. At this seed count, the independent, full-head, and diagonal variants are statistically indistinguishable in variance accuracy, with mean relative errors near $2\times10^{-4}$. The rank-one specialization reduces the parameter count only modestly beyond the diagonal form but increases the relative $L_2$ error to approximately $8\%$.

This finite-budget loss does not contradict the universal-approximation result in \ref{app:universality}. It is consistent with the rank-one structure, in which the per-dimension arguments $\lambda_d\boldsymbol{\psi}(x_d)$ are collinear before the $\tanh$ and can distinguish shape differences only through saturation. On this benchmark, the diagonal coupling $\lambda_{d,q}$ is therefore the least parameterized specialization that matches the richer variants in variance accuracy while maintaining a sub-percent relative $L_2$ error. We consequently use the diagonal specialization for the high-dimensional inference experiments in Section~\ref{sec:tractable-inference}.

\subsection{Tractable inference}
\label{sec:tractable-inference}

Using the diagonal specialization selected in Section~\ref{sec:param-efficiency}, we evaluate tractable inference across $D\in\{5,12,25,50\}$ on three benchmarks. To test performance under two factorizing input densities, we use the L\"uthen~\cite{luthen2021Sparse} and Bratley~\cite{bratley1992Implementation} functions with uniform inputs and the Oakley--O'Hagan-type function~\cite{oakley2004probabilistic} with Gaussian inputs. All three admit exact means and variances, providing reference values for the propagated moments. Their definitions and exact reference calculations are given in \ref{app:luthen}--\ref{app:oakley}. The diagonal architecture remains compact at these dimensions, making second-order optimization with SSBroyden computationally feasible and thereby further reducing the optimization error of the fitted surrogate. The complete training procedure is summarized by Algorithm~\ref{alg:highd-training} in \ref{app:impl-highd}.

The recursions directly compute the mean and raw second moment $\mathbb{E}[f^2]$. At $D=50$, the deterministic estimates overlap the exact references for both quantities across all three benchmarks, whereas the Monte Carlo estimates remain visibly dispersed at $N=2\times10^{6}$ (Figure~\ref{fig:mc-convergence-rows}).

Table~\ref{tab:tractable-inference-errors} extends this comparison across all four dimensions and reports the variance derived from the two propagated moments.

\begin{table}[H]
\centering
\caption{Relative errors of the expectation and variance for the three benchmarks across $D\in\{5,12,25,50\}$. The surrogate column reports the test relative $L_2$ error. TOUCKAN columns report one deterministic propagation through the trained surrogate. MC (median) is the median error over $200$ independent runs of the true integrand, each with $N=2\times10^{5}$ samples (\ref{app:impl-highd}). TOUCKAN is more accurate than the median Monte Carlo estimate in every case. For $D\leq25$, it outperforms all $200$ Monte Carlo mean replications and all but at most one variance replication. At $D=50$, it outperforms at least $98\%$ of the mean and $89\%$ of the variance replications. The variance entries marked with $\ddagger$ are limited by floating-point roundoff in the second-moment recursion rather than by the surrogate.}
\label{tab:tractable-inference-errors}
\medskip 
\resizebox{\linewidth}{!}{%
\begin{tabular}{llc|cc|cc}
\hline
 & & Surrogate & \multicolumn{2}{c}{$\mathbb{E}[f]$ rel.\ error} & \multicolumn{2}{c}{$\mathrm{Var}[f]$ rel.\ error} \\
Benchmark & $D$ & rel.\ $L_2$ & TOUCKAN & MC (median) & TOUCKAN & MC (median) \\
\hline
L\"uthen  & 5  & $1.0\times10^{-6}$ & $5.64\times10^{-10}$ & $1.34\times10^{-3}$ & $1.13\times10^{-8}$ & $2.10\times10^{-3}$ \\
          & 12 & $2.7\times10^{-6}$ & $1.30\times10^{-9}$  & $3.07\times10^{-4}$ & $9.27\times10^{-9}$ & $2.09\times10^{-3}$ \\
          & 25 & $4.5\times10^{-6}$ & $1.77\times10^{-9}$  & $2.53\times10^{-4}$ & $3.81\times10^{-8}$ & $2.35\times10^{-3}$ \\
          & 50 & $8.9\times10^{-6}$ & $1.63\times10^{-10}$ & $1.14\times10^{-4}$ & $1.66\times10^{-8}$ & $1.92\times10^{-3}$ \\
\hline
Bratley   & 5  & $2.9\times10^{-4}$ & $7.22\times10^{-7}$  & $1.03\times10^{-3}$ & $1.80\times10^{-6}$ & $1.74\times10^{-3}$ \\
          & 12 & $1.7\times10^{-3}$ & $1.55\times10^{-7}$  & $1.04\times10^{-3}$ & $1.45\times10^{-5}$ & $1.94\times10^{-3}$ \\
          & 25 & $4.3\times10^{-3}$ & $3.54\times10^{-7}$  & $9.83\times10^{-4}$ & $9.11\times10^{-6}$ & $1.75\times10^{-3}$ \\
          & 50 & $1.1\times10^{-2}$ & $1.31\times10^{-5}$  & $1.08\times10^{-3}$ & $4.59\times10^{-4}\,^{\ddagger}$ & $1.77\times10^{-3}$ \\
\hline
Oakley & 5  & $4.7\times10^{-5}$ & $2.68\times10^{-7}$ & $2.33\times10^{-3}$ & $1.90\times10^{-7}$ & $2.25\times10^{-3}$ \\
          & 12 & $1.4\times10^{-4}$ & $1.35\times10^{-7}$  & $1.45\times10^{-3}$ & $3.45\times10^{-7}$ & $2.17\times10^{-3}$ \\
          & 25 & $2.3\times10^{-4}$ & $4.71\times10^{-7}$  & $9.24\times10^{-4}$ & $1.02\times10^{-6}$ & $2.13\times10^{-3}$ \\
          & 50 & $1.1\times10^{-3}$ & $8.50\times10^{-7}$  & $7.99\times10^{-4}$ & $2.69\times10^{-4}\,^{\ddagger}$ & $1.86\times10^{-3}$ \\
\hline
\end{tabular}%
}
\end{table}

Across dimensions, the propagated mean errors range from $10^{-10}$ to $10^{-5}$ and the variance errors from $10^{-8}$ to $10^{-4}$, remaining below the median fixed-budget Monte Carlo error in every case. The median across replications characterizes the typical error of one fixed-budget Monte Carlo run. Even after all $200$ Monte Carlo runs are averaged at a cost of $4\times10^{7}$ evaluations (Figure~\ref{fig:summary-grid}), the deterministic mean remains more accurate in every case except Bratley at $D=50$. The deterministic variance is less accurate only for Bratley at $D=12$ and $D=50$. These exceptions are confined to Bratley, which has the largest surrogate errors at the affected dimensions.

\begin{figure}[H]
    \centering
    \includegraphics[width=0.88\linewidth]{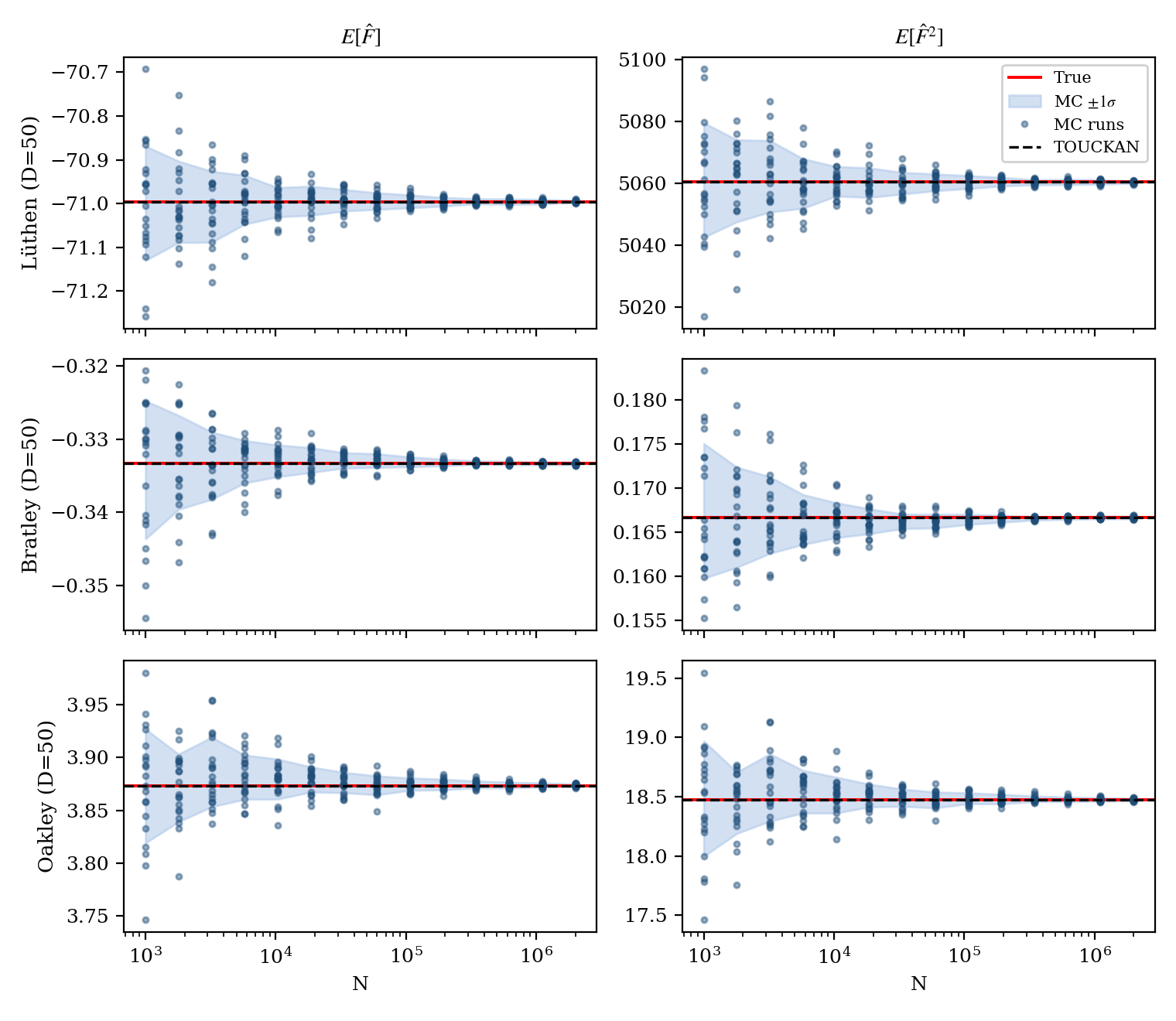}
    \caption{Analytical TOUCKAN moments and Monte Carlo convergence at $D=50$. Rows correspond to the L\"uthen, Bratley, and Oakley--O'Hagan benchmarks, while columns show the mean and raw second moment. Blue dots show $20$ independent Monte Carlo estimates of the true integrand at each sample size up to $N=2\times10^{6}$, and the blue shaded region marks their $\pm1\sigma$ envelope. The red line is the exact value, and the black dashed line is the deterministic TOUCKAN estimate. The Monte Carlo envelope contracts at the expected $N^{-1/2}$ rate but remains visibly dispersed at the largest budget. The TOUCKAN estimate overlaps the exact value in all six panels without additional sampling during moment evaluation.}
    \label{fig:mc-convergence-rows}
\end{figure}

The comparison reflects different error sources. At the fixed sampling budget, the Monte Carlo median remains near $10^{-3}$ across benchmarks and dimensions. Once the one-dimensional quadrature is resolved, the propagated errors instead follow the relative $L_2$ error of the surrogate, consistent with the decomposition in Section~\ref{sec:moment-recursion} and the sensitivity analysis in \ref{app:sensitivity}. The variance at $D=50$ is the exception. There the recursion for the second moment reaches a floating-point roundoff floor before the surrogate error, so the two marked entries of Table~\ref{tab:tractable-inference-errors} overstate the error of the surrogate itself. When roundoff is negligible, improving the deterministic estimate requires a more accurate surrogate, whereas reducing the sampling error requires additional true-function evaluations.

\begin{figure}[H]
    \centering
    \includegraphics[width=0.99\linewidth]{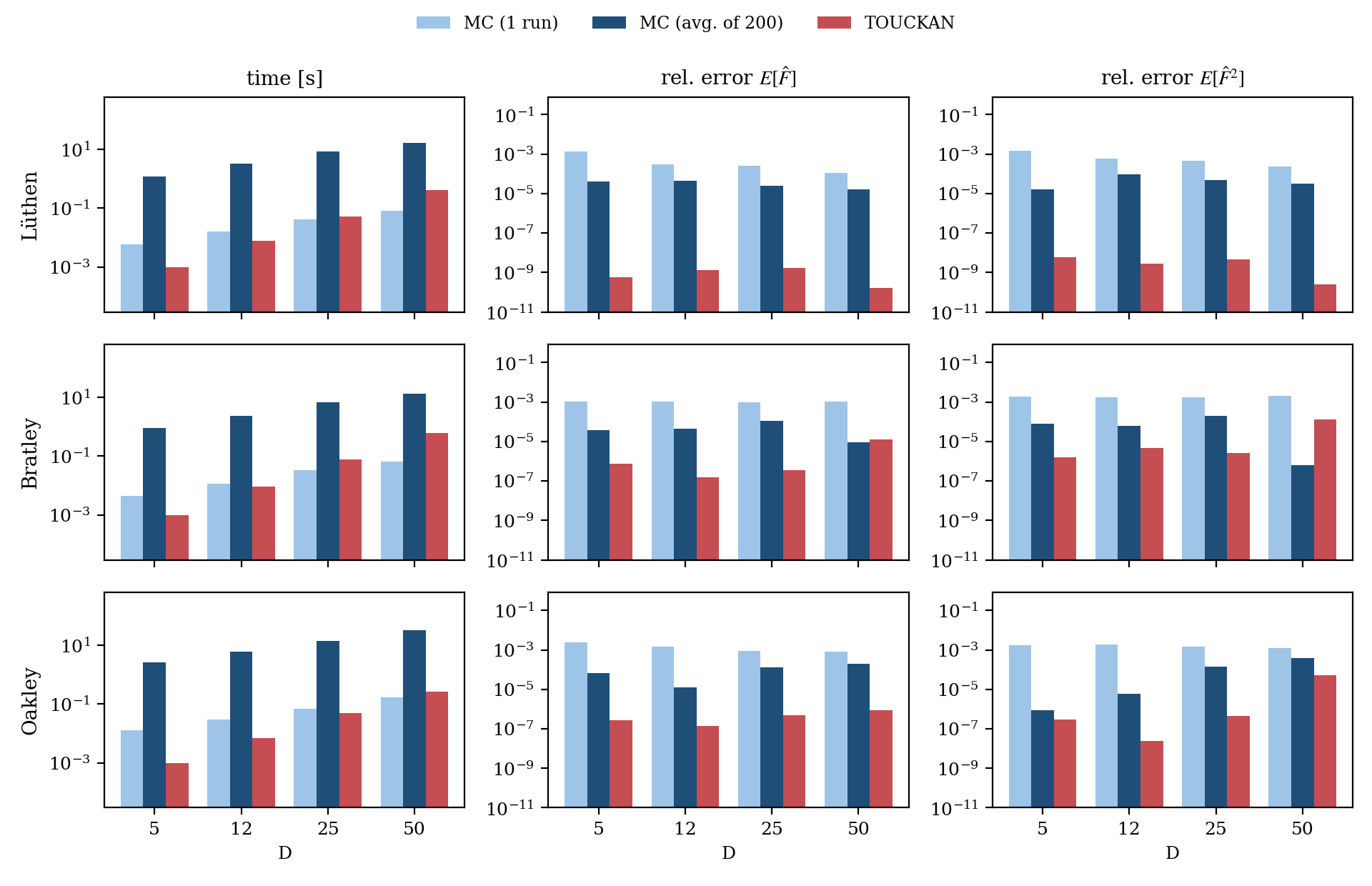}
    \caption{Cost and accuracy of TOUCKAN propagation and Monte Carlo across $D\in\{5,12,25,50\}$. Rows denote benchmarks and columns report wall time, mean error, and raw-second-moment error. MC (1 run) uses the time of one $N=2\times10^{5}$ run and the median error across $200$ independent runs. MC (avg.\ of 200) reports the total time and error obtained by averaging those runs. TOUCKAN reports one warm evaluation of the compiled deterministic recursions, excluding one-time JIT compilation. The reported timings do not include surrogate training. One TOUCKAN evaluation is up to $2{,}600$ times faster than the 200-run Monte Carlo average. For the mean and raw second moment, TOUCKAN matches or improves the averaged Monte Carlo estimator's accuracy except for Bratley at $D=50$. Timings were measured on one NVIDIA L40S GPU.}
    \label{fig:summary-grid}
\end{figure}

Figure~\ref{fig:summary-grid} reports wall time and accuracy for the mean and raw second moment, the two quantities evaluated directly by the recursions. One deterministic propagation takes approximately one millisecond at $D=5$ and less than one second at $D=50$. TOUCKAN is up to $2{,}600$ times faster than the 200-run Monte Carlo average. For the mean and raw second moment, TOUCKAN matches or improves the averaged Monte Carlo estimator's accuracy except for Bratley at $D=50$, with errors up to $10^{5}$ times smaller in the remaining cases.

These comparisons concern inference after the surrogate has been trained. Monte Carlo provides a basic sampling reference for the true integrand, whereas TOUCKAN evaluates moments of the fitted surrogate without additional queries to the original model. The same distinction applies to randomized quasi-Monte Carlo, whose comparison is reported in \ref{app:rqmc-comparison}. Both direct sampling methods require access to the true function at the chosen sampling points. The RQMC comparison uses up to $2^{18}$ evaluations per run at prescribed low-discrepancy points. Obtaining these evaluations can be prohibitive for expensive simulations, while fixed experimental observations or archived data generally do not coincide with the required locations. TOUCKAN instead reuses the trained surrogate for repeated moment evaluation.

\subsection{Exact mass conservation: Cahn--Hilliard equation}
\label{sec:cahn-hilliard}

\begin{figure}[H]
    \centering
    \includegraphics[width=0.95\linewidth]{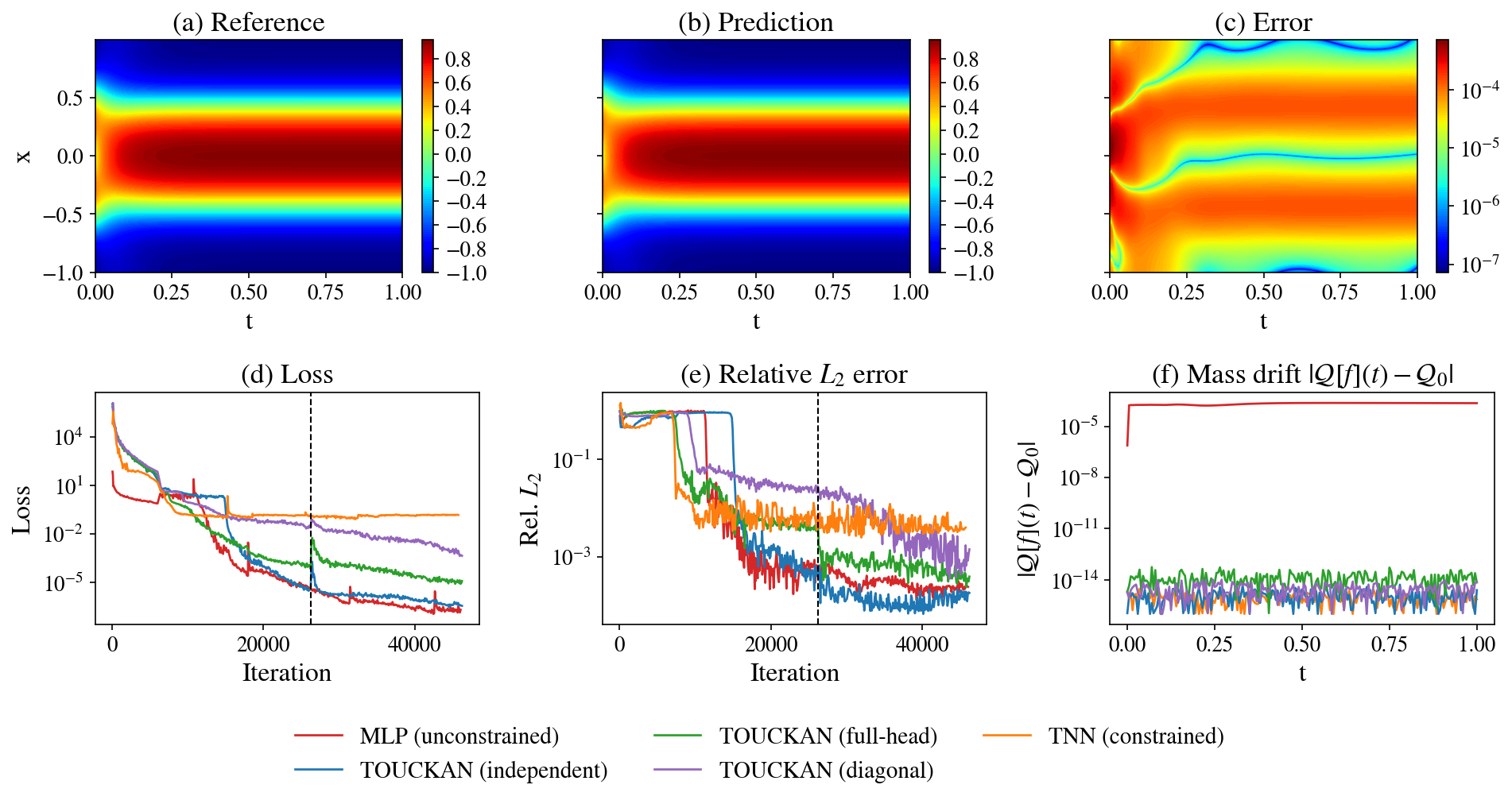}
    \caption{Cahn--Hilliard equation with the exact-mass reprojection. (a)--(c): reference solution, learned surrogate, and absolute approximation error on a logarithmic color scale for the independent TOUCKAN of Table~\ref{tab:cahn-hilliard}. (d), (e): training loss and relative $L_2$ error for the five parameter-matched models of Table~\ref{tab:cahn-hilliard}. The dashed line marks activation of the exact-mass reprojection for the constrained models. The MLP remains unconstrained throughout. (f): 512-node quadrature mass drift $|\mathcal{Q}[f](t) - \mathcal{Q}_0|$ over time. Every model with the exact-mass reprojection holds this quadrature mass at machine precision, while the unconstrained MLP baseline drifts at the level of its field error.}
    \label{fig:cahn-hilliard-solution}
\end{figure}

Beyond post-training inference, deterministic moment propagation can enforce integral constraints throughout training. We demonstrate this capability using the one-dimensional Cahn--Hilliard equation for a double-well potential,
\begin{align} \label{eq:cahn-hilliard}
\partial_t u
=
\partial_{xx}\big(u^{3} - u - \varepsilon\, \partial_{xx} u\big),
\qquad
x \in [-1,1], \quad t \in [0,1], \quad \varepsilon = 0.02.
\end{align}
With periodic boundary conditions and the initial condition $u(0,x)=\cos(\pi x)-e^{-4(\pi x)^2}$, the solution conserves the total mass $\mathcal{Q}[u](t)=\int_{-1}^{1}u(t,x)\,dx=\mathcal{Q}_0$. We impose this invariant using the exact-mass reprojection of Section~\ref{sec:exact-mass} with $w\equiv1$ and $g\equiv1$, so every projected output has total mass $\mathcal{Q}_0$ for every parameter value. The Fourier-spectral reference conserves mass to below $10^{-16}$, and the complete problem and implementation are described in \ref{app:impl-ch}.

We compare an unconstrained MLP baseline, a tensor neural network (TNN)~\cite{wang2022tensor} that we modify to enforce the same exact-mass construction, and three TOUCKAN specializations. The TOUCKANs use the independent, full-head, and diagonal inner blocks defined in Section~\ref{sec:tractable-architecture}. The TNN represents the solution as a sum of separable products in $t$ and $x$, so its mass becomes a weighted sum of one-dimensional spatial factor integrals and the same reprojection applies directly. All five models have approximately $15{,}000$ trainable parameters and receive the same total second-order optimization budget. Table~\ref{tab:cahn-hilliard} reports their field errors and mass drift, while the complete architectures, the TNN mass construction, and the training protocol are given in \ref{app:impl-ch}.

\begin{table}[h]
\centering
\caption{Field accuracy and mass conservation for the Cahn--Hilliard benchmark. Relative $L_2$ and $L_\infty$ errors are measured against the Fourier-spectral reference, while the mass column reports the worst-case drift $\max_t|\mathcal{Q}[f](t)-\mathcal{Q}_0|$ under the 512-node quadrature used by the reprojection. All models have approximately $15{,}000$ trainable parameters and receive the same total second-order optimization budget. The constrained models conserve the quadrature mass to machine precision, and for the TOUCKANs independent rules with $1024$--$8192$ nodes confirm the same precision for the continuum integral (\ref{app:sensitivity}). For comparison on the same problem, the specialized backward-compatible PINN (bc-PINN) method~\cite{mattey2022novel} reports a relative error of $1.86\times10^{-2}$. Bold marks the lowest field errors; all constrained mass drifts are at machine precision.}
\medskip
\begin{tabular}{lcccc}
\hline
Model & Params & rel.\ $L_2$ & $L_\infty$ & $\max_t |\mathcal{Q}[f](t) - \mathcal{Q}_0|$ \\
\hline
MLP (unconstrained) & $14{,}928$ & $2.4\times10^{-4}$ & $\mathbf{5.2\times10^{-4}}$ & $2.3\times10^{-4}$ \\
TNN (constrained)~\cite{wang2022tensor} & $14{,}753$ & $4.1\times10^{-3}$ & $6.7\times10^{-2}$ & $2.9\times10^{-15}$ \\
TOUCKAN (independent) & $14{,}840$ & $\mathbf{1.2\times10^{-4}}$ & $7.3\times10^{-4}$ & $5.2\times10^{-15}$ \\
TOUCKAN (full-head) & $15{,}328$ & $3.6\times10^{-4}$ & $2.2\times10^{-3}$ & $5.6\times10^{-14}$ \\
TOUCKAN (diagonal) & $14{,}912$ & $1.4\times10^{-3}$ & $4.0\times10^{-3}$ & $1.0\times10^{-14}$ \\
\hline
\end{tabular}
\label{tab:cahn-hilliard}
\end{table}

All three TOUCKAN variants conserved mass to machine precision under the exact-mass reprojection, despite their different field errors. The constrained TNN matched this behavior with a drift of $2.9\times10^{-15}$, confirming that the reprojection is not specific to TOUCKANs and applies to any architecture whose mass integral is tractable. For the TOUCKANs, independent quadrature refinements confirmed continuum conservation beyond the 512-node rule used during training, with drift remaining near $5\times10^{-15}$ (\ref{app:sensitivity}, Figure~\ref{fig:ch-continuum-mass}). By contrast, the unconstrained MLP baseline exhibited a mass drift of $2.3\times10^{-4}$, comparable to its field error.

\begin{figure}[H]
    \centering
    \includegraphics[width=0.99\linewidth]{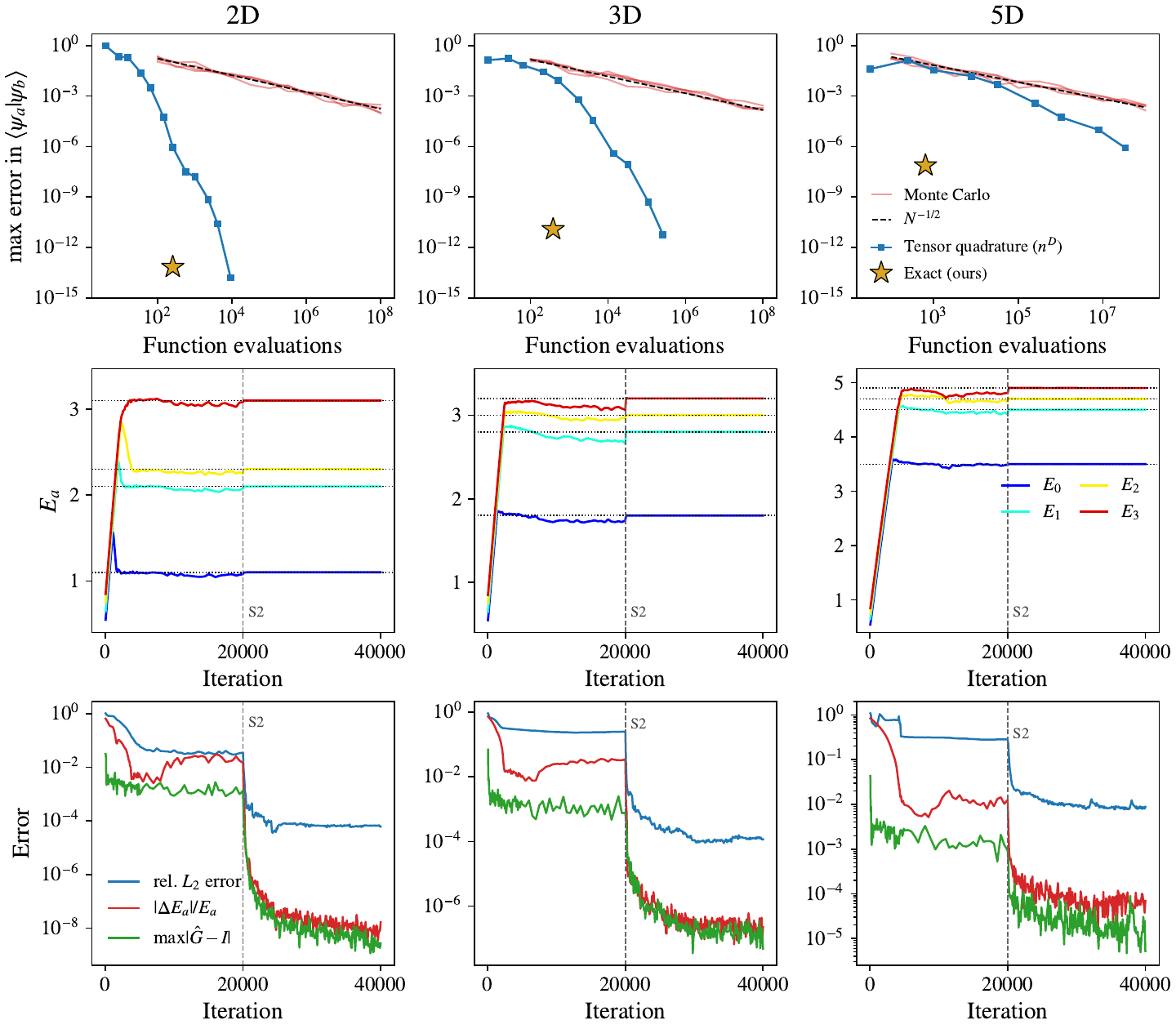}
    \caption{Schr\"odinger eigenvalue problem in $D=2$, $3$, and $5$ (columns). Top row: maximum error in the Gram entries $\langle\psi_a|\psi_b\rangle$ versus number of function evaluations, against converged tensor-product references ($200^{2}$, $120^{3}$, and $48^{5}$ nodes): four independent Monte Carlo runs (red, up to $10^{8}$ samples), tensor-product Gauss--Hermite quadrature with $n^{D}$ nodes (blue), and the exact cross-moment recursion (star), which reaches $1.5\times10^{-14}$, $1.2\times10^{-11}$, and $7.5\times10^{-8}$ with $256$, $384$, and $640$ univariate evaluations, respectively (the $D=5$ value is limited by the convergence of the $48^{5}$-node reference itself). The horizontal axis counts integrand evaluations: the TOUCKAN-based recursion additionally performs $\mathcal{O}(Q^{2}DP^{4})$ cross-moment arithmetic outside this count, and its full evaluation, including that arithmetic, completes in $1$--$3$\,ms per case. Middle row: convergence of the four trainable eigenvalues $E_a$ to the exact spectrum (dotted lines) from a common initialization below the spectrum. Bottom row: training convergence, showing the relative $L_2$ error of the eigenfunctions, the relative eigenvalue error averaged over the four states, and the off-diagonal Gram residual $\max|\widehat{G} - I|$. The dashed line marks the switch from first- to second-order optimization (S2). Together, these results show that deterministic Gram evaluation avoids both the Monte Carlo sampling floor and tensor-product growth while all four eigenpairs converge simultaneously.}
    \label{fig:schrodinger-convergence}
\end{figure}

Hard conservation remained compatible with accurate field reconstruction. The independent TOUCKAN achieved the lowest relative $L_2$ error, $1.2\times10^{-4}$, while the unconstrained MLP baseline achieved the lowest $L_\infty$ error, $5.2\times10^{-4}$. Among the constrained TOUCKANs, field accuracy followed the parameter--accuracy hierarchy observed in Section~\ref{sec:param-efficiency}. Although the TNN admits the same exact-mass construction, its relative $L_2$ error of $4.1\times10^{-3}$ lies more than an order of magnitude above that of the independent TOUCKAN at the same parameter count and optimization budget. Figure~\ref{fig:cahn-hilliard-solution} summarizes the field reconstruction, training histories, and mass conservation.

\subsection{Exact normalization: the Schr\"odinger eigenvalue problem}
\label{sec:schrodinger}

Deterministic moment propagation can also enforce normalization, extending its role from conservation laws to the admissibility of eigenstates. We consider the time-independent Schr\"odinger equation on $\mathbb{R}^{D}$,
\begin{align} \label{eq:schrodinger}
\mathcal{H}\psi
=
-\tfrac{1}{2}\Delta \psi + V(\mathbf{x})\,\psi
=
E\,\psi,
\qquad
V(\mathbf{x}) = \tfrac{1}{2}\,\mathbf{x}^{\top} A\, \mathbf{x},
\end{align}
We seek the lowest $N_s$ eigenpairs $(E_a,\psi_a)$ of a coupled harmonic oscillator with $A=R\,\mathrm{diag}(\omega_1^2,\dots,\omega_D^2)R^{\top}$. The dense rotation $R$ couples all input coordinates, so neither the potential nor its eigenfunctions are separable in the original variables. Undoing this rotation nevertheless provides closed-form reference eigenpairs, allowing precise evaluation of a genuinely multivariate problem. The complete construction is given in \ref{app:impl-schrodinger}.

A neural eigensolver must learn each eigenvalue jointly with its eigenfunction while excluding zero, repeated, and higher-energy solutions. This requires normalized and mutually orthogonal states in addition to a small equation residual. The formulation of Jin et al.~\cite{jin2022physics} enforces these integral constraints through sampled penalties and obtains states sequentially through an orthogonality loss, patience-based solution recognition, and symmetry switching. Its reported experiments are restricted to one-dimensional equations.

TOUCKAN represents all $N_s$ states with a single network in the hard-normalized form $\psi_a=f_a\sqrt{\rho}/\sqrt{G_{aa}}$ from \eqref{eq:tkan-born-density}. The Gaussian density $\rho$ captures the decay of the bound states, while the cross-moment recursion evaluates the complete Gram matrix $G_{ab}=\mathbb{E}_{\rho}[f_af_b]$ deterministically at every training step. Normalization therefore holds by construction, and orthogonality is imposed through the off-diagonal Gram entries without sampling noise. The Schr\"odinger-equation residual $r_a$ in \eqref{eq:app-schrodinger-residual} constrains each learned pair $(E_a,\psi_a)$. The trace $\operatorname{tr}(G^{-1}K)$ is the summed energy of the learned subspace, and by the Ky Fan principle its minimum is attained by the span of the lowest $N_s$ eigenstates~\cite{fan1949theorem}. All $N_s$ eigenpairs are optimized jointly from the first iteration, without an eigenvalue scan or sequential deflation. Sampling is used only to evaluate $r_a$ and the Hamiltonian matrix $K$ in \eqref{eq:app-schrodinger-K}, whereas every norm and overlap is evaluated deterministically. The derivation of the trace term and the complete loss and training procedure are given in \ref{app:impl-schrodinger}.

Table~\ref{tab:schrodinger} reports the lowest four eigenpairs in $D=2$, $3$, and $5$. Their mean relative eigenvalue errors are $1.7\times10^{-8}$, $1.4\times10^{-7}$, and $6.1\times10^{-5}$, respectively. For comparable one-dimensional eigenvalue problems, including the radial reduction of the hydrogen atom, Jin et al.~\cite{jin2022physics} reported errors of roughly one percent. Within each dimension, the eigenvalue errors are approximately the square of the eigenfunction errors, consistent with the quadratic accuracy of a variational formulation.

The loss and eigenpair errors increase together with dimension, whereas the parameter count and training time grow only mildly. The constraint errors follow a different pattern. Normalization is built into the representation, so the values in Table~\ref{tab:schrodinger} measure agreement with an independent tensor-product quadrature rather than optimization of a normalization penalty. Orthogonality also remains several orders below the eigenfunction error. The deterioration at $D=5$ therefore reflects the approximation and optimization achieved under the tested configuration rather than a failure of deterministic constraint evaluation.

\begin{table}[t]
\centering
\caption{Accuracy of the lowest four eigenpairs of the coupled oscillator in $D=2$, $3$, and $5$. Eigenvalue errors are reported for each state, and the eigenfunction row reports the worst case over the four states. Normalization is evaluated using independent tensor-product Gauss--Hermite rules with $200^2$, $120^3$, and $48^5$ nodes and reports the largest residual over the four states. Orthogonality is the maximum absolute off-diagonal entry of the normalized Gram matrix at the final training iteration, evaluated by the deterministic recursion at $M=128$. The $D=5$ normalization value is limited by the reference quadrature rather than the cross-moment recursion. Training times were measured on one NVIDIA H200 GPU.}
\label{tab:schrodinger}
\begin{tabular}{lccc}
\hline
 & $D=2$ & $D=3$ & $D=5$ \\
\hline
Trainable parameters & $10{,}963$ & $13{,}075$ & $17{,}299$ \\
Relative error $E_0$ & $2.2\times10^{-8}$ & $4.0\times10^{-7}$ & $5.6\times10^{-5}$ \\
Relative error $E_1$ & $2.3\times10^{-8}$ & $1.3\times10^{-8}$ & $2.6\times10^{-5}$ \\
Relative error $E_2$ & $7.7\times10^{-9}$ & $5.6\times10^{-8}$ & $2.8\times10^{-5}$ \\
Relative error $E_3$ & $1.4\times10^{-8}$ & $8.1\times10^{-8}$ & $1.3\times10^{-4}$ \\
Eigenfunction rel.\ $L_2$ error & $2.2\times10^{-4}$ & $2.5\times10^{-4}$ & $1.6\times10^{-2}$ \\
Normalization, $\max_a |\langle\psi_a|\psi_a\rangle - 1|$ & $2.1\times10^{-14}$ & $5.5\times10^{-12}$ & $2.3\times10^{-7}$ \\
Orthogonality, $\max_{a\neq b} |\widehat{G}_{ab}|$ & $2.7\times10^{-9}$ & $4.9\times10^{-8}$ & $5.2\times10^{-6}$ \\
Final training loss & $7.6\times10^{-12}$ & $2.7\times10^{-8}$ & $1.6\times10^{-3}$ \\
Training time & $909$\,s & $1307$\,s & $2191$\,s \\
\hline
\end{tabular}
\end{table}

Figure~\ref{fig:schrodinger-convergence} directly compares the cost of evaluating the Gram matrix. Monte Carlo remains near $10^{-4}$ error after $10^8$ samples, while tensor-product quadrature converges spectrally but grows as $n^D$. The cross-moment recursion instead reaches the reference accuracy using $128$ univariate nodes per dimension, with complete evaluations taking $1$--$3$\,ms. Deterministic Gram evaluation therefore removes the sampling floor from normalization and orthogonality without incurring tensor-product growth.

The same figure shows all four trainable eigenvalues converging to the exact spectrum after joint optimization from the first iteration, without an eigenvalue scan or sequential deflation. The second-order stage produces the final decrease in the eigenfunction, eigenvalue, and orthogonality errors. Figure~\ref{fig:schrodinger-states} confirms that the four three-dimensional eigenstates recover the corresponding modal structures. The two- and five-dimensional fields are shown in Figures~\ref{fig:schrodinger-states-2d} and~\ref{fig:schrodinger-states-5d} of \ref{app:impl-schrodinger}.

\begin{figure}
    \centering
    \includegraphics[width=0.99\linewidth]{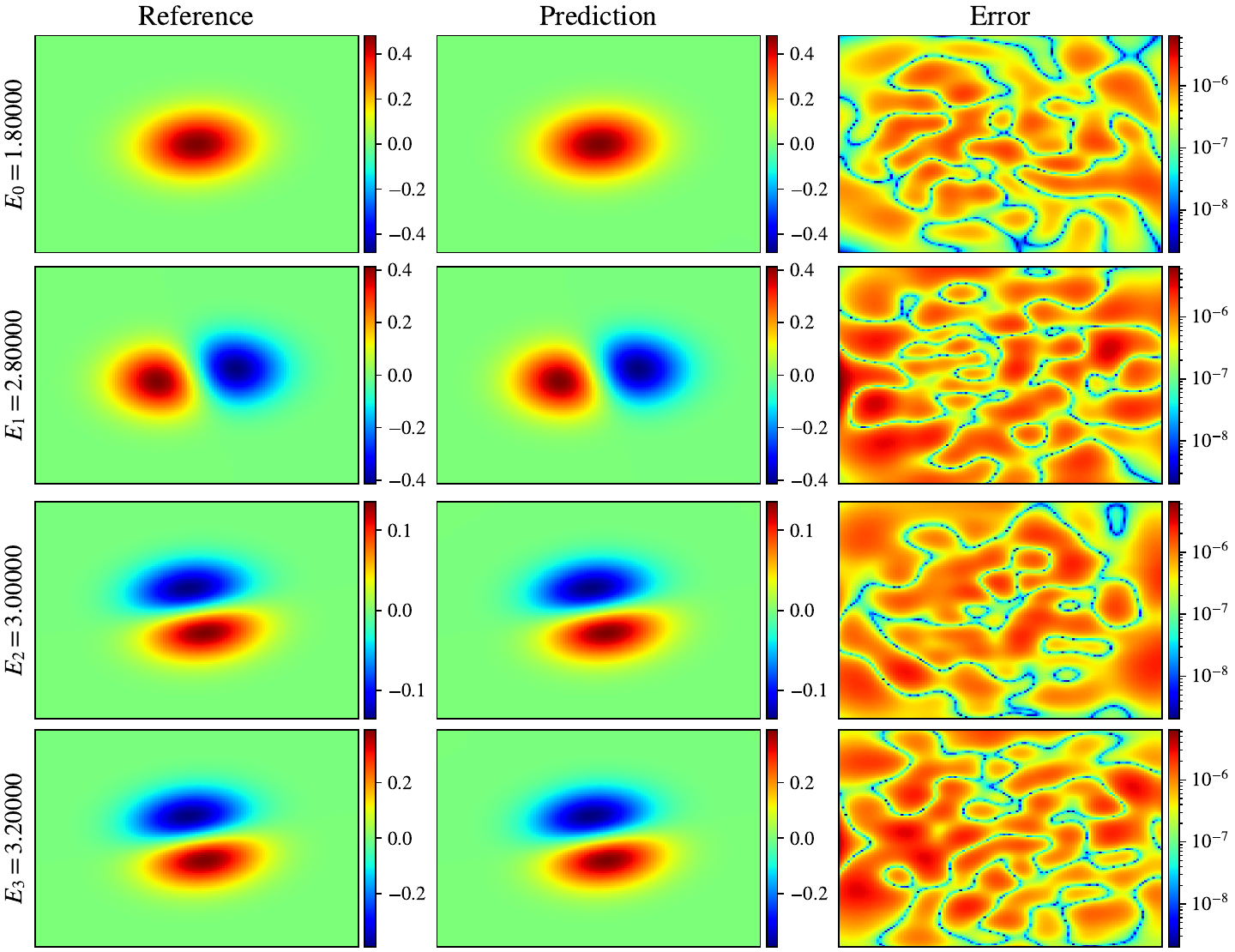}
    \caption{Lowest four eigenstates of the coupled three-dimensional oscillator on the $x_3 = 0$ plane: closed-form reference (left), learned surrogate (center), and absolute error on a logarithmic color scale (right). Each row is labeled by its learned eigenvalue $E_a$. On this slice, the $E_2$ and $E_3$ states have similar projected profiles, while the Gram residuals in Table~\ref{tab:schrodinger} quantify their full-dimensional distinction.}
    \label{fig:schrodinger-states}
\end{figure}

\newpage
\section{Summary}

In this work, we introduced TOUCKANs, a class of KAN architectures that enable tractable evaluation of statistical quantities through a decomposable polynomial structure. Finite-degree polynomial outer functions allow the multinomial theorem to rewrite the architecture as a hierarchical sum--product representation, connecting Kolmogorov--Arnold superposition with tractable probabilistic circuits. TOUCKANs instantiate this principle within the two-block KKAN architecture using either independent per-dimension networks or shared inner networks motivated by the Lorentz--Laczkovich theorems. The shared formulations specialize each dimension through full-head, diagonal, or rank-one couplings, reducing the parameter footprint as the coupling becomes more restrictive. The universality proof applies to the most restricted rank-one subclass and therefore to the complete TOUCKAN hierarchy (Theorem~\ref{thm:tkan-uat} and Remark~\ref{rem:touckan-hierarchy-uat}).

Under factorizing input measures, the resulting sum--product representation reduces expectations and higher moments to one-dimensional integrals over the network leaves. The binomial recursions assemble these quantities without constructing the fully expanded graph. Their algebra is exact in the evaluated leaf moments, while the continuum error is limited to the corresponding one-dimensional quadrature. Quantities usually estimated through sampling can therefore be evaluated deterministically at any stage of training or inference.

Across benchmarks with uniform and Gaussian inputs up to $D=50$, the deterministic mean and variance errors remained below the median fixed-budget Monte Carlo errors in every case. Once the leaf quadrature was resolved, the remaining discrepancy followed the fitted-surrogate error rather than the propagation procedure. The one exception is the variance at $D=50$, where the second-moment recursion reaches a floating-point roundoff floor before the surrogate error. The diagonal specialization provided the best tested parameter--accuracy balance for these high-dimensional experiments and kept second-order refinement computationally feasible~\cite{urban2025unveiling}. In the lower-dimensional physics-informed problems, the richer independent and full-head couplings instead attained the lowest field errors. The inner parameterization therefore controls the finite-budget accuracy--cost trade-off, while the moment recursions remain unchanged across the hierarchy.

These benchmarks validate deterministic moment evaluation within an already-trained surrogate. Repeated moments can then be computed without new true-function evaluations, including when only fixed observations or archived simulation data are available.

Deterministic moment propagation also allows integral quantities to be enforced during training rather than approximated through sampled constraints. For the Cahn--Hilliard equation, the exact-mass reprojection conserved mass to machine precision, while the independent TOUCKAN attained a relative $L_2$ error of $1.2\times10^{-4}$, more than two orders of magnitude below the PINN result reported on the same problem~\cite{mattey2022novel}. For the coupled Schr\"odinger problem, deterministic normalization and orthogonality, combined with the Ky Fan trace objective, enabled joint computation of the lowest four eigenpairs in $D=2$, $3$, and $5$. These results extend deterministic moment propagation from post-training inference to structural conservation and normalization during training.

The branch count grows linearly with $D$, and the moment recursions have the polynomial costs derived in Section~\ref{sec:method}. In practice, the tighter limitation is the coverage of the joint input space during training. Tensorized or separable architectures such as TNNs and SPINNs can form a Cartesian grid of $n^D$ joint locations from $n$ coordinate values along each axis~\cite{wang2022tensor,cho2024separable}. Thus, $100$ values per coordinate span $10^6$ joint locations in three dimensions. TOUCKAN training instead evaluates complete coordinate tuples, so $N$ collocation samples cover $N$ joint locations and matching the same three-dimensional grid requires $10^6$ pointwise samples. Although moment propagation factorizes the polynomial expansion, the implemented forward surrogate sums the coordinate branches before evaluating the outer polynomial and does not retain the axis-wise Cartesian factorization used during TNN or SPINN training. Training coverage can therefore become the practical bottleneck even when moment propagation remains tractable. The current experiments show that sparse pointwise samples suffice for the smooth targets tested through $D=50$, while targets with richer joint structure may require substantially larger training sets.

The present formulation assumes factorizing input measures, making correlated inputs a primary extension. Variable transformations, tractable conditioning, and marginalization could broaden the admissible distributions and support variance-based sensitivity analysis. Further work should address degenerate eigenspaces, architectures that scale to higher dimensions, and broader classes of decomposable neural networks. Together, these directions extend the connection between KAN-based scientific machine learning and tractable probabilistic circuits established here.

\section*{Data availability}
To support reproducibility, the source code for our implementation and the associated data will be released publicly upon acceptance of the manuscript.

\section*{Acknowledgements}
S.R. acknowledges support from the Austrian Science Fund (FWF) through the Erwin Schrödinger Fellowship [Grant DOI:  10.55776/J4774]. 
The authors acknowledge the support of the NIH grant R01AT012312, MURI/AFOSR FA9550-20-1-0358 project, the DOE-MMICS SEA-CROGS DE-SC0023191 award, and the ONR Vannevar Bush Faculty Fellowship (N00014-22-1-2795). L. Wang acknowledges the support from Singapore MOE AcRF Tier 2 Grant: MOE-T2EP20224-0012.

\bibliographystyle{elsarticle-num}
\bibliography{reference}
\newpage
\appendix

\section{Background}
\label{app:background}
\subsection{Kolmogorov--Arnold Networks}
\label{app:background-kan}

KANs are neural architectures motivated by the Kolmogorov--Arnold representation theorem~\cite{kolmogorov1963representation}. The theorem states that any continuous multivariate function on a hypercube can be represented through superpositions of continuous univariate functions. In contrast to conventional multilayer perceptrons, where nonlinear activation functions are associated with nodes, KANs place learnable univariate nonlinear functions on the edges of the computational graph~\cite{liu2025kan}.

In practical deep KAN architectures \cite{liu2025kan}, the representation is implemented layer-wise as
\begin{align} \label{eq:DeepKAN}
z_{n}^{(\ell+1)}
=
\sum_{{n'}=1}^{N_\ell}
\phi_{n,n'}^{(\ell)}
\!\big(
z_{n'}^{(\ell)}
\big),
\qquad
z^{(0)} = \mathbf{x},
\end{align}
where the functions $\phi_{n,n'}^{(\ell)} : \mathbb{R} \rightarrow \mathbb{R}$ are learnable univariate transformations associated with edges between layers. Here, $\ell=0,\dots,L$ indexes layers, and $n$ and $n'=1,\dots,N_\ell$ index neurons in layers $\ell+1$ and $\ell$, respectively. This stacked construction departs from the theorem that motivates it: the representation \eqref{eq:def-KAST} is a two-block structure, with a single level of inner functions and a single level of outer functions, whereas deep KANs compose many such layers into a deeply nested architecture. The departure also carries practical costs. The original KANs parameterize the edge functions with learnable B-splines, which are computationally expensive, and their performance degrades as the parameter count grows. Subsequent variants replace the splines with more efficient bases such as Chebyshev polynomials, whose recursive formulation restores stability in single precision \cite{shukla2024comprehensive}, or radial basis functions; broad reviews of these variants are given in \cite{noorizadegan2026practitioner,faroughi2026kolmogorov}, and necessary and sufficient conditions for the universality of KAN architectures are established in \cite{ismailov2026necessary}.

K\r{u}rkov\'a--Kolmogorov--Arnold networks (KKANs) \cite{toscano2025kurkovakans} were introduced to return to the two-block structure of the theorem. Exact Kolmogorov representations are unsuited to learning, since the inner functions of the constructive proofs are highly non-smooth and carry no adjustable parameters, although explicit constructions of approximate superpositions with $C^2$-smooth inner functions have recently been obtained \cite{song2025explicit}; following K\r{u}rkov\'a \cite{kurkova1992kolmogorov}, KKANs sacrifice exactness and approximate both levels of the superposition with smooth learnable components, using robust multilayer perceptrons as inner functions and linear combinations of basis functions as outer functions. The resulting architecture is a universal approximator regardless of the choice of basis. Our construction builds on KKANs; the structural constraints that make statistical inference tractable are introduced in Section~\ref{sec:tractable-architecture}. These structural constraints lead to tractable inference with KANs generally; and work well in practice with KKANs specifically.

\subsection{Sum--Product Networks and Tractable Inference}
\label{app:background-circuits}

Sum--Product Networks (SPNs), more generally referred to as probabilistic circuits or tractable circuits, are deep computational architectures that represent multivariate functions through recursive compositions of sums and products over structured computational graphs. Originally introduced for tractable probabilistic inference, SPNs represent probability distributions through hierarchical decompositions into mixtures (sum nodes) and factorizations (product nodes). More generally,  circuits provide a unifying framework for tractable probabilistic models \cite{choi2020Probabilistic,vergari2021Compositional} and structured function representations \cite{wang2024compositional,de2021compilation}.

Let $\mathcal{C} \equiv (\mathcal{G}, \mathbf{W})$ denote a circuit with computational graph \(\mathcal{G}\) and parameters \(\mathbf{W}\), representing a function $\mathcal{C}(\mathbf{x})$ over input variables $\mathbf{x} = \{x_1,\ldots,x_D\}$.
Here, \(\mathcal{G}\) is an acyclic directed graph with three types of nodes: input nodes, sum nodes, and product nodes. Let \(\mathrm{in}(c)\) denote the set of children of some node \(c \in \mathcal{G}\). An input node \(c\) represents a parameterized function
\begin{align}
g_c(\mathbf{x}_c)
\end{align}
over a subset of variables \(\mathbf{x}_c \subseteq \mathbf{x}\), referred to as its scope. Internal nodes recursively combine their children through weighted sums or products. A sum node computes
\begin{align}
\sum_{c' \in \mathrm{in}(c)}
w_{c,c'} f_{c'}(\mathbf{x}_{c'}),
\end{align}
whereas a product node computes
\begin{align}
\prod_{c' \in \mathrm{in}(c)}
f_{c'}(\mathbf{x}_{c'}).
\end{align}
The scope of an internal node is recursively defined as
\begin{align}
\mathbf{x}_c
=
\bigcup_{c' \in \mathrm{in}(c)}
\mathbf{x}_{c'}.
\end{align}
Probabilistic circuits can be interpreted as neural architectures with additive and multiplicative interactions arranged hierarchically through depth, where expressive nonlinearities are localized at the input functions while multiplicative structure governs variable interactions. It follows immediately from this composition that the expressiveness of SPNs multiplies with every product node, and hence scales exponentially with depth \cite{delalleau2011shallow}. 

Formally, decomposability requires
\begin{align}
\mathbf{x}_{c'} \cap \mathbf{x}_{c''} = \emptyset
\end{align}
for all distinct children \(c',c'' \in \mathrm{in}(c)\) of every product node \(c\). In many applications, one additionally assumes structured-decomposability, meaning that product nodes with identical scope decompose variables in the same manner throughout the network.

Often, circuits are additionally assumed to satisfy smoothness, meaning that all children of a sum node possess identical scope \cite{choi2020Probabilistic}. While smoothness is not required for tractable integration itself, it ensures that the circuit defines a valid probabilistic model with well-defined marginalization and conditioning operations.

In the next section, we show that polynomial outer nonlinearities induce decomposable multiplicative structure in KANs by transforming sums of univariate inner functions into recursively separable product expansions.

\section{Tractable Inference with Kolmogorov--Arnold Networks}
\label{app:tractability-details}

\subsection{Building intuition with shallow minimal-order networks}

For the quadratic expansion \eqref{eq:expansion-intuition}, assuming factorizing integration measures, it follows that
\begin{align}
    \int f(\mathbf{x}) \mathrm{d}V_{\mathbf{x}} = 
\sum_{q=1}^{2D+1}
\left(
\sum_{d=1}^{D}
\int\phi_{q,d}^2(x_d) \mathrm{d}x_d
+
2\sum_{d<d'}\left(\int
\phi_{q,d}(x_d) \mathrm{d}x_d \right) \left(\int \phi_{q,d'}(x_{d'}) \mathrm{d}x_{d'}\right)
\right)\;.
\end{align}
Consequently, the Kolmogorov--Arnold representation acquires a decomposable multiplicative structure analogous to that appearing in sum--product networks and probabilistic circuits, and the $D$-dimensional integral decomposes into $(2D+1)\big(D+\tfrac{D(D-1)}{2}\big)$ algebraic terms, assembled from $2D$ unique one-dimensional leaf moments per branch (the first and second moment of each leaf, with the cross terms reusing the first moments).

\subsection{Higher orders}
\label{app:higher-orders}

Consider an outer activation function as polynomial of order $P$, i.e. 
\begin{align} \label{eq:higher-order-activation}
    \Phi_q(t) := \sum_{p=0}^P c_p t^p
\end{align}
where $c_p$ are coefficients and the order and the summation index can depend on $q$. For now, consider only the leading order. By the multinomial theorem we have
\begin{align} \label{eq:multinomial-theorem}
\left(
\sum_{d=1}^{m}
z_{q,d}
\right)^P
=
\sum_{\substack{k_{q,d} \geq 0 \\ \sum_{d=1}^{m} k_{q,d}=P}}
\binom{P}{k_{q,1},\ldots,k_{q,m}}
\prod_{d=1}^{m}
z_{q,d}^{k_{q,d}}.
\end{align}
Identifying $z_d \equiv \phi_{q,d}(x_d)$, each term again factorizes into products over disjoint variables, yielding decomposable structure and tractable integration under factorizing input measures. 

If the outer activation is a quadratic polynomial, then we get products over pairs of univariate inner activation functions. If the outer activation is a cubic polynomial, then we get products over triples of univariate inner activation functions. More generally, we get sums over products $\prod_{d \in I}
\big(\phi_{q,d}(x_d)\big)^{k_{q,d}}$, where $I$ is a node-specific index set (analogous to the children) and the cardinality of $I$ is the order of the polynomial defining the outer activation function. The exponent $k_{q,d}$ accounts for duplicatives, e.g. $k_{q,d}=2$ for the squared terms in the expanded form \eqref{eq:expansion-intuition} (right hand side) and $k_{q,d}=1$ for the product of pairs for which $d'\neq d$.
If the input probability density function factorizes as
\begin{align}
\rho(\mathbf{x})
=
\prod_{d=1}^{D}
\rho(x_d),
\end{align}
then the expectations over these products decompose as
\begin{align} \label{eq:distribution}
\int
\prod_{d \in I} \big(
\phi_{q,d}  (x_d)\big)^{k_{q,d}}
\, \rho(\mathbf{x})
\, \mathrm{d}V_{\mathbf{x}}
=
\prod_{d \in I}
\int \left(
\phi_{q,d}(x_d)\right)^{k_{q,d}}\rho(x_d)\,\mathrm{d}x_d.
\end{align}
Thus far, we have only considered the leading order $P$. The process is repeated analogously for the trailing orders $P-1, P-2,\dots$, and distributivity of the integral with summation over the orders $p$ in \eqref{eq:higher-order-activation}, $\sum_p^P$, is trivial.

In our implementation based on Kurkova-KANs, the inner activation functions are univariate multilayer perceptrons themselves, hence the one-dimensional integrals are computed numerically. In this case, inference remains tractable but becomes near-exact, rather than exact, due to the numerical error of the one-dimensional integrals. We will see in the experiments that this is a good compromise trading inference tractability against expressivity of the architecture.
We note that these low-dimensional integrals, depending on the inference task, usually need to be computed only once. 

\subsection{Deep layers}
\label{app:tractability-deep}

This principle immediately extends also to multiple layers. Consider a nested application of a polynomial outer activation as in a deep KAN as in \eqref{eq:DeepKAN}. Define the outer activation function as in \eqref{eq:higher-order-activation}, however now the summation index $p$ additionally depends on the layer index $\ell$. We refrain from overloading the notation further via e.g. $p_{q,\ell}$ and instead demonstrate the principle. For simplicity, we first discuss the leading order $P$ at layer $\ell$, which we equip with a layer-index, i.e. $P_{\ell}$. Note that the order can depend on the layer and the edge index, but for simplicity we assume that all edges in a given order are a polynomial of the same order.
Then the leading order of layer $\ell+1$ in \eqref{eq:DeepKAN} becomes 
\begin{align}
\sum_{n_\ell=1}^{N_\ell}
\!\left(
z_{n_\ell}^{(\ell)}
\right)^{P_\ell},
\end{align}
and the leading order of layer $\ell+2$ becomes
\begin{align} \label{eq:trailing-orders}
\sum_{n_{\ell+1}=1}^{N_{\ell+1}}
\!\left(
\sum_{n_{\ell}=1}^{N_\ell}
\!\left(
z_{n_{\ell}}^{(\ell)}
\right)^{P_\ell}
\right)^{P_{\ell+1}} =
\sum_{n_{\ell+1}=1}^{N_{\ell+1}}
\sum_{n_{\ell}=1}^{N_\ell}
\!\left(
z_{n_{\ell}}^{(\ell)}
\right)^{P_\ell {P_{\ell+1}}} + \text{trailing orders}.
\end{align}
If we integrate this expression, by linearity of expectation and Fubini's theorem under mild integrability conditions, we can exchange the order of summation and integration,
\begin{align} \label{eq:trailing-orders-integrated}
\mathbb{E}\left[\sum_{n_{\ell+1}=1}^{N_{\ell+1}}
\!\left(
\sum_{n_{\ell}=1}^{N_\ell}
\!\left(
z_{n_{\ell}}^{(\ell)}
\right)^{P_\ell}
\right)^{P_{\ell+1}}\right] =
\sum_{n_{\ell+1}=1}^{N_{\ell+1}}
\sum_{n_{\ell}=1}^{N_\ell}
\!\mathbb{E}\left[\left(
z_{n_{\ell}}^{(\ell)}
\right)^{P_\ell{P_{\ell+1}}}\right] + \text{trailing orders}.
\end{align}
We can then apply again the multinomial theorem \eqref{eq:multinomial-theorem} to the leading order above and reduce the expression as discussed in \ref{app:higher-orders}. 
In addition, we can apply the multinomial theorem to all trailing orders in \eqref{eq:trailing-orders} since all orders have  polynomial form.

\subsection{Implementation}
\label{app:tractability-implementation}

Repeating the procedure across all orders and layers decomposes a TOUCKAN into low-dimensional components. We defer the complete formulation for arbitrary orders and layers to \ref{app:moments-proof}. In practice, the KAN computational graph and the repeated binomial convolutions of Section~\ref{sec:moment-recursion} perform this bookkeeping.

We illustrate this informally as follows.
Consider a definition of $f(\mathbf{x})$ via a composition operator $\mathcal{F}_{\mathcal{C}, K}$ of the set of leaves $\{\phi_{q,d}(x_d)\}$, where the operator is defined by the composition implied by the circuit $\mathcal{C}$ and $K$ the set of powers indicating the monomial orders after expansion: 
\begin{align}
f(\mathbf{x}) = \mathcal{F}_{\mathcal{C},K}\big[\{\phi_{q,d}(x_d)\}\big]
\end{align}
Under factorizing priors, we write
\begin{align} \label{eq:composition-operator-commutator}
    \mathbb{E}[f(\mathbf{x})] = \mathbb{E}\left[\mathcal{F}_{\mathcal{C},K}\big[\{\phi_{q,d}(x_d)\}\big]\right] \hat= \mathcal{F}_{\mathcal{C},K}\Big[\Big\{ \mathbb{E}\big[\phi_{q,d}^{k_{q,d}}(x_d)\big]\Big\}\Big].
\end{align}
In other words, the expectation of the composition of leaf functions can be written as the composition of the expectation of (powers of ) the leaf functions. 
Therefore, the circuit implementation replaces each leaf by the expectation of its required powers and then performs a forward pass through the circuit. Leaf replacement is a general property of probabilistic circuits~\cite{choi2020Probabilistic}, while TOUCKANs must additionally track the powers through the binomial recursion of Section~\ref{sec:moment-recursion}. We verify \eqref{eq:composition-operator-commutator} numerically in Section~\ref{sec:tractable-inference}.

\subsection{Computational complexity scaling}
\label{app:complexity-scaling}
For a shallow TOUCKAN with $Q$ outer branches and polynomial degree $P$, the multinomial expansion of each outer activation produces
\[
T(D,P)=\sum_{p=0}^{P}{D+p-1\choose p}={D+P\choose P}
\]
separable monomial terms. Hence the total algebraic cost scales as $Q{D+P\choose P}$, which is $\mathcal O(QD^P)$ for fixed $P$. In the classical KAN width $Q=2D+1$, this gives $\mathcal O(D^{P+1})$. The one-dimensional moment evaluations scale as $\mathcal O(QDPM)$ when $M$ quadrature points are used per leaf moment if distinct leaves are used per branch. For leaves with analytic expressions of the integral, we have $M=1$.
Therefore, direct quadrature with cost $\mathcal O(M^D)$ is replaced by
\[
\mathcal O\!\left(QDPM + Q{D+P\choose P}\right),
\]
which is polynomial in $D$ for fixed polynomial degree. We note that the scaling above applies to the explicitly expanded shallow TOUCKAN representation. In deep TOUCKANs, the relevant complexity measure becomes the size of the induced decomposable computational graph generated through recursive multinomial expansion across layers. Accordingly, tractability refers to polynomial complexity in the size of this expanded representation.

The computational complexity of the proposed moment-propagation procedure depends on the number of decomposed monomial terms induced by the polynomial outer activations. For a polynomial order $P$ and input dimension $D$, the multinomial expansion generates combinatorial sums over separable products, while each resulting integral factorizes into low-dimensional integrals over the individual variables. Most importantly, the inference procedure scales linearly with the number of quadrature points $M$.
In contrast to direct quadrature, whose complexity typically scales exponentially with dimension, the TOUCKAN decomposition replaces high-dimensional integration by compositions of one-dimensional moment evaluations and algebraic operations through the computational graph. Consequently, the dominant computational cost shifts from high-dimensional sampling to the evaluation of the induced low-dimensional moments and the combinatorial expansion itself. In practice, this yields substantial computational savings in high-dimensional uncertainty propagation problems compared to brute-force Monte Carlo or dense quadrature methods.


For a deep TOUCKAN, the relevant complexity is the size of the expanded decomposable computational graph. Let $T_\ell$ denote the number of separable monomial terms required to represent one node at layer $\ell$ after recursive multinomial expansion. For input dimension $D=N_0$, layer widths $N_\ell$, and polynomial degrees $P_\ell$, the first layer satisfies
\[
T_1={D+P_1\choose P_1}.
\]
For deeper layers, each node receives $N_{\ell-1}$ inputs, each of which is itself represented by $T_{\ell-1}$ separable terms. Therefore, a degree-$P_\ell$ polynomial expansion produces the upper bound
\[
T_\ell
\le
{N_{\ell-1}T_{\ell-1}+P_\ell\choose P_\ell},
\qquad \ell=2,\ldots,L.
\]
In the simple case of uniform width $N_\ell=N$ and uniform polynomial degree $P_\ell=P$, this implies
\[
T_L
=
\mathcal O\big(
N^{P+P^2+\cdots+P^{L-1}}
D^{P^L}
\big),
\]
for fixed $P$. Thus, exact moment propagation is polynomial in the input dimension for fixed depth and fixed polynomial degree, but the degree of this polynomial grows exponentially with depth. This quantifies the sense in which deep TOUCKANs are tractable with respect to the expanded decomposable representation; uncontrolled increases in both depth and polynomial order can lead to rapid combinatorial growth.

\section{General tractable moment propagation in TOUCKANs}
\label{app:moments-proof}

In this appendix, we formalize the tractable inference property of TOUCKANs for arbitrary depth and polynomial order.

\subsection{Expanded decomposable representation}

Consider a TOUCKAN represented through a computational graph
$\mathcal C$
with leaves
$\phi_{q,d}(x_d)$
and polynomial outer activations.
After recursive multinomial expansion of all polynomial activations,
the resulting function admits the representation
\begin{align}
f(\mathbf x)
=
\sum_{m=1}^{M}
c_m
\prod_{d=1}^{D}
\phi_{m,d}(x_d)^{k_{m,d}},
\label{eq:expanded-tkan}
\end{align}
where
$c_m \in \mathbb R$
are coefficients,
$k_{m,d}\in\mathbb N_0$
are nonnegative integer exponents,
and each factor depends only on a single variable $x_d$.

Equation \eqref{eq:expanded-tkan} is the fully expanded decomposable representation induced by the recursive multinomial expansion of the TOUCKAN computational graph.

\subsection{Moment propagation operator}

Let
$\rho(\mathbf x)$
be a factorizing probability density of the form
\begin{align}
\rho(\mathbf x)
=
\prod_{d=1}^{D}
\rho_d(x_d).
\label{eq:factorizing-density}
\end{align}
Define the leaf moment 
\begin{align}
\mu_{m,d}^{(k)}
=
\int
\phi_{m,d}(x_d)^k
\rho_d(x_d)
\, \mathrm d x_d.
\label{eq:leaf-moment}
\end{align}

The expectation of the TOUCKAN can then be written as
\begin{align}
\mathbb E[f(\mathbf x)]
=
\sum_{m=1}^{M}
c_m
\prod_{d=1}^{D}
\mu_{m,d}^{(k_{m,d})}.
\label{eq:moment-propagation}
\end{align}

\subsection{Main theorem}

\begin{theorem}[Tractable moment propagation in TOUCKANs]
\label{thm:tkan-tractability}
Consider a KAN with arbitrary depth and polynomial outer activations.
Assume:
\begin{enumerate}
    \item all outer activations are finite-degree polynomials,
    \item all leaf functions are integrable and their powers are integrable,
    \item the input density factorizes as in
    \eqref{eq:factorizing-density}.
\end{enumerate}

Then every statistical moment of the form
\begin{align}
\mathbb E[h(f(\mathbf x))]
=
\int
h(f(\mathbf x))
\rho(\mathbf x)
\, \mathrm dV_{\mathbf x},
\end{align}
For polynomial $h$, this integral reduces exactly to one-dimensional leaf moments of the form \eqref{eq:leaf-moment}. These moments are combined through a finite algebraic composition.
Consequently, the high-dimensional integral requires only low-dimensional moment evaluations followed by a forward pass through the expanded decomposable computational graph.
\end{theorem}

\subsection{Proof}

\begin{proof}

We proceed by induction over the computational graph.

For the base case, the leaves depend only on single variables:
\begin{align}
g_d(x_d)
=
\phi_d(x_d).
\end{align}
Hence their moments are one-dimensional integrals by definition.

Now assume that all nodes up to layer $\ell$
admit decomposable representations of the form
\begin{align}
z^{(\ell)}
=
\sum_m
c_m
\prod_d
\phi_{m,d}(x_d)^{k_{m,d}}.
\end{align}

Consider a node at layer $\ell+1$ generated through a polynomial activation
\begin{align}
\Phi(t)
=
\sum_{p=0}^{P}
a_p t^p.
\end{align}
Substituting the inductive hypothesis yields powers of sums of decomposable monomials.
Applying the multinomial theorem recursively produces again a finite sum of products of univariate functions,
\begin{align}
z^{(\ell+1)}
=
\sum_{m'}
c_{m'}
\prod_d
\phi_{m',d}(x_d)^{k_{m',d}}.
\end{align}
Therefore decomposability is preserved through every layer.

Since the density factorizes,
Fubini's theorem implies
\begin{align}
\int
\prod_d
\phi_{m,d}(x_d)^{k_{m,d}}
\rho(\mathbf x)
\, \mathrm dV_{\mathbf x}
=
\prod_d
\int
\phi_{m,d}(x_d)^{k_{m,d}}
\rho_d(x_d)
\, \mathrm d x_d.
\end{align}
Thus every monomial term factorizes into one-dimensional moments.

Linearity of integration then yields
\begin{align}
\mathbb E[f(\mathbf x)]
=
\sum_m
c_m
\prod_d
\mu_{m,d}^{(k_{m,d})},
\end{align}
which proves the claim.

For a polynomial observable $h$, the composition $h\circ f$ is again a finite polynomial composition of the same leaves, so the preceding argument yields the general claim.

\end{proof}

\section{Universal approximation of TOUCKANs}
\label{app:universality}

This appendix proves universal approximation for the complete TOUCKAN hierarchy by establishing density of its most restricted rank-one subclass. In this subclass, every dimension evaluates the same $Q$ shared profiles and specializes them through scalars $\lambda_d$ that are shared across branches. The argument follows the two-block approximation strategy of K\r{u}rkov\'a~\cite{kurkova1992kolmogorov}, but is anchored at a superposition theorem of Laczkovich~\cite{laczkovich2021superposition} whose inner data have the shared-profile structure of the TOUCKAN inner block and carry the smallest sufficient branch count available for it.

\begin{theorem}[Laczkovich \cite{laczkovich2021superposition}, Theorem 1.1]
\label{thm:laczkovich}
Let $D\ge2$, let $m>(2+\sqrt{2})D$ be an integer, and let $\lambda^{*}_{1},\dots,\lambda^{*}_{D}$ be distinct positive numbers. Then there exist continuous functions $\varphi_1,\dots,\varphi_m\in C(\mathbb{R})$ with the following property: for every bounded $f \in C(\mathbb{R}^{D})$ there is a continuous function $g \in C(\mathbb{R})$ such that
\begin{equation}
\label{eq:app-laczkovich}
f(\mathbf{x}) = \sum_{q=1}^{m} g\Big( \sum_{d=1}^{D} \lambda^{*}_{d}\, \varphi_q(x_d) \Big)
\qquad \text{for every } \mathbf{x} \in \mathbb{R}^{D}.
\end{equation}
\end{theorem}

The constants $\lambda^{*}_{d}$ and the profiles $\varphi_q$ are universal. Only the outer function $g$ depends on $f$. We anchor the argument at Theorem~\ref{thm:laczkovich} because it carries the smallest sufficient branch count among the superposition theorems of \cite{laczkovich2021superposition}. Its constants are shared across branches, so the network constructed below lies in the rank-one subclass $\lambda_{d,q} = \lambda_d$ of the diagonal coupling \eqref{eq:diag-specialization}. Density of the full diagonal class then follows a fortiori. No monotonicity property is needed in the argument below. Theorem~1.2 of \cite{laczkovich2021superposition} additionally provides monotone profiles at the larger branch count discussed in Remark~\ref{rem:tkan-uat}. We now define the set of functions realizable by the TOUCKAN architecture.

\begin{definition}[TOUCKAN approximator set]
\label{def:tkan-class}
Let $\mathcal{A}_{W}$ denote a class of shared inner networks $\boldsymbol{\psi} = (\psi_1, \dots, \psi_Q) : [0,1] \to \mathbb{R}^{Q}$ (in our implementation, basis-embedded MLPs of size $W$) such that $\bigcup_{W} \mathcal{A}_{W}$ is dense in $C([0,1]; \mathbb{R}^{Q})$ under the uniform norm $\|\boldsymbol{\psi} - \boldsymbol{\varphi}\|_{\infty} = \max_{q \le Q} \sup_{x \in [0,1]} |\psi_q(x) - \varphi_q(x)|$. Multi-output perceptrons have this property \cite{cybenko1989approximation,hornik1989multilayer}, and the basis-embedded networks of our implementation inherit the property whenever the embedding is injective on $[0,1]$. For $Q, P, W \in \mathbb{N}$ define
\begin{equation}
\label{eq:app-tkan-class}
\mathbb{T}_{Q,P,W}
=
\Big\{
F(\mathbf{x}) = \sum_{q=1}^{Q} \sum_{p=0}^{P} C_{q,p}\, T_p(\xi_q),
\quad
\xi_q = \frac{1}{D} \sum_{d=1}^{D} \tanh\big( \lambda_{d,q}\, \psi_q(x_d) \big)
\Big\},
\end{equation}
where $\boldsymbol{\psi} \in \mathcal{A}_{W}$, $\lambda_{d,q} \in \mathbb{R}$, and $C_{q,p} \in \mathbb{R}$, exactly as in \eqref{eq:tkan-architecture}--\eqref{eq:diag-specialization}.
\end{definition}

\begin{theorem}[Universal approximation of TOUCKANs]
\label{thm:tkan-uat}
Let $D \ge 2$ and let $Q \ge m_{L}(D)$, where $m_{L}(D)$ is the smallest integer exceeding $(2+\sqrt{2})D$. Then $\bigcup_{P,W} \mathbb{T}_{Q,P,W}$ is dense in $C(E^{D})$ with $E^{D} = [0,1]^{D}$: for every $f \in C(E^{D})$ and every $\varepsilon > 0$ there exist $P, W \in \mathbb{N}$ and $F \in \mathbb{T}_{Q,P,W}$ such that
\begin{equation}
\| f - F \|_{\infty} = \sup_{\mathbf{x} \in E^{D}} | f(\mathbf{x}) - F(\mathbf{x}) | < \varepsilon.
\end{equation}
\end{theorem}

\begin{remark}[Universality across the TOUCKAN hierarchy]
\label{rem:touckan-hierarchy-uat}
Let $\mathcal{T}_{\mathrm{rank}}$, $\mathcal{T}_{\mathrm{diag}}$, $\mathcal{T}_{\mathrm{full}}$, and $\mathcal{T}_{\mathrm{ind}}$ denote the rank-one, diagonal, full-head, and independent TOUCKAN classes, respectively. Their parameterizations give the nested function classes
\begin{equation}
\mathcal{T}_{\mathrm{rank}}
\subseteq
\mathcal{T}_{\mathrm{diag}}
\subseteq
\mathcal{T}_{\mathrm{full}}
\subseteq
\mathcal{T}_{\mathrm{ind}}.
\end{equation}
The first inclusion follows by imposing $\lambda_{d,q}=\lambda_d$. The diagonal construction is recovered from the full-head formulation by restricting each head to a scaled shared map, while an independent inner network can reproduce any shared-trunk construction. The proof of Theorem~\ref{thm:tkan-uat} establishes density of $\mathcal{T}_{\mathrm{rank}}$ in $C(E^D)$. Every class containing it is therefore also dense in $C(E^D)$. Thus, all four TOUCKAN parameterizations are universal approximators, and their differences concern finite-parameter accuracy and efficiency rather than asymptotic expressivity.
\end{remark}

\begin{proof}
Let $f \in C(E^{D})$ and $\varepsilon > 0$. By the Tietze extension theorem, $f$ extends to a bounded continuous function on $\mathbb{R}^{D}$ with the same supremum norm, so Theorem~\ref{thm:laczkovich} applies with $m = m_{L}(D) \le Q$: there are positive constants $\lambda^{*}_{d}$, continuous profiles $\varphi_q \in C(\mathbb{R})$, and $g \in C(\mathbb{R})$ such that, writing
\begin{equation}
s_q(\mathbf{x}) := \sum_{d=1}^{D} \lambda^{*}_{d}\, \varphi_q(x_d),
\end{equation}
we have $f(\mathbf{x}) = \sum_{q=1}^{m} g(s_q(\mathbf{x}))$ on $E^{D}$.

Step 1 (compact ranges). Each $\varphi_q$ is continuous, hence bounded on $[0,1]$: set $A := \max_{q} \max_{x \in [0,1]} |\varphi_q(x)|$ and $\Lambda := \max_{d} \lambda^{*}_{d}$, so that $s_q(\mathbf{x}) \in [-B, B]$ with $B := D \Lambda A$. The outer function $g$ is uniformly continuous on $[-B-1, B+1]$; let $\omega$ denote its modulus of continuity there.

Step 2 (inner approximation). Define the padded target $\boldsymbol{\varphi} := (\varphi_1, \dots, \varphi_m, 0, \dots, 0) \in C([0,1]; \mathbb{R}^{Q})$, whose last $Q - m$ components vanish (these channels receive zero outer coefficients in Step 5). By the density assumption on $\mathcal{A}_W$ there exists $\boldsymbol{\psi} \in \mathcal{A}_{W}$ with $\|\boldsymbol{\psi} - \boldsymbol{\varphi}\|_{\infty} < \delta$ for a $\delta \in (0,1]$ fixed below; in particular $\max_{q \le m} \max_{x \in [0,1]} | \psi_q(x) - \varphi_q(x) | < \delta$ and $|\psi_q| \le A + 1$ on $[0,1]$.

Step 3 (scalars and $\tanh$ linearization). For a scale $\tau \in \big(0, \tfrac{1}{\Lambda(A+1)}\big]$ fixed below, set the trainable scalars to $\lambda_{d,q} := \tau \lambda^{*}_{d}$, independent of $q$, so the constructed network lies in the rank-one subclass $\lambda_{d,q} = \lambda_d$. Each branch argument $u = \lambda_{d,q} \psi_q(x_d)$ then satisfies $|u| \le \tau \Lambda (A+1) \le 1$, and since $|\tanh(u) - u| \le |u|^{3}/3$,
\begin{equation}
\Big| \xi_q - \frac{\tau}{D}\sum_{d=1}^{D} \lambda^{*}_{d} \psi_q(x_d) \Big| \le \frac{\big(\tau \Lambda (A+1)\big)^{3}}{3},
\qquad
\Big| \frac{\tau}{D}\sum_{d=1}^{D} \lambda^{*}_{d} \psi_q(x_d) - \frac{\tau}{D} s_q(\mathbf{x}) \Big| \le \tau \Lambda \delta,
\end{equation}
so that $\big| \xi_q - \tfrac{\tau}{D} s_q(\mathbf{x}) \big| \le \rho := \tfrac{1}{3}\tau^{3} \Lambda^{3} (A+1)^{3} + \tau \Lambda \delta$. Note that $\xi_q \in (-1,1)$ always (the $\tanh$ bounds each branch and the average preserves the bound), and $\tfrac{\tau}{D} s_q(\mathbf{x}) \in [-\tau \Lambda A, \tau \Lambda A] \subset (-1,1)$.

Step 4 (rescaled outer function). Since $g \in C(\mathbb{R})$, define $\tilde{g}(t) := g(D t / \tau)$ for $t \in [-1,1]$, so that $g(s_q(\mathbf{x})) = \tilde{g}\big( \tfrac{\tau}{D} s_q(\mathbf{x}) \big)$. Provided $D\rho/\tau \le 1$, Step 3 gives $| D \xi_q / \tau - s_q(\mathbf{x}) | \le D\rho/\tau \le 1$; since $s_q(\mathbf{x}) \in [-B, B]$, both arguments lie in $[-B-1, B+1]$, where the modulus $\omega$ of Step 1 applies, and therefore
\begin{equation}
\big| \tilde{g}(\xi_q) - g(s_q(\mathbf{x})) \big|
\le \omega\Big( \frac{D \rho}{\tau} \Big)
= \omega\Big( \frac{D \Lambda^{3} (A+1)^{3}}{3}\, \tau^{2} + D \Lambda \delta \Big).
\end{equation}
The argument of $\omega$ tends to zero as $\tau \to 0$ and $\delta \to 0$. Indeed, the factor $\tau$ multiplying $\delta$ in $\rho$ cancels the division by $\tau$. We may therefore first fix $\tau$ and then $\delta$ so that $D\rho/\tau \le 1$ and the right-hand side is below $\varepsilon / (3m)$.

Step 5 (outer approximation). By the Weierstrass approximation theorem there is a polynomial of some degree $P$, written in the Chebyshev basis as $\Phi_q(t) = \sum_{p=0}^{P} C_{q,p} T_p(t)$, with $\max_{t \in [-1,1]} |\Phi_q(t) - \tilde{g}(t)| < \varepsilon/(3m)$; for the unused branches $m < q \le Q$ set $C_{q,p} = 0$.

Step 6 (assembly). The resulting $F \in \mathbb{T}_{Q,P,W}$ satisfies, for every $\mathbf{x} \in E^{D}$,
\begin{equation}
| f(\mathbf{x}) - F(\mathbf{x}) |
\le \sum_{q=1}^{m} \Big[ \big| g(s_q(\mathbf{x})) - \tilde{g}(\xi_q) \big| + \big| \tilde{g}(\xi_q) - \Phi_q(\xi_q) \big| \Big]
< m \Big( \frac{\varepsilon}{3m} + \frac{\varepsilon}{3m} \Big) < \varepsilon,
\end{equation}
where the second term uses $\xi_q \in (-1,1)$. This completes the proof.
\end{proof}

\begin{remark}
\label{rem:tkan-uat}
The branch count required by Theorem~\ref{thm:tkan-uat} is $Q \ge m_{L}(D)$, the smallest integer exceeding $(2+\sqrt{2})D$ and asymptotically about $3.41\,D$. This count is inherited from Theorem~\ref{thm:laczkovich}, and Laczkovich notes that it is probably not sharp~\cite{laczkovich2021superposition}. Our K\r{u}rkov\'a-variant TOUCKAN architecture replaces the single outer function $g$ in \eqref{eq:app-laczkovich} with one Chebyshev expansion $\Phi_q$ per branch. Because this architecture contains the common-outer-function class as a special case, the density statement holds a fortiori. The branch counts used in the experiments were selected for each problem and are reported in \ref{app:implementation}; these finite-budget choices do not test the sharpness of the sufficient bound. Because the scalars constructed in Step 3 are shared across branches, the proof establishes density of the rank-one subclass ($\lambda_{d,q} = \lambda_d$). Universality therefore holds throughout the complete hierarchy, and the preference for the diagonal form in Section~\ref{sec:param-efficiency} rests on finite-budget accuracy rather than asymptotic expressivity. Finally, the same argument applies verbatim to Theorem~1.2 of \cite{laczkovich2021superposition}. Its steps are the Tietze extension, padded vector approximation, small-$\tau$ linearization, and Chebyshev outer approximation. In that theorem, the constants $\lambda^{*}_{d,q}$ carry one value per dimension--branch pair and directly mirror the diagonal coupling \eqref{eq:diag-specialization}, at the larger branch count $m > (2+\sqrt{2})(2D-1)$, asymptotically about $6.83\,D$.
\end{remark}

\begin{remark}
\label{rem:tkan-uat-2} The universal approximation result is inherited from the underlying Kolmogorov--Arnold representation theorems and therefore applies to continuous functions. This setting is appropriate for the scientific machine learning applications considered here, where the quantities of interest are typically continuous or smooth fields. It is more restrictive than general machine learning settings that may contain discontinuous decision boundaries. This limitation is a general aspect of Kolmogorov--Arnold representations and is not unique to our KAN variants.
\end{remark}








\section{Moment-propagation algorithms}
\label{app:algorithms}

This appendix derives, in full detail, the exact evaluation and propagation of the first and second moments of a TOUCKAN surrogate. Exactness refers to the propagation: every recursion below is algebraically exact in the leaf moments, which are themselves one-dimensional integrals evaluated by $M$-node Gaussian quadrature. Thus, all computed moments are exact for the quadrature-induced discrete measure and approximate the continuum moments with spectral accuracy in $M$ (cf.\ the error decomposition \eqref{eq:tkan-error-decomposition}). \ref{app:alg-mean} derives the first moment and concludes with Algorithm~\ref{alg:tkan-expectation}. \ref{app:alg-variance} derives the second moment and the variance and concludes with Algorithm~\ref{alg:tkan-variance}. Throughout, the input follows the factorizing density $\rho(\mathbf{x}) = \prod_{d=1}^{D}\rho_d(x_d)$, so the coordinates $x_1,\dots,x_D$ are mutually independent, and the surrogate is the TOUCKAN
\begin{align} \label{eq:app-tkan-setting}
f(\mathbf{x})
=
\sum_{q=1}^{Q} \Phi_q(\xi_q),
\qquad
\Phi_q(\xi)
=
\sum_{p=0}^{P} C_{q,p}\, T_p(\xi),
\qquad
\xi_q
=
\sum_{d=1}^{D} \widetilde{\Psi}_{d,q}(x_d),
\end{align}
where $\widetilde{\Psi}_{d,q} = \Psi_{d,q}/D$ are the normalized branches and $T_p(\xi) = \sum_{j=0}^{p} a_{p,j}\,\xi^{j}$ is the Chebyshev polynomial of degree $p$ in monomial form, with fixed integer coefficients $a_{p,j}$.

\subsection{Exact first moment}
\label{app:alg-mean}

\paragraph{Reduction to scalar moments} The outer blocks are polynomial in the combination variables, so the first moment of the network reduces to scalar moments of each $\xi_q$. Applying linearity of expectation to \eqref{eq:app-tkan-setting} and substituting the monomial form of $T_p$,
\begin{align} \label{eq:app-mean-reduction}
\mathbb{E}\big[f(\mathbf{x})\big]
=
\sum_{q=1}^{Q} \mathbb{E}\big[\Phi_q(\xi_q)\big]
=
\sum_{q=1}^{Q}\sum_{p=0}^{P} C_{q,p}\,
\mathbb{E}\big[T_p(\xi_q)\big]
=
\sum_{q=1}^{Q}\sum_{p=0}^{P} C_{q,p}
\sum_{j=0}^{p} a_{p,j}\,
\mathbb{E}\big[\xi_q^{\,j}\big],
\end{align}
The remaining task is therefore to compute the moments $m_{q,p} = \mathbb{E}[\xi_q^{\,p}]$ for $p = 0,\dots,P$.

\paragraph{Branch isolation} The computation of these moments requires the individual branches. In practice, however, all branches are computed at the same time: the inner block is a single fused map $\boldsymbol{\Psi}\colon (x_1,\dots,x_D) \mapsto \boldsymbol{\xi} \in \mathbb{R}^{Q}$, and no individual branch is evaluated separately. Because each branch depends only on its own coordinate, the contribution of dimension $d$ can nevertheless be recovered by evaluating the fused map with every other slot held at a fixed reference point, chosen as the origin. Define
\begin{align} \label{eq:tkan-branch-isolation}
\boldsymbol{\varphi}_d(x)
=
\boldsymbol{\Psi}(0,\dots,0,\underbrace{x}_{\text{slot } d},0,\dots,0),
\qquad
\mathbf{K}
=
\boldsymbol{\Psi}(0,\dots,0),
\end{align}
both vectors in $\mathbb{R}^{Q}$. The evaluation $\boldsymbol{\varphi}_d(x)$ contains the branch of dimension $d$ together with the constant contribution of the $D-1$ zeroed branches. Summing over $d$ therefore counts this constant $D$ times instead of once, which yields the structural identity
\begin{align} \label{eq:tkan-isolation-identity}
\boldsymbol{\xi}(\mathbf{x})
=
\sum_{d=1}^{D}\big[\boldsymbol{\varphi}_d(x_d) - \mathbf{K}\big] + \mathbf{K}.
\end{align}
The identity follows from the additive structure of the combination layer alone and involves no approximation. In particular, it holds for any inner block whose branches are univariate, and therefore for all variants used in this paper.

\paragraph{Centering} The recursion is not run on the raw evaluations $\boldsymbol{\varphi}_d$, because they share the common offset $\mathbf{K}$ and this offset destroys the computation in high dimension. The failure is quantified at the end of this subsection. Instead, define the centered branches
\begin{align} \label{eq:tkan-centered-branch}
\boldsymbol{\delta}_d(x)
=
\boldsymbol{\varphi}_d(x) - \mathbf{K},
\qquad
\boldsymbol{\delta}_d(0) = \mathbf{0}.
\end{align}
Componentwise, $\delta_{d,q}(x) = \widetilde{\Psi}_{d,q}(x) - \widetilde{\Psi}_{d,q}(0)$, so each centered branch is the branch itself with its value at the origin removed. With this notation the isolation identity \eqref{eq:tkan-isolation-identity} becomes
\begin{align} \label{eq:tkan-centered-sum}
\boldsymbol{\xi}(\mathbf{x})
=
\mathbf{S}(\mathbf{x}) + \mathbf{K},
\qquad
\mathbf{S}(\mathbf{x})
=
\sum_{d=1}^{D}\boldsymbol{\delta}_d(x_d).
\end{align}
The remaining derivation follows this decomposition. The moments of the centered sum $S_q$ are computed first, and the deterministic shift by $K_q$ is reintroduced at the end.

\paragraph{Centered leaf moments} For each dimension and each branch, define the centered leaf moments
\begin{align} \label{eq:app-centered-leaf}
\widehat{\mu}_{d,q}^{(k)}
=
\mathbb{E}\big[\delta_{d,q}(x_d)^{k}\big]
=
\int \delta_{d,q}(x_d)^{k}\,\rho_d(x_d)\,dx_d,
\qquad k = 0,\dots,P.
\end{align}
Each $\widehat{\mu}_{d,q}^{(k)}$ is a one-dimensional integral. It is evaluated by a Gaussian quadrature rule with $M$ nodes whose weight function matches the marginal $\rho_d$, namely Gauss--Legendre for uniform inputs and Gauss--Hermite for Gaussian inputs. This quadrature is the only discretization in the entire procedure; everything downstream is algebraically exact in the computed leaf moments and is evaluated in double precision. The integrand follows from \eqref{eq:tkan-branch-isolation} and \eqref{eq:tkan-centered-branch}: the fused inner block is evaluated at the $M$ nodes of each dimension, and the baseline $\mathbf{K}$ is subtracted. This stage costs $DM$ evaluations of the inner block and $\mathcal{O}(QDPM)$ operations for the moment sums.

\paragraph{Binomial recursion} The moments of $S_q$ are assembled one dimension at a time. Define the centered partial sums
\begin{align} \label{eq:app-partial-sums}
S_{j,q} = \sum_{d=1}^{j}\delta_{d,q}(x_d),
\qquad
S_{D,q} = S_q.
\end{align}
The partial sum $S_{j-1,q}$ depends only on $(x_1,\dots,x_{j-1})$, while the next branch $\delta_{j,q}(x_j)$ depends only on $x_j$. Under a factorizing density, the two are therefore independent random variables. Raising the sum to the power $p$ with the binomial theorem and using this independence,
\begin{align} \label{eq:app-binomial-derivation}
\mathbb{E}\big[S_{j,q}^{\,p}\big]
&=
\mathbb{E}\Big[\big(S_{j-1,q} + \delta_{j,q}(x_j)\big)^{p}\Big]
=
\sum_{k=0}^{p}\binom{p}{k}\,
\mathbb{E}\big[S_{j-1,q}^{\,p-k}\,\delta_{j,q}^{\,k}(x_j)\big]
\nonumber \\
&=
\sum_{k=0}^{p}\binom{p}{k}\,
\mathbb{E}\big[S_{j-1,q}^{\,p-k}\big]\,
\widehat{\mu}_{j,q}^{(k)}.
\end{align}
The first line uses the binomial theorem and linearity of expectation. The second line uses independence: the expectation of a product of independent variables is the product of their expectations. The result is a dynamic program over the input dimensions. It is initialized at the first dimension, $\mathbb{E}[S_{1,q}^{\,p}] = \widehat{\mu}_{1,q}^{(p)}$, and swept forward for $j = 2,\dots,D$, carrying only the $P+1$ moments $s_{q,p} = \mathbb{E}[S_{j,q}^{\,p}]$ of the current partial sum. The combinatorial growth of a direct multinomial expansion is thereby replaced by repeated binomial convolutions, at a total cost of $\mathcal{O}(QDP^{2})$.

\paragraph{$K$-shift and outer contraction} Because $\mathbf{K}$ is a deterministic constant, the moments of $\xi_q = S_q + K_q$ follow from the centered moments through one further binomial expansion,
\begin{align} \label{eq:tkan-k-shift}
m_{q,p}
=
\mathbb{E}\big[\xi_q^{\,p}\big]
=
\sum_{j=0}^{p}\binom{p}{j}\, K_q^{\,p-j}\;
S_{q,j}.
\end{align}
Substituting the shifted moments into the reduction \eqref{eq:app-mean-reduction} yields the exact first moment,
\begin{align} \label{eq:app-exact-mean}
\mathbb{E}\big[f(\mathbf{x})\big]
=
\sum_{q=1}^{Q}\sum_{p=0}^{P} C_{q,p}
\sum_{j=0}^{p} a_{p,j}\, m_{q,j}.
\end{align}
The shift and the contraction add $\mathcal{O}(QP^{2})$ operations, independent of $D$.

\paragraph{Numerical stability} We evaluate the recursion in a centered form, which keeps every intermediate quantity bounded: each centered increment entering the dynamic programs is bounded by $2/D$, since $|\widetilde{\Psi}_{d,q}| \le 1/D$, and the partial sums remain bounded by $2$. The common offset $\mathbf{K}$ re-enters only through the well-conditioned shift \eqref{eq:tkan-k-shift} and its bivariate analogue \eqref{eq:tkan-k-shift-cross}. This centered evaluation reproduces the exact variance to a relative error of order $10^{-8}$ at every dimension tested through $D = 50$ (Table~\ref{tab:tractable-inference-errors}: $1.1\times10^{-8}$, $9.3\times10^{-9}$, $3.8\times10^{-8}$, $1.7\times10^{-8}$ at $D = 5, 12, 25, 50$). All moment evaluations are performed in double precision. The complete procedure is summarized in Algorithm~\ref{alg:tkan-expectation}.

\begin{algorithm}[!ht]
\caption{Exact first moment of a TOUCKAN surrogate.}
\label{alg:tkan-expectation}
\begin{algorithmic}[1]
\Require fused inner block $\boldsymbol{\Psi}$; outer coefficients $C_{q,p}$; Chebyshev monomial coefficients $a_{p,j}$; quadrature rules $\{(x_d^{(i)}, w_d^{(i)})\}_{i=1}^{M}$ matched to the marginals $\rho_d$
\Ensure $\mathbb{E}[f(\mathbf{x})]$
\State $\mathbf{K} \gets \boldsymbol{\Psi}(0,\dots,0)$ \Comment{baseline \eqref{eq:tkan-branch-isolation}}
\For{$d = 1,\dots,D$;\; $i = 1,\dots,M$} \Comment{isolated, centered branches \eqref{eq:tkan-centered-branch}}
    \State $\boldsymbol{\delta}_d^{(i)} \gets \boldsymbol{\Psi}\big(0,\dots,0,x_d^{(i)},0,\dots,0\big) - \mathbf{K}$ \Comment{node in slot $d$}
\EndFor
\For{$d = 1,\dots,D$;\; $q = 1,\dots,Q$;\; $k = 0,\dots,P$} \Comment{centered leaf moments \eqref{eq:app-centered-leaf}}
    \State $\widehat{\mu}_{d,q}^{(k)} \gets \sum_{i=1}^{M} w_d^{(i)}\, \big(\delta_{d,q}^{(i)}\big)^{k}$ \Comment{1D quadrature, $\mathcal{O}(QDPM)$ total}
\EndFor
\For{$q = 1,\dots,Q$} \Comment{binomial recursion \eqref{eq:app-binomial-derivation}, $\mathcal{O}(QDP^{2})$ total}
    \State $s_{q,p} \gets \widehat{\mu}_{1,q}^{(p)}$ \textbf{for} $p = 0,\dots,P$ \Comment{initialize with the first dimension}
    \For{$d = 2,\dots,D$}
        \State $s_{q,p} \gets \sum_{k=0}^{p}\binom{p}{k}\, s_{q,p-k}\,\widehat{\mu}_{d,q}^{(k)}$ \textbf{for} $p = P,\dots,0$ \Comment{in place, descending in $p$}
    \EndFor
\EndFor
\State $m_{q,p} \gets \sum_{j=0}^{p}\binom{p}{j}\, K_q^{\,p-j}\, s_{q,j}$ \textbf{for all} $q, p$ \Comment{$K$-shift \eqref{eq:tkan-k-shift}}
\State \Return $\displaystyle \sum_{q=1}^{Q}\sum_{p=0}^{P} C_{q,p}\sum_{j=0}^{p} a_{p,j}\, m_{q,j}$ \Comment{outer contraction \eqref{eq:app-exact-mean}}
\end{algorithmic}
\end{algorithm}

\subsection{Exact second moment and variance}
\label{app:alg-variance}

The variance requires the second moment, since $\mathrm{Var}[f] = \mathbb{E}[f^{2}] - (\mathbb{E}[f])^{2}$ with the first moment computed in \ref{app:alg-mean}. The derivation repeats that of the first moment for pairs of combination variables. Branch isolation and centering carry over unchanged, so the derivation below starts from the centered branches $\boldsymbol{\delta}_d$ of \eqref{eq:tkan-centered-branch} and from the baseline $\mathbf{K}$ of \eqref{eq:tkan-branch-isolation}.

\paragraph{Monomial collapse} For the second moment, the outer blocks are first collapsed into a single polynomial per branch. Substituting the monomial form of $T_p$ into the outer expansion of \eqref{eq:app-tkan-setting} and collecting powers gives
\begin{align} \label{eq:app-monomial-collapse}
\Phi_q(\xi_q)
=
\sum_{j=0}^{P} \beta_{q,j}\, \xi_q^{\,j},
\qquad
\beta_{q,j}
=
\sum_{p=j}^{P} C_{q,p}\, a_{p,j}.
\end{align}
The full surrogate is then an ordinary polynomial in the combination variables, $f(\mathbf{x}) = \sum_{q}\sum_{j}\beta_{q,j}\,\xi_q^{\,j}$.

\paragraph{Reduction to cross-moments} Squaring this polynomial and taking the expectation termwise gives
\begin{align} \label{eq:app-second-moment-reduction}
\mathbb{E}\big[f^{2}(\mathbf{x})\big]
=
\sum_{q,q'=1}^{Q}
\sum_{j,j'=0}^{P}
\beta_{q,j}\,\beta_{q',j'}\;
m_{q,q'}^{(j,j')},
\qquad
m_{q,q'}^{(j,j')}
=
\mathbb{E}\big[\xi_q^{\,j}\,\xi_{q'}^{\,j'}\big].
\end{align}
The cross-moments $m_{q,q'}^{(j,j')}$ do not factor into the moments of \ref{app:alg-mean}. Although the coordinates are independent, the combination variables $\xi_q$ and $\xi_{q'}$ are built from the same coordinates and are therefore dependent random variables. The recursion must consequently track pairs.

\paragraph{Joint centered leaf moments} For each dimension and each pair of branches, define
\begin{align} \label{eq:app-joint-centered-leaf}
\widehat{\nu}_{d,q,q'}^{(k,k')}
=
\mathbb{E}\big[\delta_{d,q}(x_d)^{k}\,\delta_{d,q'}(x_d)^{k'}\big],
\qquad k, k' = 0,\dots,P.
\end{align}
Both factors depend on the same scalar coordinate, so each $\widehat{\nu}_{d,q,q'}^{(k,k')}$ remains a one-dimensional integral. It is evaluated with the same measure-matched quadrature and the same centered branch evaluations as the marginal leaf moments \eqref{eq:app-centered-leaf}, so no new inner-block evaluations are needed. The first moments are recovered as the special case $\widehat{\nu}_{d,q,q'}^{(k,0)} = \widehat{\mu}_{d,q}^{(k)}$, and the symmetry under $(q,k) \leftrightarrow (q',k')$ restricts the computation to the pairs $q \le q'$. The cost is $\mathcal{O}(Q^{2}DP^{2}M)$.

\paragraph{Bivariate recursion} Fix a pair $(q,q')$ and recall the centered partial sums \eqref{eq:app-partial-sums}. The pair $\big(S_{j-1,q},\, S_{j-1,q'}\big)$ depends only on $(x_1,\dots,x_{j-1})$, while the branch pair $\big(\delta_{j,q}(x_j),\, \delta_{j,q'}(x_j)\big)$ depends only on $x_j$. Under a factorizing density, the two groups are therefore independent of each other, even though the variables within each group are dependent. This is the reason the recursion must carry joint quantities rather than only univariate ones. Expanding both powers with the binomial theorem and factorizing the expectation across the two groups,
\begin{align} \label{eq:app-cross-derivation}
\mathbb{E}\big[S_{j,q}^{\,p}\, S_{j,q'}^{\,p'}\big]
&=
\mathbb{E}\Big[\big(S_{j-1,q} + \delta_{j,q}(x_j)\big)^{p}\,
\big(S_{j-1,q'} + \delta_{j,q'}(x_j)\big)^{p'}\Big]
\nonumber \\
&=
\sum_{a=0}^{p}\sum_{b=0}^{p'}
\binom{p}{a}\binom{p'}{b}\,
\mathbb{E}\big[S_{j-1,q}^{\,p-a}\, S_{j-1,q'}^{\,p'-b}\big]\,
\widehat{\nu}_{j,q,q'}^{(a,b)},
\end{align}
Setting $p' = 0$ recovers the marginal recursion \eqref{eq:app-binomial-derivation}. The recursion is initialized with $\mathbb{E}[S_{1,q}^{\,p} S_{1,q'}^{\,p'}] = \widehat{\nu}_{1,q,q'}^{(p,p')}$ and swept forward for $j = 2,\dots,D$, carrying the $(P+1)^{2}$ joint moments $c_{q,q'}^{(p,p')} = \mathbb{E}[S_{j,q}^{\,p} S_{j,q'}^{\,p'}]$ of the current partial-sum pair, at a total cost of $\mathcal{O}(Q^{2}DP^{4})$.

\paragraph{Bivariate $K$-shift} Because $K_q$ and $K_{q'}$ are deterministic constants, the cross-moments of $(\xi_q, \xi_{q'}) = (S_q + K_q,\, S_{q'} + K_{q'})$ follow from the centered joint moments by expanding both binomials,
\begin{align} \label{eq:tkan-k-shift-cross}
m_{q,q'}^{(p,p')}
=
\mathbb{E}\big[\xi_q^{\,p}\,\xi_{q'}^{\,p'}\big]
=
\sum_{a=0}^{p}\sum_{b=0}^{p'}
\binom{p}{a}\binom{p'}{b}\,
K_q^{\,p-a}\, K_{q'}^{\,p'-b}\;
c_{q,q'}^{(a,b)}.
\end{align}
The shift adds $\mathcal{O}(Q^{2}P^{4})$ operations, independent of $D$.

\paragraph{Assembly} Substituting the shifted cross-moments into the reduction \eqref{eq:app-second-moment-reduction} gives the second moment. Subtracting the squared first moment \eqref{eq:app-exact-mean} gives the variance,
\begin{align} \label{eq:app-exact-variance}
\mathrm{Var}\big[f(\mathbf{x})\big]
=
\sum_{q,q'=1}^{Q}
\sum_{j,j'=0}^{P}
\beta_{q,j}\,\beta_{q',j'}\,
m_{q,q'}^{(j,j')}
-
\Big(\mathbb{E}\big[f(\mathbf{x})\big]\Big)^{2}.
\end{align}
The numerical-stability considerations of \ref{app:alg-mean} apply unchanged. The complete procedure is summarized in Algorithm~\ref{alg:tkan-variance}.

\begin{algorithm}[!ht]
\caption{Exact variance of a TOUCKAN surrogate.}
\label{alg:tkan-variance}
\begin{algorithmic}[1]
\Require fused inner block $\boldsymbol{\Psi}$; outer coefficients $C_{q,p}$; Chebyshev monomial coefficients $a_{p,j}$; quadrature rules $\{(x_d^{(i)}, w_d^{(i)})\}_{i=1}^{M}$ matched to the marginals $\rho_d$; baseline $\mathbf{K}$, centered branch evaluations $\boldsymbol{\delta}_d^{(i)}$, and mean $\mathbb{E}[f]$ from Algorithm~\ref{alg:tkan-expectation}
\Ensure $\mathrm{Var}[f(\mathbf{x})]$
\State $\beta_{q,j} \gets \sum_{p=j}^{P} C_{q,p}\, a_{p,j}$ \textbf{for all} $q, j$ \Comment{monomial collapse \eqref{eq:app-monomial-collapse}}
\For{$d = 1,\dots,D$;\; pairs $q \le q'$;\; $k, k' = 0,\dots,P$} \Comment{joint centered leaf moments \eqref{eq:app-joint-centered-leaf}}
    \State $\widehat{\nu}_{d,q,q'}^{(k,k')} \gets \sum_{i=1}^{M} w_d^{(i)}\, \big(\delta_{d,q}^{(i)}\big)^{k}\, \big(\delta_{d,q'}^{(i)}\big)^{k'}$ \Comment{1D quadrature, $\mathcal{O}(Q^{2}DP^{2}M)$ total}
\EndFor
\For{pairs $q \le q'$} \Comment{bivariate recursion \eqref{eq:app-cross-derivation}, $\mathcal{O}(Q^{2}DP^{4})$ total}
    \State $c_{q,q'}^{(p,p')} \gets \widehat{\nu}_{1,q,q'}^{(p,p')}$ \textbf{for} $p, p' = 0,\dots,P$
    \For{$d = 2,\dots,D$}
        \State $c_{q,q'}^{(p,p')} \gets \sum_{a=0}^{p}\sum_{b=0}^{p'}\binom{p}{a}\binom{p'}{b}\, c_{q,q'}^{(p-a,\,p'-b)}\,\widehat{\nu}_{d,q,q'}^{(a,b)}$ \textbf{for} $p, p' = P,\dots,0$ \Comment{in place, descending}
    \EndFor
    \State $m_{q,q'}^{(p,p')} \gets \sum_{a=0}^{p}\sum_{b=0}^{p'}\binom{p}{a}\binom{p'}{b}\, K_q^{\,p-a}\, K_{q'}^{\,p'-b}\, c_{q,q'}^{(a,b)}$ \Comment{$K$-shift \eqref{eq:tkan-k-shift-cross}}
    \State $m_{q',q}^{(p',p)} \gets m_{q,q'}^{(p,p')}$ \Comment{remaining pairs by symmetry}
\EndFor
\State $\mathbb{E}[f^{2}] \gets \sum_{q,q'=1}^{Q}\sum_{j,j'=0}^{P} \beta_{q,j}\,\beta_{q',j'}\, m_{q,q'}^{(j,j')}$ \Comment{assemble \eqref{eq:app-second-moment-reduction}}
\State \Return $\mathbb{E}[f^{2}] - \big(\mathbb{E}[f]\big)^{2}$ \Comment{exact variance \eqref{eq:app-exact-variance}}
\end{algorithmic}
\end{algorithm}

\FloatBarrier

\section{Sensitivity analysis}
\label{app:sensitivity}

Every expectation reported in this work is assembled from one-dimensional leaf moments evaluated by $M$-node Gaussian quadrature, so the order $M$ is the single accuracy parameter of the inference stage (cf.\ the error decomposition \eqref{eq:tkan-error-decomposition}). This appendix measures the sensitivity of the computed moments to $M$ for one representative of each problem family: the L\"uthen benchmark at $D=50$ (uniform inputs, Gauss--Legendre) and the Schr\"odinger Gram matrix at $D=5$ (Gaussian inputs, probabilists' Gauss--Hermite). We reuse the trained checkpoints of Section~\ref{sec:experiments} without any retraining. Errors are measured against true references rather than against self-convergence: the exact closed-form mean and variance of the L\"uthen function, and an independent $48^{5}$-node tensor-product Gauss--Hermite integration of the trained surrogate (roughly $2.5\times10^{8}$ nodes) for the Schr\"odinger Gram matrix.

\begin{figure}[htbp]
    \centering
    \includegraphics[width=0.99\linewidth]{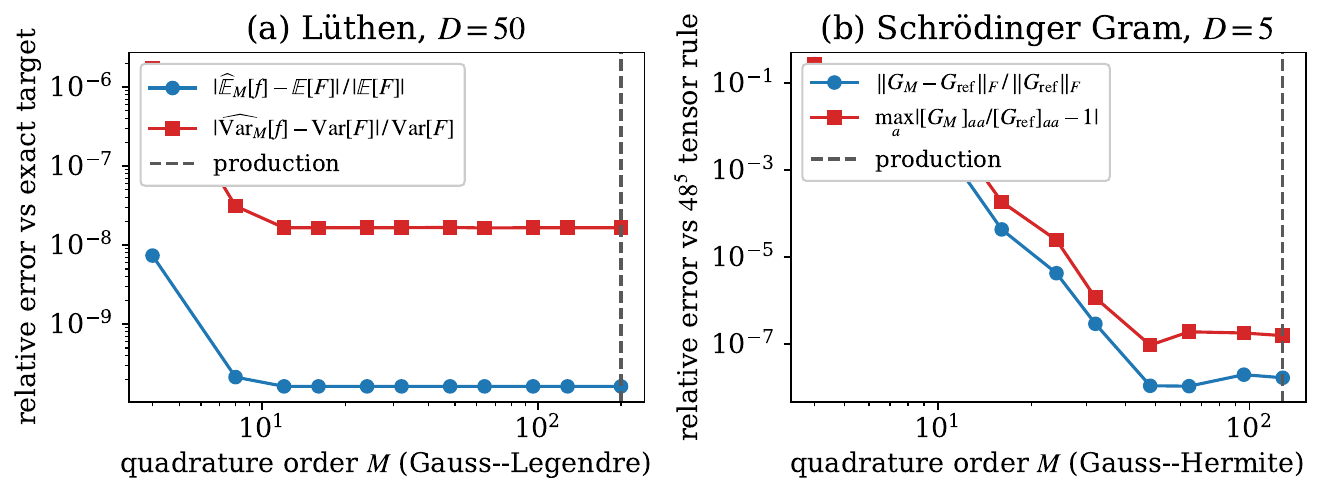}
    \caption{Sensitivity of the computed moments to the one-dimensional quadrature order $M$. (a) L\"uthen function at $D=50$: relative error of the recursion mean and variance against the exact closed-form references. (b) Schr\"odinger Gram matrix at $D=5$: relative error against an independent $48^{5}$ tensor-product Gauss--Hermite integration of the same surrogate. Dashed vertical lines mark the production orders ($M=200$ and $M=128$) used for the examples in the main paper.}
    \label{fig:quad-sensitivity}
\end{figure}

As shown in Fig.~\ref{fig:quad-sensitivity}, all quantities governed by quadrature truncation converge spectrally, as predicted. At $D=50$ the L\"uthen mean and variance errors decay rapidly and saturate from $M\approx12$ onward at $1.63\times10^{-10}$ and $1.66\times10^{-8}$ (the values reported in Table~\ref{tab:tractable-inference-errors}). Beyond this point, the first (training) term of \eqref{eq:tkan-error-decomposition} dominates and additional nodes have no effect. The Schr\"odinger Gram error decays from $2.7\times10^{-1}$ at $M=4$ to a floor near $10^{-8}$ for $M\ge48$. The production order $M=128$ carries a true relative error of $1.7\times10^{-8}$ in the Frobenius norm. The production orders of these representative cases therefore sit well inside the saturated regime, where the leaf-moment quadrature contributes negligibly to the reported errors.

\paragraph{Conservation under quadrature refinement} The same refinement analysis applies to the conserved mass of the Cahn--Hilliard experiment (Section~\ref{sec:cahn-hilliard}). There, the constraint is imposed by the reprojection \eqref{eq:tkan-mass-projection} against a $512$-node midpoint rule, so the machine-precision drift in Table~\ref{tab:cahn-hilliard} is, by itself, the discrete identity evaluated against the enforcing functional. To probe the continuum quantity independently, we re-evaluate the trained constrained solutions on midpoint rules from $N=8$ up to $N=8192$ nodes (Figure~\ref{fig:ch-continuum-mass}a). Every rule other than the $N=512$ production rule is independent of the enforcement, so the convergence trend across them validates the continuum value independently of the enforcing functional. Because the spatial branch enters only through the periodic Fourier embedding, the integrand is smooth and periodic and the midpoint rule converges spectrally, but non-monotonically: the error at $N$ nodes is the sum of the integrand's Fourier modes that alias onto the mean ($k = N, 2N, \dots$), so it oscillates with $N$ as different modes fall in and out of resonance. The worst-case drift falls from $1.0\times10^{-2}$ at $N=8$ to the double-precision floor (${\sim}5\times10^{-15}$) by $N\approx128$, beyond which the integrand carries no resolvable spectral content and refining further moves the result only at the level of floating-point rounding. The production rule $N=512$ thus sits far inside the converged regime. We confirm the mechanism quantitatively: at a fixed time the measured midpoint error matches the sum of the aliased Fourier modes, evaluated from an independent FFT of the field, to within $1\%$ at every $N$. The moment-recursion mass and the direct nodal sum agree to ${\sim}10^{-16}$ throughout. The $512$-node enforcement therefore coincides with the continuum mass to machine precision: mass conservation is a property of the solution, not an artifact of the functional used to impose it.

\begin{figure}[htbp]
    \centering
    \includegraphics[width=0.99\linewidth]{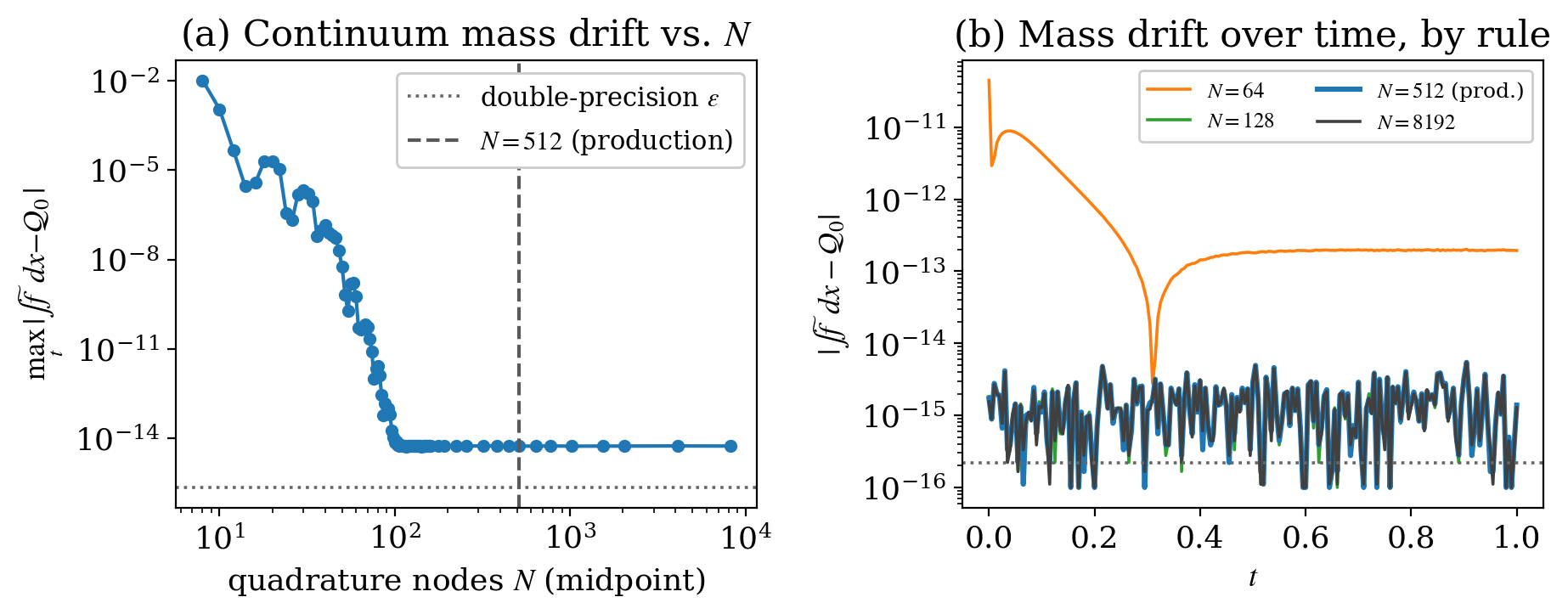}
    \caption{Continuum mass conservation of the Cahn--Hilliard surrogate under quadrature refinement, independent of the $512$-node rule used for enforcement. (a) Worst-case mass drift $\max_t |\int \widetilde f\, dx - \mathcal{Q}_0|$ of the constrained TOUCKAN as the evaluation grid is refined from $N=8$ to $N=8192$ midpoint nodes. The decay is spectral but non-monotone, the signature of periodic-quadrature aliasing: the error at $N$ nodes is the sum of the integrand's Fourier modes with $k = N, 2N, \dots$. It reaches a converged plateau near $5\times10^{-15}$ by $N\approx128$, approximately one order above machine epsilon (dotted). The production rule $N=512$ (dashed) lies far inside the converged regime. (b) Mass drift $|\int \widetilde f\, dx - \mathcal{Q}_0|$ over time $t$, evaluated on four independent midpoint rules ($N=64, 128, 512, 8192$). The under-resolved $N=64$ rule stays above the floor, while $N=128$, the production rule $N=512$, and $N=8192$ coincide at the double-precision floor at every time instance, so conservation holds pointwise in $t$ and independently of the evaluation rule.}
    \label{fig:ch-continuum-mass}
\end{figure}

\FloatBarrier

\section{Randomized quasi-Monte Carlo comparison}
\label{app:rqmc-comparison}

The analytical benchmarks also permit direct integration by randomized quasi-Monte Carlo (RQMC), which can converge much faster than Monte Carlo on smooth integrands~\cite{caflisch1998monte,dick2013high}. We use Sobol' sequences~\cite{sobol1967distribution} with Owen's nested scrambling~\cite{owen1997scrambled}, as implemented in SciPy~\cite{virtanen2020scipy}, and $200$ independent scramblings for each case. Each run reports the sample mean and sample variance. Gaussian inputs are obtained through the standard normal quantile function. Table~\ref{tab:rqmc-comparison} reports results for complete digital nets with $N=2^{18}=262{,}144$ points alongside the TOUCKAN and MC results. Table~\ref{tab:rqmc-matched} compares this budget with $N=2\times10^{5}$, matching the number of evaluations in each MC run. The latter truncates the digital net and loses its balance, increasing the median errors by up to two orders of magnitude on the smooth L\"uthen function.

\begin{table}[H]
\centering
\caption{Relative errors against the exact mean and variance for three benchmarks across $D\in\{5,12,25,50\}$. The surrogate column reports the test relative $L_2$ error. TOUCKAN evaluates moments of the trained surrogate by deterministic propagation. MC and RQMC instead evaluate the true integrand and report the median error over $200$ independent replications. Each MC run uses $N=2\times10^{5}$ pseudo-random samples, while each RQMC run uses a complete Owen-scrambled Sobol' net with $N=2^{18}=262{,}144$ points. Thus, the TOUCKAN errors include the discrepancy between the surrogate and the target, whereas the sampling errors arise from direct integration of the target. TOUCKAN has lower errors than the median MC estimate in every case and than the median RQMC estimate in most cases. The variance entries marked with $\ddagger$ are limited by floating-point roundoff in the second-moment recursion rather than by the surrogate.}
\label{tab:rqmc-comparison}
\medskip 
\resizebox{\linewidth}{!}{%
\begin{tabular}{llc|ccc|ccc}
\hline
 & & Surrogate & \multicolumn{3}{c}{$\mathbb{E}[f]$ rel.\ error} & \multicolumn{3}{c}{$\mathrm{Var}[f]$ rel.\ error} \\
Benchmark & $D$ & rel.\ $L_2$ & TOUCKAN & MC (median) & RQMC (median) & TOUCKAN & MC (median) & RQMC (median) \\
\hline
L\"uthen  & 5  & $1.0\times10^{-6}$ & $5.64\times10^{-10}$ & $1.34\times10^{-3}$ & $2.51\times10^{-9}$  & $1.13\times10^{-8}$ & $2.10\times10^{-3}$ & $3.81\times10^{-6}$ \\
          & 12 & $2.7\times10^{-6}$ & $1.30\times10^{-9}$  & $3.07\times10^{-4}$ & $8.46\times10^{-10}$ & $9.27\times10^{-9}$ & $2.09\times10^{-3}$ & $3.81\times10^{-6}$ \\
          & 25 & $4.5\times10^{-6}$ & $1.77\times10^{-9}$  & $2.53\times10^{-4}$ & $8.28\times10^{-10}$ & $3.81\times10^{-8}$ & $2.35\times10^{-3}$ & $3.81\times10^{-6}$ \\
          & 50 & $8.9\times10^{-6}$ & $1.63\times10^{-10}$ & $1.14\times10^{-4}$ & $5.46\times10^{-10}$ & $1.66\times10^{-8}$ & $1.92\times10^{-3}$ & $3.97\times10^{-6}$ \\
\hline
Bratley   & 5  & $2.9\times10^{-4}$ & $7.22\times10^{-7}$  & $1.03\times10^{-3}$ & $9.65\times10^{-10}$ & $1.80\times10^{-6}$ & $1.74\times10^{-3}$ & $3.80\times10^{-6}$ \\
          & 12 & $1.7\times10^{-3}$ & $1.55\times10^{-7}$  & $1.04\times10^{-3}$ & $1.91\times10^{-6}$  & $1.45\times10^{-5}$ & $1.94\times10^{-3}$ & $9.96\times10^{-6}$ \\
          & 25 & $4.3\times10^{-3}$ & $3.54\times10^{-7}$  & $9.83\times10^{-4}$ & $1.74\times10^{-6}$  & $9.11\times10^{-6}$ & $1.75\times10^{-3}$ & $1.04\times10^{-5}$ \\
          & 50 & $1.1\times10^{-2}$ & $1.31\times10^{-5}$  & $1.08\times10^{-3}$ & $2.02\times10^{-6}$  & $4.59\times10^{-4}\,^{\ddagger}$ & $1.77\times10^{-3}$ & $8.24\times10^{-6}$ \\
\hline
Oakley & 5  & $4.7\times10^{-5}$ & $2.68\times10^{-7}$ & $2.33\times10^{-3}$ & $2.55\times10^{-6}$ & $1.90\times10^{-7}$ & $2.25\times10^{-3}$ & $1.79\times10^{-5}$ \\
          & 12 & $1.4\times10^{-4}$ & $1.35\times10^{-7}$  & $1.45\times10^{-3}$ & $2.04\times10^{-6}$ & $3.45\times10^{-7}$ & $2.17\times10^{-3}$ & $7.09\times10^{-5}$ \\
          & 25 & $2.3\times10^{-4}$ & $4.71\times10^{-7}$  & $9.24\times10^{-4}$ & $2.31\times10^{-6}$ & $1.02\times10^{-6}$ & $2.13\times10^{-3}$ & $1.45\times10^{-4}$ \\
          & 50 & $1.1\times10^{-3}$ & $8.50\times10^{-7}$  & $7.99\times10^{-4}$ & $2.65\times10^{-6}$ & $2.69\times10^{-4}\,^{\ddagger}$ & $1.86\times10^{-3}$ & $2.67\times10^{-4}$ \\
\hline
\end{tabular}%
}
\end{table}

RQMC substantially reduces the sampling errors relative to MC on these benchmarks. At $N=2^{18}$, its median mean errors are of order $10^{-9}$ or smaller for L\"uthen and $10^{-6}$ or smaller for Bratley and Oakley--O'Hagan, while its median variance errors range from $4\times10^{-6}$ to $3\times10^{-4}$. TOUCKAN has lower errors in most cases. The exceptions are the L\"uthen means at $D=12$ and $25$, the Bratley means at $D=5$ and $50$ and variances at $D=12$ and $50$, and the Oakley--O'Hagan variance at $D=50$. The L\"uthen means are already accurate to $10^{-9}$ with both methods, while the two Oakley--O'Hagan variance errors are both approximately $2.7\times10^{-4}$. At the matched budget, TOUCKAN has lower errors in $22$ of the $24$ cases, with both Bratley moments at $D=50$ as the exceptions.

\begin{table}[H]
\centering
\caption{Median relative errors of RQMC over $200$ scramblings at the matched budget $N=2\times10^{5}$ and at $N=2^{18}$, for the true integrands of Table~\ref{tab:rqmc-comparison}.}
\label{tab:rqmc-matched}
\medskip
\begin{tabular}{llcc|cc}
\hline
 & & \multicolumn{2}{c|}{$\mathbb{E}[f]$ rel.\ error} & \multicolumn{2}{c}{$\mathrm{Var}[f]$ rel.\ error} \\
Benchmark & $D$ & $N=2\times10^{5}$ & $N=2^{18}$ & $N=2\times10^{5}$ & $N=2^{18}$ \\
\hline
L\"uthen & 5  & $5.31\times10^{-7}$ & $2.51\times10^{-9}$  & $6.44\times10^{-6}$ & $3.81\times10^{-6}$ \\
         & 12 & $1.74\times10^{-7}$ & $8.46\times10^{-10}$ & $2.56\times10^{-5}$ & $3.81\times10^{-6}$ \\
         & 25 & $7.79\times10^{-8}$ & $8.28\times10^{-10}$ & $3.30\times10^{-5}$ & $3.81\times10^{-6}$ \\
         & 50 & $4.62\times10^{-8}$ & $5.46\times10^{-10}$ & $5.26\times10^{-5}$ & $3.97\times10^{-6}$ \\
\hline
Bratley  & 5  & $2.22\times10^{-6}$ & $9.65\times10^{-10}$ & $1.33\times10^{-5}$ & $3.80\times10^{-6}$ \\
         & 12 & $4.39\times10^{-6}$ & $1.91\times10^{-6}$  & $2.09\times10^{-5}$ & $9.96\times10^{-6}$ \\
         & 25 & $4.65\times10^{-6}$ & $1.74\times10^{-6}$  & $2.28\times10^{-5}$ & $1.04\times10^{-5}$ \\
         & 50 & $4.44\times10^{-6}$ & $2.02\times10^{-6}$  & $2.33\times10^{-5}$ & $8.24\times10^{-6}$ \\
\hline
Oakley   & 5  & $7.10\times10^{-6}$ & $2.55\times10^{-6}$  & $4.61\times10^{-5}$ & $1.79\times10^{-5}$ \\
         & 12 & $5.91\times10^{-6}$ & $2.04\times10^{-6}$  & $9.69\times10^{-5}$ & $7.09\times10^{-5}$ \\
         & 25 & $6.45\times10^{-6}$ & $2.31\times10^{-6}$  & $2.21\times10^{-4}$ & $1.45\times10^{-4}$ \\
         & 50 & $7.93\times10^{-6}$ & $2.65\times10^{-6}$  & $4.42\times10^{-4}$ & $2.67\times10^{-4}$ \\
\hline
\end{tabular}
\end{table}

These results concern different computational settings. MC and RQMC evaluate the true integrand, whereas TOUCKAN evaluates moments of an already-trained surrogate. Direct sampling is useful when new function evaluations are affordable. When only fixed observations or expensive simulation data are available, repeated surrogate moments and integral constraints can instead be evaluated without new calls to the original model. The analytical benchmarks have inexpensive integrands, so the measured sampling times do not represent the cost of generating data from an expensive simulation. At the same sample count, one RQMC replicate costs $0.83$--$0.98$ times one MC replicate for the uniform benchmarks and approximately $1.55$ times for the Gaussian benchmark, where the quantile transform dominates. As in the main timing comparison, surrogate training is excluded, and TOUCKAN times refer to warm evaluations with one-time compilation excluded. These timings therefore describe moment evaluation after surrogate construction rather than the total cost of solving a new integration problem.

\FloatBarrier

\section{Problem descriptions}
\label{app:problems}

This appendix defines the three benchmark problems used in our experiments. All three share the structural properties that make exact reference moments available: the input coordinates are independent, and the function is built from sums and products of univariate terms, so every moment factorizes into one-dimensional integrals that are available in closed form or computable to machine precision by one-dimensional Gaussian quadrature. The reference means and variances below are therefore exact, not Monte Carlo estimates.

\subsection{L\"uthen function (uniform inputs)}
\label{app:luthen}

The first problem is the high-dimensional polynomial-chaos benchmark of L\"uthen et al.\ \cite{luthen2021Sparse},
\begin{align}
    f(\mathbf{x})
    ={}&
    3
    - \frac{5}{D}\sum_{d=1}^{D} d\,x_d
    + \frac{1}{D}\sum_{d=1}^{D} d\,x_d^3
    + \frac{1}{3D}\sum_{d=1}^{D} d \ln\!\left(x_d^2 + x_d^4\right)
    \nonumber\\
    &+ \mathbf{1}_{D \ge 5}\left(x_1 x_2^2 + x_2 x_4 - x_3 x_5\right) + \mathbf{1}_{D \ge 51}\, x_{51} + \mathbf{1}_{D \ge 54}\, x_{50}\,x_{54}^2,
    \label{eq:app-luthen}
\end{align}
with inputs
\begin{align}
    x_d \sim \mathcal{U}(1,2) \quad \forall\, d \neq 20,
    \qquad
    x_{20} \sim \mathcal{U}(1,3).
    \label{eq:app-luthen-inputs}
\end{align}
The interaction terms use fixed indices, so the indicators $\mathbf{1}_{D \ge 5}$, $\mathbf{1}_{D \ge 51}$, and $\mathbf{1}_{D \ge 54}$ define the function at every dimension: the first three interaction terms are present in all experiments reported here ($D \ge 5$), while the terms $x_{51}$ and $x_{50}x_{54}^2$ contribute only once $D \ge 51$ and $D \ge 54$, respectively. The same indicators apply to the reference moments below. By independence and linearity, the expectation separates term by term,
\begin{align}
    \mathbb{E}[f]
    ={}&
    3
    - \frac{5}{D}\sum_{d=1}^{D} d\,\mathbb{E}[x_d]
    + \frac{1}{D}\sum_{d=1}^{D} d\,\mathbb{E}[x_d^3]
    + \frac{1}{3D}\sum_{d=1}^{D} d\,\mathbb{E}\!\left[\ln\!\left(x_d^2 + x_d^4\right)\right]
    \nonumber\\
    &+ \mathbf{1}_{D \ge 5}\left(\mathbb{E}[x_1]\,\mathbb{E}[x_2^2] + \mathbb{E}[x_2]\,\mathbb{E}[x_4] - \mathbb{E}[x_3]\,\mathbb{E}[x_5]\right) + \mathbf{1}_{D \ge 51}\,\mathbb{E}[x_{51}] + \mathbf{1}_{D \ge 54}\,\mathbb{E}[x_{50}]\,\mathbb{E}[x_{54}^2],
    \label{eq:app-luthen-mean}
\end{align}
where every factor is a one-dimensional integral over $\mathcal{U}(1,2)$ (or $\mathcal{U}(1,3)$ for $x_{20}$) available in closed form, including the logarithmic term. The exact variance follows in the same way: expanding $f^2$ and applying independence reduces $\mathbb{E}[f^2]$ to a finite combination of products of one-dimensional moments, and $\mathrm{Var}[f] = \mathbb{E}[f^2] - (\mathbb{E}[f])^2$. In practice we evaluate all one-dimensional integrals to machine precision with $200$-node Gauss--Legendre quadrature. Table~\ref{tab:luthen-references} lists the resulting reference moments for every dimension reported in Section~\ref{sec:experiments}.

\begin{table}[h]
\centering
\caption{Exact reference moments of the L\"uthen function \eqref{eq:app-luthen} for every reported input dimension, in physical units.}
\label{tab:luthen-references}
\begin{tabular}{lccccc}
\hline
$D$ & $5$ & $12$ & $20$ & $25$ & $50$ \\
\hline
$\mathbb{E}[f]$ & $-2.80818$ & $-13.66773$ & $-22.02031$ & $-30.58880$ & $-70.99621$ \\
$\mathrm{Var}[f]$ & $6.39955$ & $7.55664$ & $34.18982$ & $26.33687$ & $20.07893$ \\
\hline
\end{tabular}
\end{table}

\subsection{Bratley function (uniform inputs)}
\label{app:bratley}

The second problem is the alternating sum of nested products of Bratley et al.,
\begin{align}
    f(\mathbf{x})
    =
    \sum_{d=1}^{D} (-1)^{d} \prod_{j=1}^{d} x_j,
    \qquad
    \mathbf{x} \sim \mathcal{U}[0,1]^{D}.
    \label{eq:app-bratley}
\end{align}
Since the coordinates are independent with $\mathbb{E}[x_j] = \tfrac{1}{2}$, each nested product has mean $\mathbb{E}\big[\prod_{j=1}^{d}x_j\big] = 2^{-d}$, and the exact mean is the geometric sum
\begin{align}
    \mathbb{E}[f]
    =
    \sum_{d=1}^{D}\left(-\tfrac{1}{2}\right)^{d}
    =
    -\frac{1}{3}\left(1 - \left(-\tfrac{1}{2}\right)^{D}\right)
    \;\xrightarrow{\;D\to\infty\;}\;
    -\frac{1}{3}.
    \label{eq:app-bratley-mean}
\end{align}
For the second moment, the product of two nested products squares the shared coordinates and leaves the remaining ones linear; with $\mathbb{E}[x_j^2] = \tfrac{1}{3}$,
\begin{align}
    \mathbb{E}\!\left[\,\prod_{j=1}^{d} x_j \prod_{j=1}^{d'} x_j\right]
    =
    \left(\tfrac{1}{3}\right)^{\min(d,d')}
    \left(\tfrac{1}{2}\right)^{|d-d'|},
    \label{eq:app-bratley-cross}
\end{align}
so the exact variance is the finite double sum
\begin{align}
    \mathrm{Var}[f]
    =
    \sum_{d=1}^{D}\sum_{d'=1}^{D}
    (-1)^{d+d'}
    \left(\tfrac{1}{3}\right)^{\min(d,d')}
    \left(\tfrac{1}{2}\right)^{|d-d'|}
    -
    \big(\mathbb{E}[f]\big)^{2}.
    \label{eq:app-bratley-var}
\end{align}

\subsection{Oakley--O'Hagan function (Gaussian inputs)}
\label{app:oakley}

The third problem is an Oakley--O'Hagan-type function with standard Gaussian inputs,
\begin{align}
    f(\mathbf{x})
    =
    \sum_{d=1}^{D}
    \left(
    a_{1,d}\,x_d + a_{2,d}\sin x_d + a_{3,d}\cos x_d
    \right)
    +
    \mathbf{x}^{\top} M\, \mathbf{x},
    \qquad
    \mathbf{x} \sim \mathcal{N}(0, I_D).
    \label{eq:app-oakley}
\end{align}
The coefficients are drawn once from a fixed seed (NumPy \texttt{RandomState(7)}) and then held fixed. Their realized values will be reproducible from the code repository released upon acceptance: $a_{1,d}, a_{2,d}, a_{3,d} \sim \mathcal{U}(0.5, 1.5)/\sqrt{D}$, and $M \in \mathbb{R}^{D\times D}$ is symmetric with diagonal entries $M_{dd} \sim \mathcal{U}(-0.5, 0.5)/\sqrt{D}$ and off-diagonal entries $M_{dd'} = S_{dd'}/2$, where $S_{dd'} \sim \mathcal{U}(-1.2/D,\, 1.2/D)$. This scaling keeps the mean and the variance of $f$ of order one at every dimension $D$. Grouping the univariate terms as $g_d(x) = a_{1,d}x + a_{2,d}\sin x + a_{3,d}\cos x + M_{dd}\,x^2$, the standard Gaussian identities $\mathbb{E}[\sin X] = 0$, $\mathbb{E}[\cos X] = \mathbb{E}[X \sin X] = e^{-1/2}$, $\mathbb{E}[X\cos X] = \mathbb{E}[X^2 \sin X] = \mathbb{E}[X^2\cos X] = 0$, $\mathbb{E}[X^2]=1$, and $\mathbb{E}[X^4]=3$ give the exact mean
\begin{align}
    \mathbb{E}[f]
    =
    \sum_{d=1}^{D}
    \left(
    a_{3,d}\,e^{-1/2} + M_{dd}
    \right).
    \label{eq:app-oakley-mean}
\end{align}
The ANOVA decomposition of \eqref{eq:app-oakley} truncates exactly at order two: the main effects are the centered $g_d$, and the only interactions are the pair terms $S_{dd'}\,x_d x_{d'}$ with $S_{dd'} = 2M_{dd'}$ for $d<d'$. The exact variance is therefore
\begin{align}
    \mathrm{Var}[f]
    =
    \sum_{d=1}^{D}
    \left(
    \mathbb{E}[g_d^2] - \mathbb{E}[g_d]^2
    \right)
    +
    \sum_{d<d'}
    \left(2M_{dd'}\right)^2,
    \label{eq:app-oakley-var}
\end{align}
with $\mathbb{E}[g_d] = a_{3,d}\,e^{-1/2} + M_{dd}$ and
\begin{align}
    \mathbb{E}[g_d^2]
    =
    a_{1,d}^2
    + \tfrac{1}{2}\,a_{2,d}^2\left(1 - e^{-2}\right)
    + \tfrac{1}{2}\,a_{3,d}^2\left(1 + e^{-2}\right)
    + 3M_{dd}^2
    + 2\,a_{1,d}\,a_{2,d}\,e^{-1/2}.
    \label{eq:app-oakley-g2}
\end{align}

\section{Implementation details}
\label{app:implementation}

This appendix records the architecture and training configuration of the experiments in Section~\ref{sec:experiments}.

\paragraph{Reproducibility conventions} All experiments run in double precision. Here, one iteration denotes one optimizer update computed on one batch of the stated size; all training totals below are iteration counts. Reported moment-evaluation timings are warm times of the compiled kernels, and the one-time JIT compilation is reported separately where relevant. The quadrature orders of the moment recursions are: $200$-node Gauss--Legendre for the uniform-input benchmarks (L\"uthen and Bratley, used both for the propagated moments and for the reference moments), $128$-node Gauss--Hermite for the Gaussian-input benchmark (Oakley--O'Hagan) and for the Schr\"odinger Gram matrix, and a $512$-node midpoint rule for the periodic Cahn--Hilliard integrand. All Gauss--Hermite rules are the probabilists' rule, with weight function $e^{-x^{2}/2}$ and weights normalized by $1/\sqrt{2\pi}$, so that each rule integrates against the standard normal density. Random seeds: five seeds ($111$, $222$, $333$, $444$, $555$) for the parameter-efficiency study (Section~\ref{sec:param-efficiency}) and $111$ for the remaining function-approximation benchmarks (Section~\ref{sec:tractable-inference}), $9999$ for the Cahn--Hilliard and Schr\"odinger experiments, and NumPy \texttt{RandomState(7)} for the fixed problem data (the Oakley coefficients and the oscillator rotation), whose realized values will be reproducible from the code repository released upon acceptance.

\paragraph{Error metrics} Throughout the experiments, the accuracy of a scalar quantity $s$ (an eigenvalue, a moment, a conserved integral) is reported as the relative error against its reference $s^{\mathrm{ref}}$,
\begin{align} \label{eq:rel-error}
\varepsilon(s)
=
\frac{|s - s^{\mathrm{ref}}|}{|s^{\mathrm{ref}}|},
\end{align}
and the accuracy of a field $u$ as the relative $L_2$ error against its reference $u^{\mathrm{ref}}$,
\begin{align} \label{eq:rel-l2}
\varepsilon_{L_2}(u)
=
\frac{\lVert u - u^{\mathrm{ref}} \rVert_{2}}{\lVert u^{\mathrm{ref}} \rVert_{2}},
\end{align}
with the norms evaluated on the test grid or test set specified in each experiment.

\subsection{Parameter-efficiency analysis}
\label{app:impl-luthen-eff}\label{app:impl-luthen}

All four inner-block variants of Section~\ref{sec:param-efficiency} share the same backbone. The shared trunk is $h(x) = \mathrm{MLP}(\mathrm{emb}(x))$, where $\mathrm{emb}$ is a Chebyshev feature embedding of degree $15$ and the MLP has three hidden layers of width $H = 20$ with $\tanh$ activations. The number of branches is $Q = 2D + 1 = 41$, and the outer block evaluates per-branch Chebyshev expansions of degree $P = 15$, for $Q(P+1) = 656$ outer parameters, identical across all variants. The variants differ only in how the per-dimension branches are obtained from the trunk: independent per-dimension networks ($\Psi_{d,q} = \tanh(\mathrm{MLP}_d(\mathrm{emb}_d(x_d)))_q$, about $2{,}057$ parameters per dimension); a full linear head per dimension ($\Psi_{d} = \tanh(W_d\, h(x_d) + b_d)$ with $W_d \in \mathbb{R}^{41 \times 20}$, $861$ parameters per dimension); the diagonal specialization \eqref{eq:diag-specialization} ($41$ scalars per dimension, applied to the shared head $\boldsymbol{\psi}(x) = W h(x) + b \in \mathbb{R}^{41}$); and the rank-one Lorentz form ($\lambda_{d,q} = \lambda_d$, one scalar per dimension). In every case the inner block remains a sum of univariate branches, and the branch decomposition used by the moment recursion was verified against the forward pass to $\leq 2\times10^{-13}$ on all trained checkpoints.

Training uses $N_C = 200{,}000$ Latin hypercube samples in batches of $20{,}000$ and $10{,}000$ iterations of the SOAP optimizer \cite{vyas2024soap} in double precision. Each variant is trained from five random seeds ($111$, $222$, $333$, $444$, $555$), and the error columns of Table~\ref{tab:param-efficiency} report the mean $\pm$ standard deviation across them. The inner block uses an initial learning rate of $10^{-3}$ with exponential decay by a factor $0.9$ every $5{,}000$ steps; the outer block is updated by a separate SOAP instance with the same schedule scaled by $0.2$. Per-sample variational residual-based attention (vRBA) weights \cite{toscano2026variational} are updated as $\lambda_i \leftarrow \gamma \lambda_i + \eta\, r_i / \max_j r_j$ with $\eta = 0.01$ and $\gamma = 0.999$. The relative $L_2$ test errors in Table~\ref{tab:param-efficiency} are computed on $20{,}000$ independent Latin hypercube points, and the analytical moments of each trained surrogate are cross-checked against a $2\times10^{6}$-sample Monte Carlo estimate of the same network, agreeing to within one Monte Carlo standard error.

\subsection{High-dimensional benchmarks}
\label{app:impl-highd}

The twelve surrogates of Section~\ref{sec:tractable-inference} span three benchmarks and $D\in\{5,12,25,50\}$. Each uses the diagonal inner block \eqref{eq:diag-specialization} with a shared trunk of three hidden layers of width $H=\max\{20,\lfloor Q/2\rfloor\}$ and $\tanh$ activations, followed by a linear head producing $Q$ branch features; the outer block evaluates per-branch Chebyshev expansions of degree $P=15$. The L\"uthen models use $Q=2D+1$. The Oakley--O'Hagan models use the same rule except at $D=5$, where $Q=12$, while the Bratley models use $Q=\lfloor2.5D\rfloor$. At $D=50$, these choices give $Q=101$ ($H=50$) for L\"uthen and Oakley--O'Hagan and $Q=125$ ($H=62$) for Bratley, with $17{,}783$, $18{,}482$, and $25{,}007$ trainable parameters, respectively. For the uniform problems the trunk input is a Chebyshev feature embedding of degree $15$ on inputs affinely normalized to $[0,1]$; for the Gaussian problem, whose inputs are unbounded, the embedding is replaced by random Fourier features whose frequencies are drawn from $\mathcal{N}(0, \sigma^{2})$ with bandwidth $\sigma=1$. Targets are standardized with the exact reference moments of \ref{app:luthen}--\ref{app:oakley}, and for the Gaussian case the training inputs are drawn by mapping a Latin hypercube sample through the standard normal quantile function.

Training proceeds in three stages in double precision (seed $111$), on $N_C = 200{,}000$ Latin hypercube samples with batch size $20{,}000$; the full procedure is summarized in Algorithm~\ref{alg:highd-training}. Stage one runs $5{,}000$ iterations of Adam with initial learning rate $10^{-3}$, exponentially decayed by a factor $0.9$ every $5{,}000$ steps, with per-sample vRBA weights \cite{toscano2026variational} (quadratic potential, $\eta = 0.01$, $\gamma = 0.999$). Stage two polishes all parameters with the self-scaled Broyden (SSBroyden) quasi-Newton method with backtracking line search \cite{urban2025unveiling,jnini2026curvature} for $10{,}000$ iterations, executed in chunks of $100$ iterations with the dense inverse-Hessian approximation carried across chunks and a fresh batch drawn per chunk. Stage three refits the last linear layer (the outer Chebyshev coefficients) by ridge-regularized least squares with regularization parameter $\lambda = 10^{-7}$ on the full training set, which is an exact linear solve since the outer block is linear in its coefficients. Relative $L_2$ test errors are computed on $20{,}000$ independent Latin hypercube points. Across all twelve checkpoints, the branch decomposition used by the moment recursion agrees with the forward pass to $\leq 4\times10^{-10}$, and the analytical moments at the production quadrature orders agree with $2\times10^{6}$-sample Monte Carlo estimates of the same networks to within one Monte Carlo standard error. The warm moment evaluation takes $0.3$--$0.6$\,s per model. One-time JIT compilation takes $1$--$38$\,s and is excluded from the reported warm timings.

\paragraph{Sampling baselines} The Monte Carlo columns of Table~\ref{tab:tractable-inference-errors} use $200$ independent runs of the true integrand with $N=2\times10^{5}$ pseudo-random samples each. Each run reports the sample mean and the unbiased sample variance.

\begin{algorithm}[!ht]
\caption{Three-stage training for the high-dimensional benchmarks.}
\label{alg:highd-training}
\begin{algorithmic}[1]
\Require training set $\{(\mathbf{x}_i, u_i)\}_{i=1}^{N_C}$ (Latin hypercube, standardized targets); TOUCKAN parameters $\theta = (\theta_{\mathrm{inner}}, C)$ with outer Chebyshev coefficients $C$; batch size $B$; Adam iterations $N_1$; SSBroyden iterations $N_2$ in chunks of $c$; ridge parameter $\lambda_{\mathrm{ridge}}$
\Ensure trained surrogate $f_\theta$, ready for the moment evaluation of Algorithms~\ref{alg:tkan-expectation} and~\ref{alg:tkan-variance}
\State \textbf{Stage 1 — Adam warm-up with vRBA \cite{toscano2026variational}:}
\For{$e = 1,\dots,N_1$}
    \State draw a batch of size $B$; compute residuals $r_i = |f_\theta(\mathbf{x}_i) - u_i|$
    \State $\lambda_i \gets \gamma \lambda_i + \eta\, r_i / \max_j r_j$ \Comment{vRBA weights, quadratic potential}
    \State update $\theta$ with Adam on $\frac{1}{B}\sum_i \big((\lambda_i + 0.1)\, r_i\big)^2$ \Comment{lr $10^{-3}$, decay $0.9$ per $5{,}000$ steps}
\EndFor
\State \textbf{Stage 2 — SSBroyden polish \cite{jnini2026curvature}:} $\mathbf{w} \gets \mathrm{vec}(\theta)$,\; $H \gets I$
\For{$k = 1,\dots,N_2/c$}
    \State draw a fresh batch of size $B$
    \State run $c$ SSBroyden iterations with backtracking line search on the plain MSE, starting from $(\mathbf{w}, H)$
    \State carry the dense inverse-Hessian approximation: $H \gets \frac{1}{2}(H + H^{\top})$, reset $H \gets I$ if not positive definite
\EndFor
\State \textbf{Stage 3 — outer-layer refit:} with $\theta_{\mathrm{inner}}$ frozen, assemble the outer Chebyshev features $\Phi_{i,(q,p)} = T_p\big(\xi_q(\mathbf{x}_i)\big)$ on the full training set
\State $C \gets \arg\min_{C} \|\Phi\, C - \mathbf{u}\|_2^2 + \lambda_{\mathrm{ridge}} \|C\|_2^2$ \Comment{exact linear solve}
\end{algorithmic}
\end{algorithm}

\subsection{Cahn--Hilliard equation}
\label{app:impl-ch}

The reference solution of Section~\ref{sec:cahn-hilliard} is computed with a Fourier-spectral exponential time-differencing (ETDRK4) solver on $512$ spatial points with internal step $\Delta t = 10^{-5}$, stored as $201$ snapshots on a uniform time grid; its own mass drift is below $10^{-16}$. The conserved value $\mathcal{Q}_0 \approx -0.28209479$ is the grid quadrature of the $t = 0$ snapshot. Training uses $N_C = 200{,}000$ Latin hypercube collocation points in $[0,1] \times [-1,1]$ with batch size $20{,}000$, and the $512$ reference points of the $t = 0$ slice as initial-condition points. Periodicity in $x$ is enforced exactly by a periodic Fourier-feature embedding of the spatial input, so the loss has only two terms: the mean-squared residual of \eqref{eq:cahn-hilliard} in conservative form, with the fourth derivative obtained by nesting forward-mode directional derivatives, and the initial-condition mean-squared error. The residual term carries per-point vRBA weights \cite{toscano2026variational} (exponential potential; $\eta = 0.01$, $\gamma = 0.999$ in the first-order stage, $\gamma = 0.9$ in the quasi-Newton stage), which also define the sampling distribution of the batches. The initial-condition term carries a global weight set adaptively by a gradient-norm ratio in the first-order stage and fixed at $100$ in the quasi-Newton stage. All runs use double precision and seed $9999$.

All five models are parameter-matched at approximately $15$k. The MLP and the TOUCKANs use $\tanh$ activations, and the TNN uses sine activations. The independent TOUCKAN assigns one inner network per input dimension: the spatial branch wraps its input in the periodic embedding ($k_{\max} = 15$), the time branch acts on the raw coordinate, and each branch is an MLP with four hidden layers of width $40$ followed by a head producing $Q = 40$ bounded channels, for $14{,}840$ parameters. The full-head and diagonal variants share one embedding and one trunk between the two input dimensions, as in \ref{app:impl-luthen-eff}. The shared embedding spans the full period in $x$, which preserves the hard periodicity. The time interval $t \in [0,1]$ covers at most half of that period. On this interval the embedding is smooth and one-to-one. For the time coordinate it therefore acts as an ordinary, non-periodic feature expansion. The full-head model uses a trunk of six hidden layers of width $40$ with one linear head per dimension producing $Q = 64$ channels ($15{,}328$ parameters). The diagonal model uses a trunk of six hidden layers of width $44$, a single shared head, and one trainable scalar per dimension and channel ($14{,}912$ parameters). All three TOUCKANs use outer Chebyshev expansions of degree $P = 9$. The MLP baseline applies the periodic embedding ($k_{\max} = 5$) to the spatial coordinate, concatenates it with $t$, and processes the result with five hidden layers of width $59$ ($14{,}928$ parameters). The TNN~\cite{wang2022tensor} represents the field as the separable sum $f(t,x) = \sum_{j=1}^{R} c_j\,\varphi_j(t)\,\psi_j(x)$ of rank $R = 24$ with trainable coefficients $c_j$. Two factor networks with three hidden layers of width $53$ produce the temporal factors $\varphi_j$ and the spatial factors $\psi_j$, and the spatial input passes through the same periodic embedding ($k_{\max} = 5$), which keeps the periodicity exact. Both networks use orthogonal weight initialization with biases at $0.5$, and each factor is normalized by its root-mean-square value on its one-dimensional quadrature grid (the $512$-node spatial rule and a $256$-node midpoint rule in time), following~\cite{wang2022tensor}. The model has $14{,}753$ parameters.

The exact-mass reprojection implements Section~\ref{sec:exact-mass} with $w \equiv 1$ and $g \equiv 1$. At each loss evaluation, the spatial branch is evaluated at $512$ midpoint quadrature nodes. The branch depends on $x$ only through the periodic embedding, so the integrands $\widetilde{\Psi}^{\,k}$ are smooth and periodic and the midpoint rule converges spectrally. The weighted moments then follow from the binomial recursion, with the time branch contributing known powers, and the contraction with the outer Chebyshev coefficients through the monomial change of basis yields $\mathcal{Q}[f](t)$ in closed form \eqref{eq:tkan-conserved-closed-form}. The output is the projection \eqref{eq:tkan-mass-projection} with $\mathcal{Q}_w[g] = |\Omega| = 2$. The correction is constant in $x$, so it is invisible to the spatial derivatives in the residual and enters only through the time derivative. Its cost is one evaluation of the spatial branch at the $512$ quadrature nodes per loss call. The constraint is thus enforced exactly for the $512$-node quadrature functional. Since the integrand is smooth and periodic, the continuum mass is conserved to the spectral accuracy of the midpoint rule, which we verify directly by refining the integration grid (~\ref{app:sensitivity}, Figure~\ref{fig:ch-continuum-mass}).

The same reprojection carries over to the TNN because its representation is already separable. The mass functional factorizes as $\mathcal{Q}[f](t) = \sum_{j} c_j\,\varphi_j(t) \int_{\Omega} \psi_j(x)\,dx$, where the one-dimensional factor integrals are evaluated on the same $512$-node midpoint rule, and the output is projected exactly as in \eqref{eq:tkan-mass-projection}. As in the TOUCKAN case, the cost is one evaluation of the spatial factor network at the quadrature nodes per loss call.

The constrained TOUCKANs are trained by fitting first and projecting second. Each architecture is first trained without the constraint: $6{,}250$ iterations of the Lion optimizer \cite{chen2023symbolic} (initial learning rate $10^{-3}$, exponential decay by $0.9$ every $5{,}000$ steps, outer coefficients at $20\%$ of the base rate), followed by $20{,}000$ SSBroyden iterations \cite{jnini2026curvature} with backtracking line search, executed in chunks of $100$ with the inverse-Hessian approximation carried across chunks and a fresh vRBA-sampled batch per chunk. The exact-mass reprojection is then applied, the converged parameters are loaded unchanged (the constrained and unconstrained models share the same parameter tree), and the model is polished for a further $20{,}000$ SSBroyden iterations under the constraint. The order matters: training the constrained model from a random initialization stalls, because the projection makes the output invariant to the spatially constant mode of the raw network, and this null space blocks the escape path that the unconstrained optimization uses early in training. The warm start instead places the model on the constraint manifold near the solution, and the second-order stage only polishes. The constrained TNN follows the same fit-then-project schedule as the TOUCKANs. The MLP baseline, which carries no constraint, is trained from scratch with the same pipeline at $6{,}250$ first-order and $40{,}000$ SSBroyden iterations, so every model in Table~\ref{tab:cahn-hilliard} receives the same total second-order budget. Errors are computed on the exact $512 \times 201$ reference grid without interpolation, and the mass drift is the maximum over the snapshots of the grid-quadrature deviation from $\mathcal{Q}_0$.

For reference, the unconstrained twins at the warm-start point ($6{,}250$ Lion and $20{,}000$ SSBroyden iterations, before the exact-mass reprojection is applied) reach relative $L_2$ errors of $6.6\times10^{-4}$ (independent), $3.3\times10^{-3}$ (full-head), and $2.8\times10^{-2}$ (diagonal), with mass drifts of $6.4\times10^{-4}$, $3.6\times10^{-3}$, and $3.2\times10^{-2}$, respectively. The TNN twin reaches $4.7\times10^{-3}$ with a mass drift of $5.2\times10^{-3}$. In each case the unconstrained drift sits at the level of the field error, consistent with the soft-conservation behavior of the MLP baseline, while the constrained continuation reduces the drift to machine precision (Table~\ref{tab:cahn-hilliard}). Because the constrained polish also continues the optimization, these paired values do not isolate the effect of the reprojection on field accuracy. An equal-budget unconstrained continuation is left for future work.

\subsection{Schr\"odinger eigenvalue problem}
\label{app:impl-schrodinger}

The test problem is chosen to be genuinely multivariate and yet exactly solvable, so that errors can be measured precisely. It is the coupled harmonic oscillator: the potential in \eqref{eq:schrodinger} is $V(\mathbf{x}) = \tfrac{1}{2}\mathbf{x}^{\top}A\mathbf{x}$, and specifying the problem means specifying the matrix $A$. We build it as $A = R\,\mathrm{diag}(\omega_1^{2},\dots,\omega_D^{2})\,R^{\top}$ from two ingredients. The frequencies $\omega_i = 1 + 0.2\,(i-1)$, so $\boldsymbol{\omega} = (1.0, 1.2)$ in $D = 2$ and $(1.0, 1.2, 1.4)$ in $D = 3$, set the stiffness of the $D$ normal modes. The frequencies are commensurate, so the full spectrum $E(\mathbf{n})$ contains degeneracies at higher energies. The four target eigenvalues, however, are non-degenerate in each dimension considered ($1.1$, $2.1$, $2.3$, $3.1$ for $D = 2$; $1.8$, $2.8$, $3.0$, $3.2$ for $D = 3$; $3.5$, $4.5$, $4.7$, $4.9$ for $D = 5$). Hence the target eigenfunctions are unique up to sign and per-state errors are well defined. The rotation $R \in \mathbb{R}^{D\times D}$ is a dense orthogonal matrix, the QR factor of a standard normal matrix drawn with NumPy \texttt{RandomState(7)}. This rotation mixes all coordinates, so the potential contains every cross term $x_d x_{d'}$ and neither the potential nor the eigenfunctions separate in the input variables. The surrogate therefore cannot exploit any one-dimensional structure of the problem.

The exact reference is obtained by undoing the rotation. In the rotated coordinates $\mathbf{y} = R^{\top}\mathbf{x}$, the Hamiltonian splits into $D$ independent one-dimensional oscillators with frequencies $\omega_i$, whose solutions are classical. The eigenfunctions are products of one-dimensional Hermite functions of the rotated coordinates, indexed by occupation numbers $\mathbf{n} = (n_1,\dots,n_D) \in \mathbb{N}_0^{D}$, and the eigenvalues are $E(\mathbf{n}) = \sum_{i=1}^{D}\omega_i\,(n_i + \tfrac{1}{2})$; sorting them gives the lowest $N_s$ eigenpairs. These closed forms are used only to evaluate errors. Training never sees them.

\begin{algorithm}[!ht]
\caption{Hard-normalized eigensolver training for the Schr\"odinger problem.}
\label{alg:schrodinger-training}
\begin{algorithmic}[1]
\Require collocation set $\{\mathbf{x}_i\}_{i=1}^{N_C}$ drawn from $\rho$; TOUCKAN parameters $\theta$ with $N_s$ outer coefficient columns; eigenvalues $E_1,\dots,E_{N_s}$ initialized below the spectrum; weights $\lambda_B$, $w_R$; batch size $B$; first-order iterations $N_1$; SSBroyden iterations $N_2$ in chunks of $c$
\Ensure eigenpairs $(E_a, \psi_a)$ with $\psi_a = f_a\,\chi/\sqrt{G_{aa}}$, normalized identically in $\theta$
\Function{Loss}{$\theta, E$; $\mathcal{B}, \mathcal{B}'$} \Comment{$\mathcal{B}$: residual batch; $\mathcal{B}'$: uniform batch for $K$}
    \State $\big(f_1(\mathbf{x}),\dots,f_{N_s}(\mathbf{x})\big) \gets \mathrm{TOUCKAN}_{\theta}(\mathbf{x})$,\; $\forall\,\mathbf{x}$ in both batches \Comment{one network \eqref{eq:tkan-architecture}, $N_s$ output columns}
    \State $G \gets \mathbb{E}_{\rho}[f_a f_b]$ from $\theta$, by the cross-moment recursion \eqref{eq:tkan-cross-recursion} \Comment{deterministic Gram matrix, differentiable}
    \State $W(\mathbf{x}) \gets \tfrac{D}{4} - \tfrac{1}{8}|\mathbf{x}|^{2} + \tfrac{1}{2}\,\mathbf{x}^{\top}A\,\mathbf{x}$ \Comment{envelope-transformed potential}
    \State $r_a \gets -\tfrac{1}{2}\Delta f_a + \tfrac{1}{2}\,\mathbf{x}\cdot\nabla f_a + W f_a - E_a f_a$,\; $\forall\,\mathbf{x}$ in both batches \Comment{residual Eq. \eqref{eq:app-schrodinger-residual}}
    \State $\lambda_{i,a} \gets \gamma\lambda_{i,a} + \eta\, |r_a(\mathbf{x}_i)| / \max_j |r_a(\mathbf{x}_j)|$ \Comment{vRBA weights}
    \State $\mathcal{L}_{\mathrm{res}} \gets \frac{1}{|\mathcal{B}|}\sum_{\mathbf{x}_i \in \mathcal{B}}\sum_{a} \big(\lambda_{i,a}\, r_a(\mathbf{x}_i)/\sqrt{G_{aa}}\big)^{2}$ \Comment{Monte Carlo}
    \State $\widehat{G}_{ab} \gets G_{ab}/\sqrt{G_{aa}G_{bb}}$
    \State $\mathcal{L}_{\mathrm{orth}} \gets \sum_{a \neq b} \widehat{G}_{ab}^{\,2}$ \Comment{exact, free of Monte Carlo noise}
    \State $K_{ab} \gets \frac{1}{|\mathcal{B}'|}\sum_{\mathbf{x}_i \in \mathcal{B}'} f_a(\mathbf{x}_i)\big(r_b(\mathbf{x}_i) + E_b f_b(\mathbf{x}_i)\big)$ \Comment{energy matrix Eq. \eqref{eq:app-schrodinger-K}, Monte Carlo}
    \State $K \gets \frac{1}{2}(K + K^{\top})$
    \State \Return $\mathcal{L}_{\mathrm{res}} + \lambda_B\, \mathcal{L}_{\mathrm{orth}} + w_R \operatorname{tr}(G^{-1}K)$ \Comment{Ky Fan selection}
\EndFunction
\State \textbf{Stage 1 — first-order, all $N_s$ states simultaneously:}
\For{$e = 1,\dots,N_1$} \Comment{one gradient step per iteration}
    \State $\mathcal{B} \gets$ $B$ pool points, sampled with probability $\propto$ their vRBA weights \Comment{focuses high-residual regions}
    \State $\mathcal{B}' \gets$ $B'$ pool points, sampled uniformly \Comment{unbiased, for $K$}
    \State $g \gets \nabla_{(\theta, E)}\, \Call{Loss}{\theta, E; \mathcal{B}, \mathcal{B}'}$
    \State $(\theta, E) \gets \mathrm{SOAP}\big((\theta, E),\, g\big)$ \cite{vyas2024soap} \Comment{lr $10^{-3}$, decay $0.9$ per $5{,}000$ steps}
\EndFor
\State \textbf{Stage 2 — SSBroyden polish \cite{jnini2026curvature}:}
\State $\mathbf{w} \gets \mathrm{vec}(\theta, E)$,\; $H \gets I$
\For{$k = 1,\dots,N_2/c$}
    \State draw fresh $\mathcal{B}$ and $\mathcal{B}'$ as in Stage 1
    \State run $c$ SSBroyden iterations with backtracking line search on \Call{Loss}{$\mathbf{w}$; $\mathcal{B}, \mathcal{B}'$}, from $(\mathbf{w}, H)$
    \State $H \gets \frac{1}{2}(H + H^{\top})$; reset $H \gets I$ if not positive definite
    \State keep the lowest-residual $\mathbf{w}$ seen so far
\EndFor
\end{algorithmic}
\end{algorithm}

Solving \eqref{eq:schrodinger} with a network means finding $N_s$ functions $\psi_1,\dots,\psi_{N_s}$ and $N_s$ numbers $E_1,\dots,E_{N_s}$ that meet four requirements at once. (i) Each pair must satisfy the equation, $\mathcal{H}\psi_a = E_a\psi_a$. (ii) Each state must be normalized, $\int \psi_a^{2}\,d\mathbf{x} = 1$. Without this requirement, $\psi \equiv 0$ solves (i) for every $E$. (iii) The states must be mutually orthogonal, $\int \psi_a\psi_b\,d\mathbf{x} = 0$. Without this requirement, all $N_s$ outputs can converge to the same state. (iv) The family must consist of the lowest $N_s$ eigenpairs. Without this requirement, any $N_s$ distinct eigenstates satisfy (i)--(iii). The method assigns one ingredient to each requirement: a residual loss for (i), hard normalization for (ii), a deterministically evaluated orthogonality penalty for (iii), and an energy-minimization term for (iv). The eigenvalues $E_a$ are trainable scalars throughout.

The representation starts from the known decay of the solutions. Bound states of a confining potential decay rapidly at infinity. Each state is therefore written as a product $f_a\,\chi$, where $f_a$ is one output of a single TOUCKAN and $\chi(\mathbf{x}) = (2\pi)^{-D/4}\,e^{-|\mathbf{x}|^{2}/4}$ is a fixed Gaussian envelope. The envelope plays two roles. First, it carries the decay exactly, so no boundary loss and no domain truncation are needed. Second, its square is the standard normal density, $\chi^{2} = \rho$, so the inner product of any two states becomes a Gaussian expectation of a product of network outputs, $\int (f_a\chi)(f_b\chi)\,d\mathbf{x} = \mathbb{E}_{\rho}[f_a f_b]$. This is precisely the quantity that the cross-moment recursion of Section~\ref{sec:variance-recursion} computes in closed form.

This observation makes requirement (ii) free. Collect the expectations into the Gram matrix $G \in \mathbb{R}^{N_s \times N_s}$ with entries $G_{ab} = \mathbb{E}_{\rho}[f_a f_b]$. The diagonal entry $G_{aa}$ is the squared norm of the product $f_a\chi$. Each state is normalized by dividing by it, $\psi_a = f_a\,\chi/\sqrt{G_{aa}}$, which is the hard-normalized form \eqref{eq:tkan-born-density}. The identity $\int \psi_a^{2}\,d\mathbf{x} = 1$ then holds for every value of the network parameters, at every stage of training; normalization never appears in the loss. The construction also removes the trivial solution: the loss below is invariant under rescaling $f_a \mapsto c\,f_a$, so shrinking a state toward zero changes nothing and is never rewarded.

Requirement (i) becomes a residual loss, posed for $f_a$ rather than for $\psi_a$. Substituting $\psi_a = f_a\,\chi/\sqrt{G_{aa}}$ into \eqref{eq:schrodinger} and dividing by $\chi$ gives the residual
\begin{align} \label{eq:app-schrodinger-residual}
r_a
=
-\tfrac{1}{2}\Delta f_a
+ \tfrac{1}{2}\,\mathbf{x}\cdot\nabla f_a
+ W(\mathbf{x})\, f_a
- E_a f_a,
\qquad
W(\mathbf{x})
=
\tfrac{D}{4} - \tfrac{|\mathbf{x}|^{2}}{8} + \tfrac{1}{2}\,\mathbf{x}^{\top} A\,\mathbf{x},
\end{align}
so that $r_a \chi = (\mathcal{H} - E_a)(f_a \chi)$. The derivatives of $\chi$ are absorbed analytically into the drift and potential terms of \eqref{eq:app-schrodinger-residual}. No derivative ever acts on $\chi$ numerically. The collocation set is drawn from $\rho$ itself, and $\mathbb{E}_{\rho}\big[(r_a/\sqrt{G_{aa}})^{2}\big] = \int (\mathcal{H}\psi_a - E_a\psi_a)^{2}\, d\mathbf{x}$. The residual batches, however, are resampled from this set with probabilities proportional to the vRBA weights, and the summands additionally carry those weights (Eq.  \eqref{eq:app-schrodinger-loss}), with no inverse-probability correction. The optimized objective is therefore an adaptively reweighted $\rho$-residual rather than an unbiased estimate of this integral, but the two have the same minimizers, since both vanish exactly when $r_a \equiv 0$. The residual is also the only term in which the trainable eigenvalues appear, so it is what pins each $E_a$ to its state. Monte Carlo enters the formulation only through such Hamiltonian expectations: this residual, and the numerator of the selection term below. Every norm and overlap is deterministic and sampling-free.

The full Gram matrix is deterministic and differentiable at every training step, exact up to the one-dimensional quadrature error of the joint leaf moments. The states share one inner block, so $f_a$ and $f_b$ are polynomials in the same combination variables $\xi_q$, and the product $f_a f_b$ expands in exactly the cross-moments $m_{q,q'}^{(j,j')} = \mathbb{E}[\xi_q^{\,j}\xi_{q'}^{\,j'}]$ of Section~\ref{sec:variance-recursion}. Each Gram matrix element is therefore the bilinear analogue of the second-moment reduction \eqref{eq:tkan-second-moment-reduction}, $G_{ab} = \sum_{q,q'}\sum_{j,j'} \beta^{(a)}_{q,j}\beta^{(b)}_{q',j'}\, m_{q,q'}^{(j,j')}$, where $\beta^{(a)}$ are the monomial-form outer coefficients of state $a$. The cross-moments are computed once per step: joint leaf moments on a $128$-node Gauss--Hermite rule per dimension, then the bivariate binomial recursion \eqref{eq:tkan-cross-recursion} across dimensions. All $N_s^{2}$ matrix elements reuse the same cross-moments and differ only in the coefficient contraction. No sampling is involved at any point, and the result is differentiated through like any other layer of the network.

In summary, the loss combines three terms,
\begin{align} \label{eq:app-schrodinger-loss}
\mathcal{L}
=
\frac{1}{B}\sum_{i=1}^{B}\sum_{a=1}^{N_s}
\Big(\lambda_{i,a}\, \frac{r_a(\mathbf{x}_i)}{\sqrt{G_{aa}}}\Big)^{2}
\;+\;
\lambda_B \sum_{a \neq b} \widehat{G}_{ab}^{\,2}
\;+\;
w_R\, \operatorname{tr}\big(G^{-1}K\big),
\qquad
\widehat{G}_{ab} = \frac{G_{ab}}{\sqrt{G_{aa}G_{bb}}}.
\end{align}
The three terms implement requirements (i), (iii), and (iv) in order; requirement (ii) is built into the representation and needs no term. The first term is the eigen-residual of (i), with per-point vRBA weights $\lambda_{i,a}$ \cite{toscano2026variational}. The second term enforces the orthogonality of (iii). It is computed entirely from the deterministically evaluated Gram: the matrix $G$ delivered by the recursion is normalized to $\widehat{G}$, and the squared off-diagonal entries are summed. No collocation points enter this term, so it carries no sampling noise, and the diagonal of $\widehat{G}$ is one by construction, so it contains no normalization component. Its weight is fixed at $\lambda_B = 100$ in both training stages.

The trace term implements requirement (iv): it measures the total energy of the family, and minimizing it selects the lowest states. The matrix $K$ holds the energy matrix elements of the unnormalized states,
\begin{align} \label{eq:app-schrodinger-K}
K_{ab}
=
\int (f_a\chi)\, \mathcal{H}\,(f_b\chi)\, d\mathbf{x}
=
\mathbb{E}_{\rho}\big[f_a\, \mathcal{H}^{\chi} f_b\big],
\qquad
\mathcal{H}^{\chi} f = r + E f,
\end{align}
where $\mathcal{H}^{\chi} f := \chi^{-1}\,\mathcal{H}(\chi f)$ is the envelope-transformed Hamiltonian. By \eqref{eq:app-schrodinger-residual}, it evaluates as $\mathcal{H}^{\chi} f = r + E f$. Both ingredients, the residual and the state, are already available. For a single state, $K_{aa}/G_{aa}$ is its energy, the Rayleigh quotient. For the family of states, $\operatorname{tr}(G^{-1}K)$ is the natural generalization: it equals the sum of the energies of the $N_s$ states after orthonormalization, and it depends only on the subspace they span. By the Ky Fan principle \cite{fan1949theorem}, its minimum over all $N_s$-dimensional subspaces is the sum of the $N_s$ lowest eigenvalues, attained exactly when the span is the lowest eigenspace. This is requirement (iv). Three properties make the term safe. Copying an already occupied state does not lower the trace, because the trace sees only the span. Collapse of the family onto fewer than $N_s$ directions drives $G$ toward singularity and raises the trace, so collapse acts as a barrier rather than a shortcut. Indeed, no collapse or singular solve was observed in any reported run. 
The identity $\mathcal{H}^{\chi} f_a = r_a + E_a f_a$ cancels the trainable eigenvalues inside the trace, so the term steers only the network functions and leaves the $E_a$ to the residual. The formulation assumes $G \succ 0$, in particular $G_{aa} > 0$ in the normalization; this held in every reported run. At the trained parameters, the Gram matrix is well conditioned in every dimension, with smallest eigenvalue bounded away from zero and condition number of order $10^{1}$--$10^{2}$. The off-diagonal Gram residual that training monitors falls below $10^{-5}$ in every case (Figure~\ref{fig:schrodinger-convergence}, bottom row), so the family stays near-orthogonal and far from collapse. Table~\ref{tab:gram-conditioning} reports these quantities. The term is implemented as a direct dense solve, $\operatorname{tr}\big(G^{-1}K\big) = \operatorname{tr}\big(\mathrm{solve}(G, K)\big)$, with no added jitter or regularization. The numerator $K$ is the second and last Monte Carlo estimate, computed on its own uniformly drawn batch and symmetrized. This estimate is unbiased for the exact energy couplings under the stated uniform sampling. The denominator $G$ is deterministic, exact up to the one-dimensional quadrature error. Because the exact (expected) trace is stationary precisely at eigenspaces, the trace term introduces no bias at convergence for the exact objective and is never annealed. The weight for the trace term is $w_R = 0.1$ throughout.

\begin{table}[H]
\centering
\caption{Conditioning of the Gram matrix $G$ at the trained parameters for the Schr\"odinger eigenstate runs ($N_s = 4$ states, evaluated by the exact recursion at the production order $M = 128$). The condition number $\kappa(G) = \lambda_{\max}/\lambda_{\min}$ stays at order $10^{1}$--$10^{2}$ with $\lambda_{\min}(G) > 0$ at the trained parameters, so the direct solve $\mathrm{solve}(G, K)$ is stable without regularization. Successful completion of every run indicates no observed singular solve, but these final-parameter values do not quantify the conditioning during optimization. The last column reports the off-diagonal Gram residual at the final training iteration (Figure~\ref{fig:schrodinger-convergence}).}
\label{tab:gram-conditioning}
\begin{tabular}{ccccc}
\hline
$D$ & $\lambda_{\min}(G)$ & $\lambda_{\max}(G)$ & $\kappa(G)$ & off-diag.\ residual \\
\hline
2 & $3.97\times10^{-4}$ & $8.96\times10^{-2}$ & $2.26\times10^{2}$ & $2.73\times10^{-9}$ \\
3 & $4.58\times10^{-5}$ & $1.42\times10^{-3}$ & $3.11\times10^{1}$ & $4.88\times10^{-8}$ \\
5 & $6.99\times10^{-9}$ & $2.23\times10^{-7}$ & $3.20\times10^{1}$ & $5.22\times10^{-6}$ \\
\hline
\end{tabular}
\end{table}

One network represents all $N_s$ states. The inner block is the intermediate sharing form of Section~\ref{sec:tractable-architecture}: a shared random Fourier embedding ($15$ frequencies, bandwidth $\sigma = 1$, matched to the unbounded Gaussian inputs), a shared trunk of four hidden layers of width $32$ with $\tanh$ activations, and one full linear head per input dimension producing $Q = 64$ branch channels. Each branch ends in a $\tanh$ and the channels are averaged over dimensions, so every combination variable lies in $(-1,1)$. The outer block evaluates per-branch Chebyshev expansions of degree $P = 9$ with $N_s$ coefficient columns: the states share the same combination variables and differ only through their outer coefficients. Together with the $N_s$ trainable eigenvalues, this gives the parameter counts of Table~\ref{tab:schrodinger}.

Training is simultaneous and assumes nothing about the spectrum. All $N_s$ states are active from the first iteration; there is no sequential scan over eigenvalues and no scheduled sweep of trial energies. The $E_a$ are initialized at small distinct values below the bottom of the spectrum, for which a lower bound on the potential suffices, and rise to their converged values by tracking the Rayleigh quotients of their states through the residual term. The optimization follows the two-stage template of the other experiments, in double precision on $N_C = 200{,}000$ collocation points drawn by mapping a Latin hypercube sample through the standard normal quantile function, with batch size $20{,}000$. The full procedure is summarized in Algorithm~\ref{alg:schrodinger-training}. Stage one runs $20{,}000$ first-order iterations (SOAP \cite{vyas2024soap}; initial learning rate $10^{-3}$, decayed by $0.9$ every $5{,}000$ steps, outer coefficients at $20\%$ of the base rate), with vRBA weights (exponential potential, $\eta = 0.01$, $\gamma = 0.999$) that also define the batch sampling distribution. Stage two runs $20{,}000$ SSBroyden iterations \cite{jnini2026curvature} with backtracking line search, in chunks of $100$ with a fresh vRBA-sampled batch per chunk and the inverse-Hessian approximation carried across chunks. The parameters with the lowest mean residual over the second stage are kept. As in the other experiments, all reported results are single-run observations at the stated seed.

The errors of Table~\ref{tab:schrodinger} follow the definitions of Eq. \eqref{eq:rel-error} and Eq. \eqref{eq:rel-l2}, with the exact spectrum and eigenfunctions as references, the sign of each $\psi_a$ aligned to the reference (eigenfunctions are defined only up to sign), and the $L_2$ norms estimated on an independent set of $20{,}000$ test points drawn from $\rho$, separate from the collocation set used for training.

The formulation inherits the problem statement of Jin et al.\ \cite{jin2022physics} but replaces its constraint machinery. We retain a residual loss with trainable eigenvalues, normalization and orthogonality requirements, and decay at infinity. In their MLP-based method, normalization and orthogonality are evaluated through collocation sums, and states are obtained sequentially through a sampled orthogonality penalty with patience-based solution recognition and symmetry switching. The trial-energy drive inherited from their earlier formulation is disabled in the reported experiments, and infinity is replaced by a truncated domain. These mechanisms introduce sampling or truncation error, and the reported experiments solve one-dimensional equations. TOUCKAN instead handles the integral constraints without sampling: normalization is imposed by division, and orthogonality is enforced through the deterministic Gram matrix. State selection uses a single variational trace term, while the envelope carries the decay over all of $\mathbb{R}^{D}$. This formulation can therefore be trained on non-separable multivariate problems while learning all $N_s$ states simultaneously.

Figures~\ref{fig:schrodinger-states-2d} and \ref{fig:schrodinger-states-5d} complement Figure~\ref{fig:schrodinger-states} of the main text with the learned eigenstates of the $D = 2$ and $D = 5$ runs.

\begin{figure}[H]
    \centering
    \includegraphics[width=0.99\linewidth]{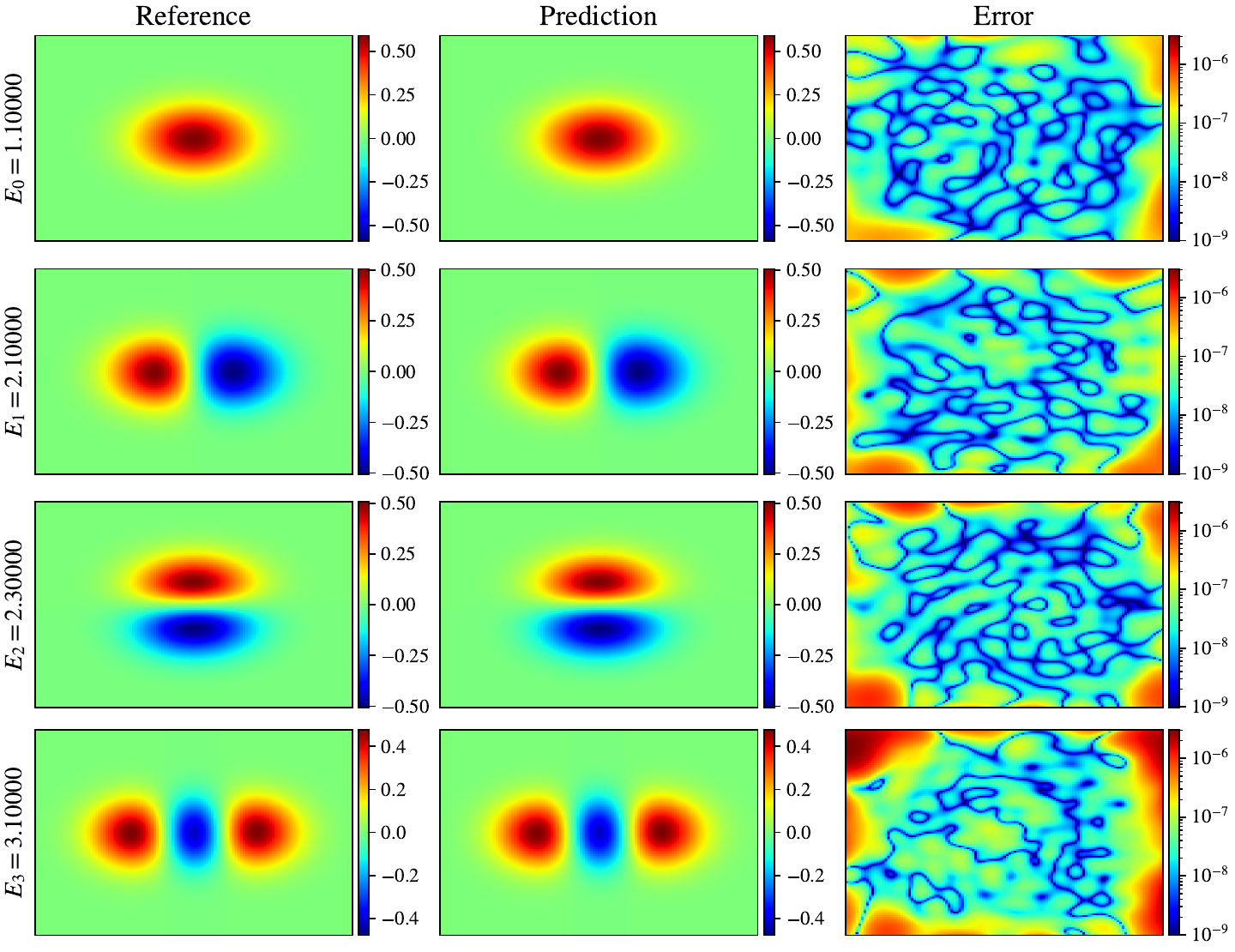}
    \caption{Lowest four eigenstates of the coupled two-dimensional oscillator: closed-form reference (left), learned surrogate (center), and absolute error on a logarithmic color scale (right). Each row is labeled by its learned eigenvalue $E_a$.}
    \label{fig:schrodinger-states-2d}
\end{figure}

\begin{figure}[H]
    \centering
    \includegraphics[width=0.99\linewidth]{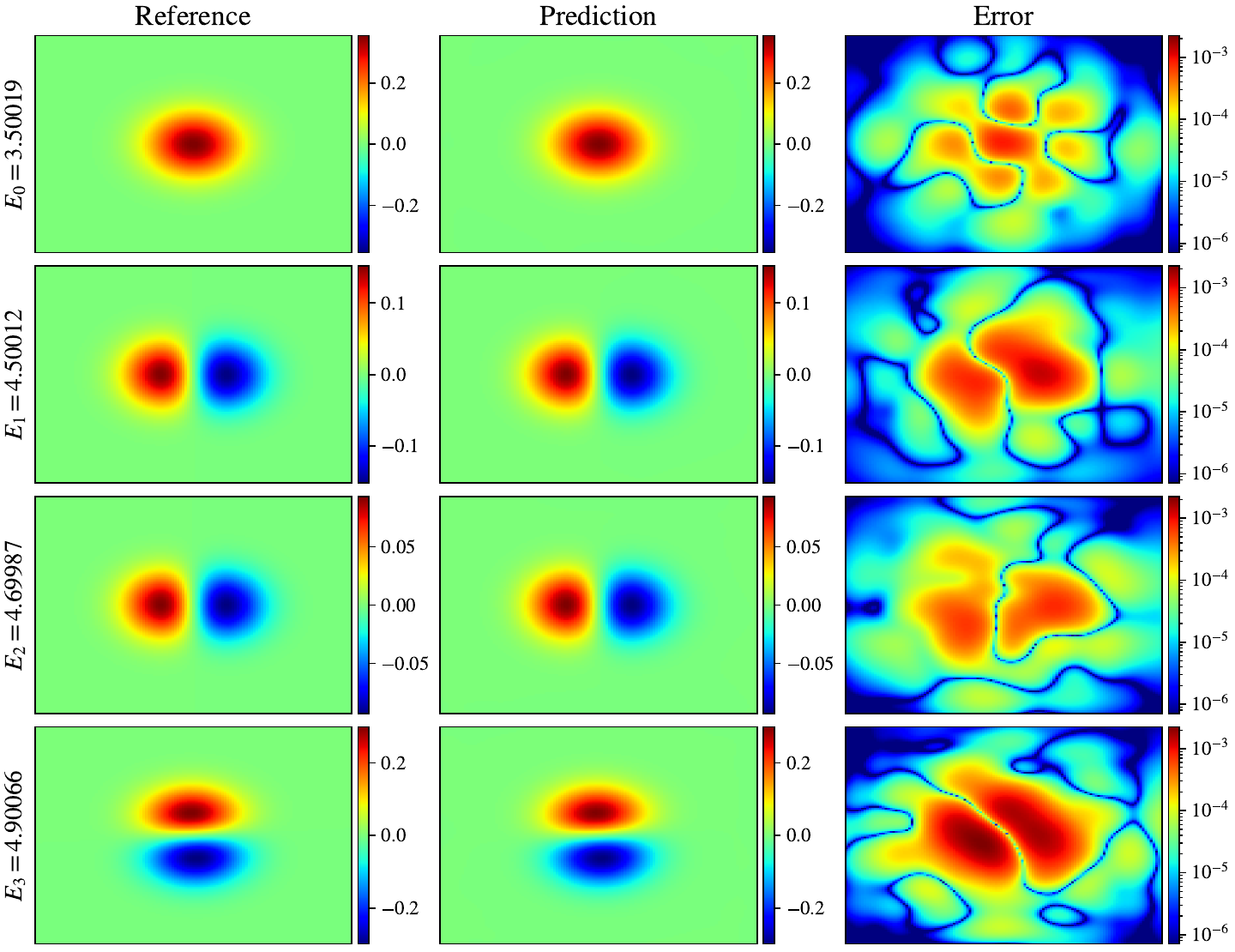}
    \caption{Lowest four eigenstates of the coupled five-dimensional oscillator on the $x_3 = x_4 = x_5 = 0$ plane: closed-form reference (left), learned surrogate (center), and absolute error on a logarithmic color scale (right). Each row is labeled by its learned eigenvalue $E_a$. On this slice, the $E_1$ and $E_2$ states have similar projected profiles, while the Gram residuals in Table~\ref{tab:schrodinger} quantify their full-dimensional distinction.}
    \label{fig:schrodinger-states-5d}
\end{figure}

\end{document}